\documentclass{article}

\usepackage{fullpage,amsmath,amsthm,amssymb,enumerate,enumitem,mathtools,float,accents,url,enumerate, tikz, subcaption,comment}

\usepackage[round]{natbib}
\usepackage[colorlinks = true, citecolor = blue,
linkcolor = blue]{hyperref}
\usepackage{cleveref}
\usepackage{parskip}
\usepackage{setspace}

\newtheorem{theorem}{Theorem}[section]
\newtheorem*{theorem*}{Theorem}
\newtheorem{lemma}[theorem]{Lemma}
\newtheorem*{lemma*}{Lemma}
\newtheorem{proposition}[theorem]{Proposition}

\newtheorem{example}{Example}
\newtheorem{remark}{Remark}

\newtheorem*{definition*}{Definition}

\newtheorem{innercustomgeneric}{\customgenericname}
\providecommand{\customgenericname}{}
\newcommand{\newcustomtheorem}[2]{%
  \newenvironment{#1}[1]
  {%
   \renewcommand\customgenericname{#2}%
   \renewcommand\theinnercustomgeneric{##1}%
   \innercustomgeneric
  }
  {\endinnercustomgeneric}
}

\newcustomtheorem{customexample}{Example}

\newcommand{\overbar}[1]{\mkern1.5mu\overline{\mkern-1.5mu#1\mkern-1.5mu}\mkern 1.5mu}

\usepackage{xcolor}

\DeclarePairedDelimiter{\tripnorm}{\lvert\kern-0.25ex\lvert\kern-0.25ex\lvert}{\rvert\kern-0.25ex\rvert\kern-0.25ex\rvert}
\DeclarePairedDelimiterX{\helperkldiv}[2]{(}{)}{#1\;\delimsize\|\;#2}

\newcommand{\cross}{\times}

\newlength{\dhatheight}

\newcommand{\diag}{\textnormal{diag}}

\newcommand\indep{\protect\mathpalette{\protect\independenT}{\perp}}
\def\independenT#1#2{\mathrel{\rlap{$#1#2$}\mkern2mu{#1#2}}}

\title{
Augmented balancing weights as linear regression\thanks{We would like to thank David Arbour, Eli Ben-Michael, Andreas Buja, Alex D'Amour, Skip Hirshberg, Guido Imbens, Apoorva Lal, Mark van der Laan, Whitney Newey, Rahul Singh, Jann Spiess, and Qingyuan Zhao for useful discussion and comments. A.F. and D.B-S. were supported in part by the Institute of Education Sciences, U.S. Department of Education, through Grant R305D200010. The opinions expressed are those of the authors and do not represent the views of the Institute or the U.S. Department of Education. O.D. was supported by NIH grant  579679 and by the FWO grant 1222522N. E.L.O. was supported by ONR grant N000142112820 and by the Simons Institute for Theoretical Computer Science.}}
\author{David Bruns-Smith \\ UC Berkeley \and Oliver Dukes \\ Ghent Univ. \and 
Avi Feller \\ UC Berkeley \and 
Elizabeth L. Ogburn \\ Johns Hopkins Univ.}

\date{\today}

\begin{document}

\maketitle

\begin{abstract}
We provide a novel characterization of augmented balancing weights, also known as automatic debiased machine learning (AutoDML). These popular \emph{doubly robust} or \emph{de-biased machine learning} estimators combine outcome modeling with balancing weights --- weights that achieve covariate balance directly in lieu of estimating and inverting the propensity score. When the outcome and weighting models are both linear in some (possibly infinite) basis, we show that the augmented estimator is equivalent to a single linear model with coefficients that combine the coefficients from the original outcome model and coefficients from an unpenalized ordinary least squares (OLS) fit on the same data. We see that, under certain choices of regularization parameters, the augmented estimator often collapses to the OLS estimator alone; this occurs for example in a re-analysis of the \citet{lalonde1986evaluating} dataset. We then extend these results to specific choices of outcome and weighting models. We first show that the augmented estimator that uses (kernel) ridge regression for both outcome and weighting models is equivalent to a single, undersmoothed (kernel) ridge regression. This holds numerically in finite samples and lays the groundwork for a novel analysis of undersmoothing and asymptotic rates of convergence. When the weighting model is instead lasso-penalized regression, we give closed-form expressions for special cases and demonstrate a ``double selection'' property. Our framework opens the black box on this increasingly popular class of estimators, bridges the gap between existing results on the semiparametric efficiency of undersmoothed and doubly robust estimators, and provides new insights into the performance of augmented balancing weights.
\end{abstract}

\clearpage
\pagenumbering{arabic}
\section{Introduction}

Combining outcome modeling and weighting, as in augmented inverse propensity score weighting (AIPW) and other doubly robust (DR) or double machine learning (DML) estimators, is a core strategy for estimating causal effects using observational data. A growing body of literature finds weights by solving a ``balancing weights'' optimization problem to estimate weights directly, rather than by first estimating the propensity score and then inverting. DR versions of these estimators are referred to by a number of terms, including \emph{augmented balancing weights} \citep{athey2018approximate, hirshberg2021augmented}, \emph{automatic debiased machine learning} \citep[AutoDML;][]{chernozhukov2022automatic}, and \emph{generalized regression estimators} \citep[GREG;][]{deville1992calibration}; see \citet{ben2021balancing} for a review. 
Moreover, this strategy has been applied to a wide range of linear estimands via the Riesz representation theorem \citep[e.g.,][]{hirshberg2021augmented, chernozhukov2022riesz}.
In this paper, we consider augmented balancing weights in which the estimators for both the outcome model and the balancing weights are based on penalized linear regressions in some possibly infinite basis; in addition to all high-dimensional linear models, this broad class includes popular nonparametric models such as kernel regression and certain forms of random forests and neural networks. 

We first show that, somewhat surprisingly, augmenting any regularized linear outcome regression (the ``base learner") with linear balancing weights is numerically equivalent to a single linear outcome regression applied to the target covariate profile.
The resulting coefficients are an affine (and often convex) combination of the base learner model coefficients and unregularized OLS coefficients; 
the hyperparameter for the balancing weights estimator directly controls the regularization path defining the affine combination. In the extreme case where the weighting hyperparameter is set to zero --- which we show can easily occur in practice --- the entire procedure is equivalent to estimating a single, unregularized OLS regression.

We specialize these results to ridge and lasso regularization ($\ell_2$ and $\ell_\infty$ balancing, respectively) and show that augmenting an outcome regression estimator with balancing weights generally corresponds to a form of \emph{undersmoothing}. 
Most notably, we show that an augmented balancing weight estimator that use (kernel) ridge regression for both outcome and weighting models --- which we refer to as ``double ridge'' --- collapses to a single, undersmoothed (kernel) ridge regression estimator. 

We leverage these results to prove novel \emph{statistical} results for double ridge estimators and to make progress towards practical hyperparameter tuning, which remains an open problem in this area.
We first make explicit the connection between asymptotic results for double kernel ridge estimators \citep[e.g.,][]{singh2020kernel} and prior results on optimal undersmoothing for a single kernel ridge outcome model \citep[e.g.,][]{mou2023kernel}, showing that the latter is also semiparametrically efficient. 
This generalizes the argument in \citet{robins2007comment} that ``OLS is doubly robust'' to a much broader class of penalized parametric and non-parametric regression estimators. 
As a complementary analysis, we next adapt existing finite sample error analysis results for single ridge regression \citep{dobriban2018high} to derive the finite-sample-exact bias and variance of double ridge estimators. Using these expressions, we can compute oracle hyperparameters for any given data-generating process.

Finally, we illustrate our results with several numerical examples. We first explore hyperparameter tuning for double ridge regression in an extensive simulation study on 36 data-generating processes, and compare three practical methods to the optimal hyperparameter computed using our finite sample analysis. Surprisingly, asymptotic theory and our simulation results suggest equating the hyperparameters for the outcome and weighting models. We  caution against the naive application of hyperparameter tuning based solely on cross-validating the weighting model, forms of which have been suggested previously.
This approach can lead to setting the weighting hyperparameter to exactly zero --- and therefore recovering standard OLS --- even in scenarios where OLS is far from optimal. 
We emphasize this point by applying our results to the canonical \citet{lalonde1986evaluating} study, highlighting that researchers can inadvertently recover OLS in practice. 

Broadly, our results provide important insights into the nexus of causal inference and machine learning. First, these results open the black box on the growing number of methods based on augmented balancing weights and AutoDML --- methods that can sometimes be difficult to taxonomize or understand. We show that, under linearity, these estimators all share an underlying and very simple structure. Our results further highlight that estimation choices for augmented balancing weights can lead to potentially unexpected behavior. 
At a high level, as causal inference moves towards incorporating machine learning and automation, our work highlights how the traditional lines between weighting and regression-based approaches are becoming increasingly blurred.

Second, our results connect two approaches to ``automate'' semiparametric causal inference. AutoDML and related methods exploit the fact that we can estimate a Riesz representer without a closed form expression for a wide class of functionals. The estimated Riesz representer then augments a base learner by bias correcting a plug-in estimator of the functional. Older approaches, such as undersmoothing \citep{goldstein1992optimal,newey1998undersmoothing}, twicing kernels \citep{newey2004twicing}, and sieve estimation \citep{newey1994asymptotic,shen1997methods},  avoid estimation of the Riesz representer, tuning the base learner regression fit such that an additional bias correction is not required. Achieving this optimal tuning in practice has long been a hurdle for the implementation of these methods. Subject to certain conditions, both approaches can yield estimators that are asymptotically efficient. We show that if all required tuning parameters are defined in terms of an $\ell_2$-norm constraint, then these approaches can be numerically identical even in finite samples. We use these equivalences to make progress toward practical hyperparameter selection and find promising directions for new theoretical analysis.  

In Section \ref{sec:background} we introduce the problem setup, identification assumptions, and common estimation methods; we also review balancing weights and previous results linking balancing weights to outcome regression models. In Section \ref{sec:newresults} we present our new numerical results, and in Sections \ref{sec:L2} and \ref{sec:l8_aug_section} we cache out the implications for $\ell_2$ and $\ell_\infty$ balancing weights specifically. 
Building on our numerical results, Section \ref{sec:kernel_ridge_theory} explores both asymptotic and finite sample statistical results for kernel ridge regression.
Section \ref{sec:numerical_illustrations} illustrates our results with a simulation study and application to canonical data sets.
Section \ref{sec:discussion} offers some other directions for future research.
The appendix includes extensive additional technical discussion and extensions.

\subsection{Related work}
\label{sec:lit_review}

\paragraph{Balancing weights and AutoDML.} With deep roots in survey calibration methods and the \emph{generalized regression estimator} \citep[GREG; see][]{deville1992calibration,lumley2011connections,gao2022soft},
a large and growing causal inference literature uses balancing weights estimation in place of traditional inverse propensity score weighting (IPW). \citet{ben2021balancing} provide a recent review; we discuss specific examples at length in Section \ref{linearbalancingproblems} below. This approach typically balances features of the covariate distributions in the different treatment groups, with the aim of minimising the maximal design-conditional mean squared error of the treatment effect estimator. Of particular interest here are augmented balancing weights estimators that combine balancing weights with outcome regression; see, for example, \citet{athey2018approximate, hirshberg2021augmented, ben2021augmented}. 

A parallel literature in econometrics instead focuses on so-called \emph{automatic} 
estimation of the Riesz representer, of which IPW are a special case, where ``automatic" refers to the fact that we can estimate the Riesz representer without obtaining a closed form expression. Estimating the Riesz representer directly, under the assumption that it is linear in some basis, dates back at least to \cite{robins2008higher}; see also \cite{robins2007comment}. The corresponding augmented estimation framework has more recently come to be known as Automatic Debiased Machine Learning, or AutoDML; see, among others, \citet{chernozhukov2022locally}, \citet{chernozhukov2022riesznet}, \citet{chernozhukov2022automatic}, and \citet{chernozhukov2022riesz}. 
This approach has also been applied in a range of settings, including to corrupted data \citep{agarwal2021causal}, to dynamic treatment regimes \citep{chernozhukov2022dynamic}, and to address noncompliance \citep{singh2022late}. 
As we discuss in Appendix \ref{apx:cross-fit-unreg-weight}, the AutoDML approach nearly always employs cross-fitting and is typically motivated by asymptotic properties rather than achieving minimax design-conditional mean squared error.

\paragraph{Numerical equivalences for balancing weights.} Many seminal papers highlight connections between weighting approaches, such as balancing weights and IPW, and outcome modeling; see \citet{bruns2022outcome} for discussion. Most relevant are a series of papers that show numerical equivalences between linear regression and (exact) balancing weights, especially \citet{robins2007comment,kline2011oaxaca,chattopadhyay2021implied}, and between kernel ridge regression and forms of kernel weighting \citep{kallus2020generalized, hirshberg2019minimax}. We discuss these equivalences at length in \Cref{apx:equiv-conditions}. Finally, as we discuss in \Cref{sec:nonlinear_weights}, there are close connections between balancing weights and Empirical Likelihood \citep{hellerstein1999el, newey2004el}.

\section{Problem setup and background}\label{sec:background}

\subsection{Setup and motivation}

The core results in our paper are numeric equivalences for existing estimation procedures, and as such these results hold absent any causal assumptions or statistical model. Nonetheless, a primary motivation for this work is the task of estimating unobserved counterfactual means in causal inference, as well as estimating the broad class of linear functionals described in \cite{chernozhukov2018learning}.
We briefly review the corresponding setup, emphasizing that this is purely for interpretation.

\subsubsection{Example: Estimating counterfactual means}
Let $X, Y, Z$ be random variables defined on $\mathcal{X} , \mathbb{R}, \mathcal{Z}$ with joint probability distribution $p$. To begin, consider the example of a binary treatment, $\mathcal{Z} = \{0,1\}$ and covariates $X$. Define potential or counterfactual outcomes $Y(1)$ and $Y(0)$ under assignment to treatment and control, respectively. Under SUTVA \citep{rubin1980randomization}, we observe outcomes $Y = Z Y(1) + (1 - Z)Y(0)$. To estimate the average treatment effect, $\mathbb{E}[ Y(1) - Y(0) ]$, we first estimate the means of the partially observed potential outcomes. We initially focus on estimating $\mathbb{E}[ Y(1) ]$; a symmetric argument holds for $\mathbb{E}[Y(0)]$.

Let $m(x,z) \coloneqq \mathbb{E}[Y \mid X = x, Z = z]$ be the \emph{outcome model}, $e(x) \coloneqq \mathbb{P}[Z = 1 \mid X = x]$ be the \emph{propensity score}, and $\alpha(x,z) = z / e(x)$ be the \emph{inverse propensity score weights} (IPW). Under the additional assumptions of \emph{conditional ignorability}, $Y(1) \indep Z \mid X$, and \emph{overlap}, $\mathbb{E}[\alpha(X,Z)^2] < \infty$, we have that $\mathbb{E}[ Y(1) ]$ is identified by $ \mathbb{E}[ m(X,1) ]$, a linear functional of the observed data distribution. 

There are three broad strategies for estimating $\mathbb{E}[ Y(1) ]$. First, the identifying functional above suggests estimating the outcome model, $m(x, 1)$ among those units with $Z=1$, and plugging this into the \emph{regression functional}, $\mathbb{E}[ m(X,1) ]$. Second, the equality $\mathbb{E}[ m(X,1) ] = \mathbb{E}[Z/e(X) Y] = \mathbb{E}[\alpha(X,Z) Y]$ suggests estimating the inverse propensity score weights, $\alpha(x,z) = z/e(x)$, and plugging these into the \emph{weighting functional}. Finally, we can combine these two via the \emph{doubly robust functional} \citep{robins1994estimation}:
$$
\mathbb{E}[m(X,1) + \alpha(X,Z)(Y - m(X,1)].
$$
This functional has the attractive property of being equal to $\mathbb{E}[ m(X,1) ]$ even if either one of $\alpha$ or $m$ is replaced with an arbitrary function of $X$ and $Z$, hence the term ``doubly robust.'' Doubly robust estimators have been studied extensively in semiparametric theory; note that $m(X,1) + \alpha(X,Z) (Y - m(X,Z))-\psi(m)$ coincides with the efficient influence function for $\psi(m)$ under a nonparametric model (see \citealp{kennedy2022semiparametric} for a review of the relevant theory).See \citet{chernozhukov2018double, kennedy2022semiparametric} for recent overviews of the active literature in causal inference and machine learning focused on estimating versions of this functional.

\subsubsection{General class of functionals via the Riesz representer}
Our results apply well beyond the example above. In particular, they apply to any functional of the form 
\begin{align}
    \psi(m) = \mathbb{E}[h(X_i, Z_i, m)],\label{eq:general_linear_functional}
\end{align}
where $\mathcal{Z}$ is an arbitrary set; $Z$ a random variable with support $\mathcal{Z}$; and $h$ is a real-valued, mean-squared continuous linear functional of $m$ \citep{chernozhukov2018learning,hirshberg2021augmented, chernozhukov2022automatic}. 
 Following \cite{chernozhukov2022automatic, chernozhukov2022riesz}, we can generalize the weighting functional to this general class of estimands via the \emph{Riesz representer}, which is a function $\alpha(X,Z) \in L_2(p)$ such that, for all square-integrable functions $f \in L_2(p)$:
\begin{align}
\mathbb{E}[h(X, Z, f)] = \mathbb{E} [ \alpha(X,Z) f(X,Z) ].\label{eq:riesz-estimand}  
\end{align} 

As in the counterfactual mean example, we can identify the more general target functional in \eqref{eq:riesz-estimand} via the outcome regression functional in \eqref{eq:general_linear_functional}, via the Riesz representer functional in \eqref{eq:riesz-estimand} with $f=m$, or via the doubly robust functional 
\begin{align}
    \mathbb{E}[h(X, Z, m) + \alpha(X,Z) (Y - m(X,Z))]. \label{eq:doubly_robust_functional}
\end{align} 
Estimators of this DR functional are \emph{augmented} in the sense that they augment the ``plug-in,'' ``outcome regression,'' or ``base learner'' estimator of $\mathbb{E}[h(X, Z, m)]$ with appropriately weighted residuals; or, equivalently, that augment the weighting estimator with an appropriate outcome regression. This is the class of estimators to which our results apply. Doubly robust estimators have been studied extensively in semiparametric theory. In particular, $h(X, Z, m) + \alpha(X,Z) (Y - m(X,Z))-\psi(m)$ coincides with the efficient influence function for $\psi(m)$ under a nonparametric model (see \citealp{kennedy2022semiparametric} for a review of the relevant theory).
In future work we will explore whether we can extend our results to a different class of functionals that admit DR functional forms, first introduced by \cite{robins2008higher}, and to the superset of such functionals characterized by  \cite{rotnitzky2021characterization}.


\subsection{Balancing weights: Background and general form} 

The core idea behind balancing weights is to estimate the Riesz representer directly ---  rather than via an analytic functional form (e.g., by estimating the propensity score and inverting it). 
As a result, balancing weights do not require a known analytic form for the Riesz representer \citep{chernozhukov2022riesz}, are often much more stable \citep{zubizarreta2015stable}, and offer improved control of finite sample covariate imbalance \citep{zhao2019covariate}. We briefly describe two primary motivations for this approach.

First, a central property of the Riesz representer is that the corresponding weights, $w(X, Z) = \alpha(X,Z)$, are the unique weights that satisfy the \emph{population balance property} property in \Cref{eq:riesz-estimand} for all square-integrable functions $f \in L_2(p)$.
For our target estimand $\psi(m)$ we only need to satisfy the condition in \Cref{eq:riesz-estimand} for the special case of $f=m$. If we are willing to assume that $m$ lies in a model class $\mathcal{F} \subset L_2(p)$, then it suffices to balance functions in that class. This is achieved by minimizing the imbalance over $\mathcal{F}$:
\begin{align}
 \text{Imbalance}_\mathcal{F}(w) \coloneqq \sup_{f \in \mathcal{F}} \Big\{ \mathbb{E}[ w(X,Z) f(X,Z)] - \mathbb{E}[h(X,Z,f)] \Big\}.
 \label{eq:imbalanceF} 
\end{align}
As we discuss next, balancing weights minimize a (penalized) sample analog of Equation \eqref{eq:imbalanceF}.

Alternatively, \citet{chernozhukov2022automatic} consider finding weights $f$ that minimize the mean-squared error for $\alpha(X,Z)$:
\begin{align}
    \min_{f \in \mathcal{F}} \left\{ \mathbb{E}\left[\left(f(X,Z) - \alpha(X,Z)\right)^2\right] \right\} . \label{eq:generaldual}
\end{align}
Automatic estimation of the Riesz representer, also known as \emph{Riesz regression} \citep{chernozhukov2024riesz_regression}, minimizes a sample analog of Equation \eqref{eq:generaldual}. 
When $\mathcal{F}$ is convex, then up to choice of hyperparameters (see \eqref{eq:autoform} below), the solutions to Equations \eqref{eq:imbalanceF} and \eqref{eq:generaldual} are equivalent.

\subsection{Linear balancing weights} \label{linearbalancingproblems}

In this paper, we consider the special case in which the outcome models are linear in some basis expansion of $X$ and $Z$. 
This is an extremely broad class that encompasses linear and polynomial models of  arbitrary functions of $X$ and $Z$ and with dimension possibly larger than the sample size, as well as non-parametric models such as reproducing kernel Hilbert spaces \citep[RKHSs;][]{gretton2012kernel}, the Highly-Adaptive Lasso \citep{benkeser2016hal}, the neural tangent kernel space of infinite-width neural networks \citep{jacot2018neural}, and ``honest'' random forests \citep{agarwal2022randomforests}.
However, this class excludes models for $m$ that are fundamentally non-linear in their parameters, like general neural networks or generalized linear models with a non-linear link function. We extend our results to arbitrary nonlinear balancing weights in Appendix \ref{sec:nonlinear_weights}.

Under linearity, the imbalance over all $f \in \mathcal{F}$ has a simple closed form. 
Because our results concern numeric equivalences, we will focus on the finite sample version of the linear balancing weights problem. 
Let $\mathcal{F} = \{ f(x,z) = \theta^\top\phi(x,z) : \Vert \theta \Vert \leq 1 \}$ where $\Vert \cdot \Vert$ can be any norm on $\mathbb{R}^d$. The general setup constrains $\Vert \theta \Vert \leq r$; we set $r = 1$ without loss of generality, which simplifies exposition below. Let $\Vert \cdot \Vert_*$ be the \emph{dual norm} of $\Vert \cdot \Vert$; that is, $ \Vert v \Vert_* \coloneqq \sup_{\Vert u \Vert \leq 1} u^\top v.$ Many common vector norms have familiar, closed-form, dual norms, e.g., the dual norm of the $\ell_2$-norm is the $\ell_2$-norm; and the dual norm of the $\ell_1$-norm is the $\ell_\infty$-norm. Let $X_p, Y_p, Z_p$ be $n$ i.i.d. samples from the distribution $p$ of the observed data.  Define the feature map $\phi : \mathcal{X} \times \mathcal{Z} \rightarrow \mathbb{R}^d$ and let $\phi_j  : \mathcal{X} \times \mathcal{Z} \rightarrow \mathbb{R}$ denote the mapping for the $j$th feature. Define $\Phi_p \coloneqq \phi(X_p, Z_p)$ and  let $\Phi_q \coloneqq h(X_p, Z_p, \phi)$ denote the \emph{target features}.
We will write $\hat{\mathbb{E}}$ for sample averages; define $\overbar{\Phi}_p \coloneqq \hat{\mathbb{E}}[ \Phi_p ]$ and $\overbar{\Phi}_q \coloneqq \hat{\mathbb{E}}[ \Phi_q ]$. 
For exposition, we assume that $d < n$ and that $\Phi_p$ has rank $d$. We emphasize that this is not necessary for our results --- one can replace $\mathbb{R}^d$ with an infinite-dimensional Hilbert space $\mathcal{H}$ and relax the rank restriction. See \Cref{sec:high-dimensions} for a formal presentation of the high-dimensional ($d>n$) setting.

In what follows we write $w$ for the $1 \cross n$ vector $w(\Phi_p)$, to highlight the fact that we will estimate $w$ directly rather than as an explicit function of $X$ or $\Phi_p$. Using the derivation above, we can directly calculate the finite sample imbalance as:
$$ \widehat{\text{Imbalance}_{\mathcal{F}}}(w) = \Vert \tfrac{1}{n} w \Phi_p - \bar{\Phi}_q \Vert_*.$$

Now we can write the penalized sample analog of balancing weights optimization problem in \eqref{eq:imbalanceF} equivalently as either:
\begin{align*}
\text{Penalized form:} \qquad & 
    \min_{w\in\mathbb{R}^n} \Big\{ \Vert  \tfrac{1}{n} w \Phi_p - \bar{\Phi}_q \Vert^2_* + \delta_1 \Vert w\Vert_2^2 \Big\}\\[1em]
\text{Constrained form:} \qquad &  \min_{w\in\mathbb{R}^n} \Vert w \Vert_2^2 \\
    &   \text{such that } \Vert \tfrac{1}{n} w \Phi_p - \bar{\Phi}_q \Vert_* \leq \delta_2 .
\end{align*}
Furthermore, we can write the equivalent problem in 
\eqref{eq:generaldual} as:
\begin{align}
    \text{Riesz regression form:} \qquad &     \min_{\theta \in \mathbb{R}^d} \Big\{ \tfrac{1}{n} \theta^\top( \Phi_p^\top\Phi_p)\theta - \tfrac{1}{n}2\theta^\top\bar{\Phi}_q + \delta_3 \Vert \theta \Vert  \Big\},\label{eq:autoform}
\end{align}

where we use the terminology ``Riesz regression'' from \citet{chernozhukov2024riesz_regression}.
For any parameter $\delta_2 > 0$ and corresponding constrained problem solution $\hat{w}$, there exists a parameter $\delta_3 > 0$ such that $\hat{w} = \delta_3 \Phi_p \hat{\theta}$, where $\hat\theta$ is the solution to the Riesz regression form. 
As a result, for any norm $\Vert \cdot \Vert$, the penalized and constrained forms will always produce weights that are linear in $\Phi_p$ \citep[see][Section 9]{ben2021balancing}. Therefore, since the problems are equivalent, we typically use a generic $\delta$ to denote the regularization parameter, and will specify the particular form only if necessary. In Appendix \ref{apx:balancing-examples} we illustrate several concrete examples for this problem and in Appendix \ref{sec:nonlinear_weights} we consider alternative dispersion parameters and discuss popular forms of balancing that constrain the weights to be non-negative.

\begin{remark}[Intercept]
An important constraint in practice is to normalize the weights, $\frac{1}{n} \sum_{i=1}^n w_i = 1$. This corresponds to replacing $\Phi_p$ and $\Phi_q$ with their centered forms, $\Phi_p - \bar{\Phi}_p$ and $\Phi_q - \bar{\Phi}_p$, in the dual form of the balancing weights problem. This is also equivalent to adding a column of $1$s to $\Phi_p$. Appropriately accounting for this normalization, however, unnecessarily complicates the notation. Therefore, without loss of generality, we will assume that the features are centered throughout, that is, $\bar{\Phi}_p = 0$. 
\end{remark}

\begin{remark}[Equivalence with kernel ridge regression] \label{remark:ridge_equiv}
For the special case of $\ell_2$ balancing (as in Appendix \ref{apx:balancing-examples}) the balancing weights problem is numerically equivalent to directly estimating the conditional expectation $\mathbb{E}[Y_p | \Phi_p]$ via (kernel) ridge regression and applying the estimated coefficients to $\overbar{\Phi}_q$. Moreover, the solution to the balancing weights problem has a closed form that is always linear in $\overbar{\Phi}_q$; we provide further details in \Cref{apx:equiv-conditions}. 
For exact balance with $\delta = 0$, the balancing weights problem is equivalent to fitting unregularized OLS; see, for example, 
 \citet{robins2007comment}, \citet{kline2011oaxaca}, and \citet{chattopadhyay2020balancing}.
\end{remark}

\section{Novel equivalence results for (augmented) balancing weights and outcome regression models}
\label{sec:newresults}

Our first main result demonstrates that \emph{any} linear balancing weights estimator is equivalent to applying OLS to the re-weighted features. Our second result provides a novel analysis of augmented balancing weights, demonstrating that augmenting any linear balancing weights estimator with a linear outcome regression estimator is equivalent to a plug-in estimator of a new linear model with coefficients that are a weighted combination of estimated OLS coefficients and the coefficients of the original linear outcome model.

\subsection{Weighting alone}
Our first result is that estimating $\psi(m)$ with any linear balancing weights is equivalent to fitting OLS for the regression of $Y_p$ on $\Phi_p$ and then applying those coefficients to the re-weighted target feature profile.
The key idea for this result begins with the simple unregularized regression prediction for $\psi(m)$, $\overbar{\Phi}_q \hat{\beta}_{\text{ols}}$.

\begin{proposition} \label{prop:bal_wt_OLS}
Let $\hat{w}^\delta \coloneqq  \hat{\theta}^\delta \Phi_p^\top$,  $\hat{\theta}^\delta \in \mathbb{R}^d$, be any linear balancing weights, with corresponding weighted features $\hat{\Phi}_q^\delta \coloneqq \tfrac{1}{n} \hat{w}^\delta \Phi_p$. 
Let $\hat{\beta}_{\text{ols}} = ( \Phi_p^\top\Phi_p)^\dag \Phi_p^\top Y_p$ be the OLS coefficients of the regression of $Y_p$ on $\Phi_p$. Then: 
\begin{align*}
\hat{\mathbb{E}}\left[\hat{w}^\delta \circ Y_p\right] &=\hat{\Phi}_q^\delta \hat{\beta}_{\text{ols}}  \\[0.5em]
&=  \left(\bar{\Phi}_p 
+ \widehat{\Delta}^\delta \right) \hat{\beta}_{\text{ols}},
\end{align*}
where $\widehat{\Delta}^\delta =  \hat{\Phi}_q^\delta - \bar{\Phi}_p$ is the mean feature shift implied by the balancing weights and where superscript $\delta$ indicates possible dependence on a hyperparameter. We have assumed without loss of generality that $\bar \Phi_p=0$, but we sometimes use $\hat \Delta$ notation to demonstrate the role of mean feature shift in various expressions. We use the symbol $\circ$ to denote element-wise multiplication.
\end{proposition}

Note that here we have written the OLS coefficients using the pseudo-inverse $\dag$. For clarity in the main text, we focus on the full rank setting, where $(\Phi_p^\top \Phi_p)^\dag = (\Phi_p^\top \Phi_p)^{-1}$; we provide a proof for the general setting in \Cref{sec:regpath-d-bigger-n}.
In \Cref{sec:nonlinear_weights}, we extend Proposition \ref{prop:bal_wt_OLS} to non-linear balancing weights, including those wiht a non-negativity constraint. 

We can interpret this result via a contrast with standard regularization. Regularized regression models navigate a bias-variance trade-off by regularizing estimated coefficients $\hat\beta_{\text{reg}}$ relative to $\hat{\beta}_{\text{ols}}$, leading to $\overbar{\Phi}_q \hat{\beta}_{\text{reg}}$. The balancing weights approach instead keeps $\hat{\beta}_{\text{ols}}$ fixed and regularizes the target feature distribution by penalizing the implied feature shift,  $\widehat{\Delta}^\delta = \hat{\Phi}^\delta_q - \overbar{\Phi}_p$.

We emphasize that this is a new and quite general result. As we discuss in \Cref{apx:equiv-conditions}, it has been shown previously that for exact balancing weights, $ \hat{\mathbb{E}}[\hat{w}_\text{exact} Y_p ] = \overbar{\Phi}_q \hat{\beta}_\text{ols}$. However, \Cref{prop:bal_wt_OLS} holds for any weights of the form $w = \theta \Phi_p^\top$ with arbitrary $\theta \in \mathbb{R}^d$.
In Sections \ref{sec:L2} and \ref{sec:l8_aug_section}, we consider the particular form of $\hat{\Phi}_q^\delta$ for $\ell_2$ and $\ell_\infty$ balancing, respectively.

\subsection{Augmented balancing weights}

We can immediately extend this to augmented balancing weights, which regularize \emph{both} the coefficients and the feature shift.
Let $\hat\beta^{\lambda}_\text{reg}$ be the coefficients of any regularized linear model for the relationship between $Y_p$ and $\Phi_p$, where the superscript $\lambda$ indicates dependence on a hyperparameter (e.g., estimated by regularized least squares). We consider augmenting $\hat{\mathbb{E}}\left[\hat{w}^\delta \circ Y_p\right]$ with $\hat\beta^{\lambda}_\text{reg}$ using
the doubly robust functional representation in Equation \eqref{eq:doubly_robust_functional}. The augmented estimator is:
\begin{align}
   \hat{\mathbb{E}}[\Phi_q \hat{\beta}_\text{reg}^\lambda] + 
        \hat{\mathbb{E}}[ \hat{w}^\delta \circ (Y_p - \Phi_p \hat{\beta}_\text{reg}^\lambda)] \label{eq:aug_bal_est_v1} 
    =    \hat{\mathbb{E}}[ \hat{w}^\delta \circ Y_p] + \hat{\mathbb{E}}\left[\left(\Phi_q - \hat{\Phi}_q^\delta\right)\hat{\beta}_\text{reg}^\lambda\right].
\end{align}
Many recently proposed estimators have this form; see e.g., \citet{athey2018approximate, ben2021balancing}. If the weighting model and outcome model have different bases, our result applies to a shared basis by either combining the dictionaries as in \cite{chernozhukov2022automatic} or by applying an appropriate projection as in \cite{hirshberg2021augmented}. 

We apply Proposition \ref{prop:bal_wt_OLS} to the first term of the right-hand side of \eqref{eq:aug_bal_est_v1} to yield the following result. As this result is purely numerical, it applies to arbitrary vectors $\hat\beta^{\lambda}_\text{reg} \in \mathbb{R}^d$, but substantively we think of $\hat\beta^{\lambda}_\text{reg}$ as the estimated coefficients from an outcome model.
\begin{proposition}\label{generalregularizationpath}
    For any $\hat{\beta}_\text{reg}^\lambda \in \mathbb{R}^d$, and any linear balancing weights estimator with estimated coefficients $ \hat{\theta}^\delta \in \mathbb{R}^d$, and with $\hat{w}^\delta \coloneqq \hat{\theta}^\delta \Phi_p^\top$ and $\hat{\Phi}_q^\delta \coloneqq \tfrac{1}{n} \hat{w}^\delta \Phi_p$, the resulting augmented estimator
    \begin{align*}
       & \hat{\mathbb{E}}[ \hat{w}^\delta \circ Y_p] + \hat{\mathbb{E}}\left[\left(\Phi_q - \hat{\Phi}_q^\delta\right)\hat{\beta}_\text{reg}^\lambda\right] \\
       &= \hat{\mathbb{E}}\left[ \hat{\Phi}_q^\delta \hat{\beta}_{\text{ols}} + \left(\Phi_q - \hat{\Phi}_q^\delta\right)\hat{\beta}_\text{reg}^\lambda\right] \\
        &= \hat{\mathbb{E}}[ \Phi_q \hat{\beta}_\text{aug}],
    \end{align*}
    where the $j$th element of $\hat{\beta}_\text{aug}$ is:
    \begin{align*}
        \hat{\beta}_{\text{aug} ,j} &\coloneqq \left(1-a_j^\delta\right) \hat{\beta}_{\text{reg},j}^\lambda + a_j^\delta \hat{\beta}_{\text{ols},j}\\[0.5em]
        a_j^\delta &\coloneqq \frac{\widehat{\Delta}_{j}^\delta}{ \Delta_j }, 
    \end{align*}
    where $\Delta_j = \overbar{\Phi}_{q,j} - \overbar{\Phi}_{p,j}$ is the observed mean feature shift for feature $j$; and $\widehat{\Delta}^\delta_j = \hat{\Phi}_{q,j}^\delta - \overbar{\Phi}_{p,j}$ is the feature shift for feature $j$ implied by the balancing weights model. 
    Finally, $a^\delta \in [0,1]^d$ when the covariance matrix is diagonal, $ (\Phi_p^\top\Phi_p) = \text{diag}(\sigma^2_1, \sigma^2_2, ..., \sigma^2_d)$,
with $\sigma^2_j > 0$.
\end{proposition}

This is our central numerical result for augmented balancing weights: when both the outcome and weighting models are linear, the augmented estimator is equivalent to a linear model applied to the target features $\Phi_q$, with coefficients that are element-wise affine combinations of the base learner coefficients, $\hat{\beta}_\text{reg}^\lambda$, and the coefficients $\hat{\beta}_{\text{ols}}$ from an OLS regression of $Y_p$ on $\Phi_p$. (The coefficients are additionally \emph{convex} combinations of $\hat{\beta}_\text{reg}^\lambda$ and $\hat{\beta}_{\text{ols}}$ when the covariance matrix is diagonal.)
In Sections \ref{sec:L2} and \ref{sec:l8_aug_section} below, we analyze some of the properties of the augmented estimator for  $\ell_2$ and $\ell_\infty$ balancing weights problems respectively.

 The regularization parameter for the balancing weights problem, $\delta$, parameterizes the path between $\hat{\beta}_\text{reg}^\lambda$ and $\hat{\beta}_\text{ols}$. To see this, consider the cases where $\delta \to 0$ and $\delta \to \infty$. As $\delta \to 0$ the balancing weights problem prioritizes minimizing balance over controlling variance, and $\widehat{\Delta}^\delta_j \to \Delta_j$ for all $j$. (Recall that we assume $\overbar{\Phi}_{p,j} = 0$ for all $j$. Thus, $\Delta_j = \overbar{\Phi}_{q,j}$ and $\widehat{\Delta}_j^\delta = \hat{\Phi}_{q,j}^\delta$. So $\widehat{\Delta}_{j}^\delta \rightarrow \Delta_j$ is equivalent to $\hat{\Phi}_{q}^\delta \rightarrow \overbar{\Phi}_{q,j}$.)
In this case, $a_j^\delta = \widehat{\Delta}_j^\delta / \Delta_j \to 1$, and the weights fully ``de-bias'' the original outcome model by recovering unregularized regression, $\hat{\beta}_\text{aug} \to \hat{\beta}_\text{ols}$. In Section \ref{sec:numerical_illustration}, we will see that when chosen by cross-validation, $\delta$ sometimes equals exactly $0$ in applied problems; thus even when $\hat{\beta}_\text{reg}^\lambda$ is a sophisticated regularized estimator, the final augmented point estimate can nonetheless be numerically equivalent to the simple OLS plug-in estimate.  
Conversely, as $\delta \to \infty$, the balancing weights problem prioritizes controlling variance, leading to uniform weights and $\widehat{\Delta}_j \to 0$. In this case, $a_j^\delta = \widehat{\Delta}_j^\delta / \Delta_j \rightarrow 0$, the weighting model does very little, and $\hat{\beta}_\text{aug} \rightarrow \hat{\beta}_\text{reg}^\lambda$.

It is also instructive to consider two other extremes: unregularized outcome model and unregularized balancing weights. First, consider the special case of fitting an unregularized linear regression outcome model, i.e., $\hat{\beta}_\text{reg}^\lambda = \hat{\beta}_{\text{ols}}$. Then Proposition \ref{generalregularizationpath} reproduces the result, originally due to  \citet{robins2007comment}, that ``OLS is doubly robust" \citep[see also][]{kline2011oaxaca}.
This is because $\hat{\beta}_{\text{aug}}  = \hat{\beta}_{\text{ols}}$ for arbitrary linear weights $\hat{\theta}^\delta \in \mathbb{R}^d$. Thus, OLS augmented by \textit{any} choice of linear balancing weights collapses to OLS alone. Equivalently, we can view OLS alone as an augmented estimator that combines an OLS base learner with linear balancing weights. 

A similar result holds for unregularized balancing weights, i.e., exact balancing weights. Let $\hat{w}_\text{exact}$ be the solution to a balancing weights problem in \Cref{linearbalancingproblems} with hyperparameter $\delta = 0$, and let $\hat{\beta}_\text{reg}^\lambda \in \mathbb{R}^d$ be arbitrary coefficients. Then from the balance condition, $\hat{\Phi}_q = \overbar{\Phi}_q$, $a_j^\delta = 1$ for all $j$, and we have that $\hat{\beta}_{\text{aug}}  = \hat{\beta}_{\text{ols}}$.
Thus, the augmented exact balancing weights estimator also collapses to the OLS regression estimator. Equivalently, 
the augmented exact balancing weights estimator collapses to the \textit{unaugmented} exact balancing weights estimator. \citet{zhao2017entropy} use a very similar result to argue that entropy balancing, a form of exact balancing weights, is doubly robust.

Finally, before we turn to new results for $\ell_2$ and $\ell_\infty$ balancing, we briefly comment on several points that are discussed in more detail in the Appendix.

\begin{remark}[Sample splitting]
Sample splitting is a common technique in the AutoDML literature especially, in which we only apply the outcome and weighting models to data points not used for estimation; see, for example, \citet{newey2018cross, chernozhukov2022automatic}. Since \Cref{generalregularizationpath} holds for arbitrary vectors $\hat{\beta}^\lambda_\text{reg}$ and $\hat{\theta}^\delta$, the results still hold under cross-fitting. See \Cref{sec:sample-split} for an extended discussion.
\end{remark}

\begin{remark}[Infinite dimensional setting]
While we emphasize the linear, low-dimensional setting where $\Phi_p^\top\Phi_p$ is invertible, \Cref{generalregularizationpath} holds far more broadly. The result remains true when the function class $\mathcal{F}$ is a subset of \emph{any} Hilbert space. This includes the high dimensional setting where $d > n$ and the infinite dimensional setting. See \Cref{sec:high-dimensions} for a formal statement.
\end{remark}

\begin{remark}[Nonlinear balancing weights] A rich tradition in survey statistics \citep[e.g.,][]{deville1992calibration}, machine learning \citep[e.g.,][]{menon2016linking}, and causal inference \citep[e.g.,][]{vermeulen2015bias,zhao2019covariate,tan2020regularized} focuses on \emph{non-linear} balancing weights, such as when the weights correspond to a specific \emph{link function} $g(\cdot)$ applied to the linear predictor, $\hat{w} = g(\hat{\theta} \Phi_p^\top)$, or, equivalently, when the balancing weights problem penalizes an alternative dispersion penalty. In Appendix \ref{sec:nonlinear_weights}, we briefly consider extending Proposition \ref{prop:bal_wt_OLS} to nonlinear weights and show that the nonlinearity introduces an additional approximation error. A more thorough extension is a promising direction for future research. 
\end{remark}

\begin{remark}[Non-negative weights] 
A common modification of the (minimum variance) balancing weights problem is to constrain the estimated weights to be non-negative or on the simplex; examples include Stable Balancing Weights \citep{zubizarreta2015stable} and the Synthetic Control Method \citep{abadie2010synthetic}, as well as their augmented analogues \citep{athey2018approximate,ben2021augmented}.  Such weights have a number of attractive practical properties: they limit extrapolation; they ensure that the final weighting estimator is sample bounded; and they are typically sparse, which can sometimes aid interpretability \citep{robins2007comment}. In Appendix \ref{sec:simplex}, we extend Proposition \ref{prop:bal_wt_OLS} and show that restricting weights to be non-negative is equivalent to sample trimming. In particular, let $\hat{w}_+^\delta$ be the estimated non-negative weights and $\hat{\beta}^+_{\text{ols}}$ be the OLS coefficient of the regression of $Y_p$ on $\Phi_p$, but restricted to units with positive weight. Then, Proposition \ref{prop:bal_wt_OLS} continues to hold, but with $\hat{\beta}^+_{\text{ols}}$ in place of the unrestricted $\hat{\beta}_{\text{ols}}$: $\hat{\mathbb{E}}\left[\hat{w}_+^\delta \circ Y_p\right] =  \hat{\Phi}_q^\delta \hat{\beta}_{\text{ols}}^+$. See \citet{arbour2024simplex} for additional discussion of the simplex constraint.
\end{remark}

\begin{remark}[Bilinear form]
As pointed out by a reviewer, (many of) the functionals we consider can be written as a bilinear form $\alpha^T\Sigma \beta$ where $\beta$ is the coefficient for the outcome model, $\alpha$ is the coefficient for the Riesz representer and $\Sigma$ is the some weighted population Gram matrix \citep{robins2008higher}; for $E[Y(1)]$, it would be $E[Z\phi(X)\phi(X)^T]$.  
\Cref{generalregularizationpath} suggests that $\beta$ can be estimated using the methods we discuss here, and moreover that the aggregation weights would then be entangled with $\Sigma$ or $\alpha$. Understanding whether this could be used to then motivate new estimators is an interesting topic for future work. 
\end{remark}

\section{Augmented $\ell_2$ Balancing Weights} \label{sec:L2}

In this section, we study $\ell_2$ balancing weights estimators, which are commonly used in the context of kernel balancing \citep{gretton2012kernel, hirshberg2019minimax, kallus2020generalized, ben2021multilevel} and for panel data methods \citep{abadie2010synthetic, ben2021augmented}. 
We first show that the regularization path $a_j^\delta$ from \Cref{generalregularizationpath} follows typical ridge regression shrinkage, with a smooth decay. Moreover, augmenting with $\ell_2$ balancing weights is equivalent to boosting with ridge regression, and always overfits relative to the unaugmented outcome model alone. We then show that when the outcome model used to augment 
$\ell_2$ balancing weights is also a ridge regression (which we refer to as ``double ridge''), the augmented estimator is itself equivalent to a single, generalized ridge regression, albeit undersmoothed relative to the base learner. These results extend immediately to the RKHS setting of ``double kernel ridge'' estimation, combining kernel balancing weights and kernel ridge regression. 
In Section \ref{sec:kernel_ridge_theory}, we show the implications of these numeric results for undersmoothing in the statistical sense.

While the following results hold for arbitrary covariance matrices, in the main text we simplify the presentation by assuming that $\Phi_p^\top\Phi_p$ is diagonal; that is, $ (\Phi_p^\top\Phi_p) = \text{diag}(\sigma^2_1, \sigma^2_2, ..., \sigma^2_d)$,
with $\sigma^2_j > 0$. We show that this is without loss of generality for $\ell_2$ balancing in \Cref{apx:correlated-features}.

\subsection{General linear outcome model}
Following Remark \ref{remark:ridge_equiv} above, 
$\ell_2$ balancing weights, including kernel balancing weights, have a closed form that is always linear in $\overbar{\Phi}_q$. Our next result applies this closed form to Proposition \ref{generalregularizationpath} to derive the regularization path that results from augmenting an arbitrary linear outcome model with $\ell_2$ balancing weights. Although this is an immediate consequence of Proposition \ref{generalregularizationpath}, the resulting form of the augmented estimator has unique structure that warrants a new result.

\begin{proposition}\label{l2augment}
    Let $\hat{w}_{\ell_2}^\delta$  be (penalized) linear balancing weights with regularization parameter $\delta$ and $\mathcal{F} = \{ f(x) = \theta^\top \phi(x) : \Vert \theta \Vert_2 \leq 1 \}$. Then
     $\tfrac{1}{n} \hat{w}_{\ell_2}^\delta =  \overbar{\Phi}_q(\Phi_p^\top\Phi_p + \delta I)^{-1} \Phi_p^\top  . $
     Therefore, the augmented $\ell_2$ balancing weights estimator with outcome model $\hat{\beta}_\text{reg}^\lambda \in \mathbb{R}^d$ has the form
    \begin{align}
        &\hat{\mathbb{E}}[\Phi_q \hat{\beta}_\text{reg}^\lambda] + \hat{\mathbb{E}}[ \hat{w}_{\ell_2}^\delta (Y_p - \Phi_p \hat{\beta}_\text{reg}^\lambda)] =  \hat{\mathbb{E}}[\Phi_q \hat{\beta}_{\ell_2}],\nonumber
    \end{align}
    where the $jth$ coefficient of $\hat{\beta}_{\ell_2}$ is given by
    \begin{align}
        \hat{\beta}_{\ell_2 , j} &\coloneqq \left(1- a_j^\delta\right) \hat{\beta}^\lambda_{\text{reg},j} + a_j^\delta \hat{\beta}_{\text{ols},j}\label{l2augmodel}\\
        a_j^\delta &\coloneqq \frac{\sigma^2_j}{\sigma^2_j + \delta}.\nonumber
    \end{align}
\end{proposition}
In this case, the $a_j^\delta$ are exactly equal to the standard regularization path of ridge regression. To see this, recall that ridge regression with penalty $\delta$ shrinks the $\hat{\beta}_\text{ols}$ coefficients as follows:
\begin{align} \hat{\beta}_{\text{ridge},j}^\delta = \left(\frac{\sigma^2_j}{\sigma^2_j + \delta}\right)\hat{\beta}_{\text{ols} , j } = a_j^\delta \hat{\beta}_{\text{ols} , j }. \label{eq:ridgeform} \end{align}
This is identical to the expression in \eqref{l2augmodel} but with 
$\hat{\beta}_\text{reg}^\lambda$ set to $0$: Ridge regression shrinks $\hat{\beta}_\text{ols}$ towards $0$ with regularization path $a_j^\delta$, while $\ell_2$ augmenting shrinks $\hat{\beta}_\text{ols}$ towards $\hat{\beta}_\text{reg}^\lambda$ with the same regularization path. 

As an illustration, the right panel of Figure \ref{fig:regpath} shows $\hat{\beta}_{\ell_2}$ (on the y-axis) for ten covariates, with $\delta$ increasing from  $0$ (on the x-axis). The dots on the left pick out $\hat{\beta}_{\text{ols}}$; when $\delta = 0$, then $a_j^0 = 1$ and $\hat{\beta}_{\ell_2} = \hat{\beta}_{\text{ols}}$. The limit on the right shows $\hat{\beta}_{\text{reg}}^\lambda$. The smooth regularization path is characteristic of ridge regression shrinkage. 

\begin{figure}[tbp]
    \centering
    \begin{subfigure}{0.32\textwidth}
        \centering 
        \includegraphics[width=\textwidth]{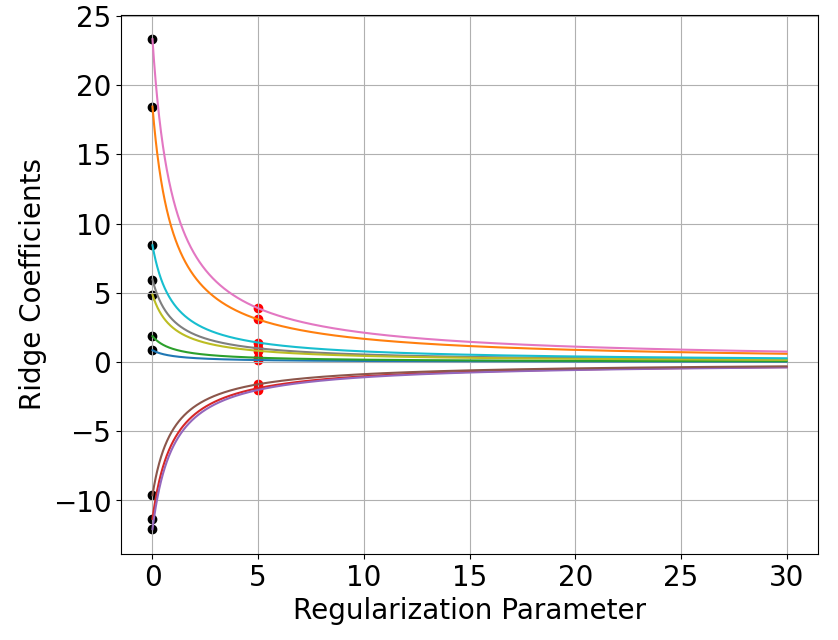}
        \caption{Outcome model}
    \end{subfigure}
    \begin{subfigure}{0.32\textwidth}
        \centering 
        \includegraphics[width=\textwidth]{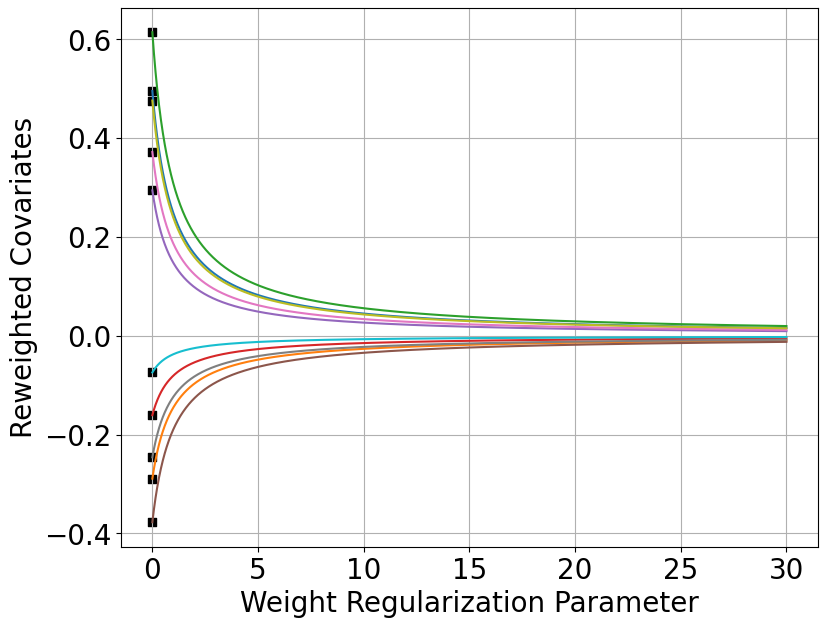}
        \caption{Weighting model}
    \end{subfigure}
        \begin{subfigure}{0.32\textwidth}
        \centering 
        \includegraphics[width=\textwidth]{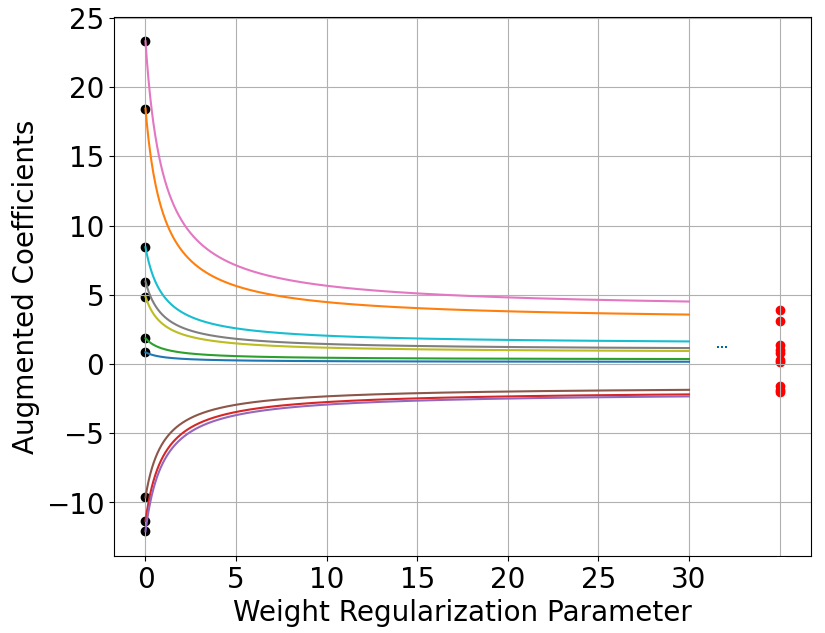}
        \caption{Augmented model}
    \end{subfigure}
\caption{
Regularization paths for ``double ridge'' augmented $\ell_2$ balancing weights. Panel (a) shows the coefficients $\hat{\beta}_{\text{reg}}^\lambda$ of a ridge regression of $Y_p$ on $\Phi_p$ with hyperparameter $\lambda$.
The black dots on the left are the OLS coefficients, with $\lambda = 0$. The red dots at $\lambda = 5$ illustrate the coefficients at a plausible hyperparameter value, $\hat{\beta}_{\text{reg}}^{5}$. 
Panel (b) shows re-weighted covariates, $\hat{\Phi}_q^\delta$, for the $\ell_2$ balancing weights problem with hyperparameter $\delta$; the black dots show exact balance, which corresponds to OLS. As $\delta$ increases, the weights converge to uniform weights and  $\hat{\Phi}_q^\delta$ converges to $\overline{\Phi}_p$, which we have centered at zero.
Panel (c) shows the augmented coefficients, $\hat{\beta}_{\ell_2}$ as a function of the weight regularization parameter $\delta$. The black dots on the left are the OLS coefficients. As $\delta \to \infty$, the coefficients converge to $\hat{\beta}_{\text{reg}}^5$.
All three regularization paths have essentially identical qualitative behavior.
}
    \label{fig:regpath}
\end{figure}

We can also view $\hat{\beta}_{\ell_2}$ as the output of a single iteration of a ridge boosting procedure, fit using $Y_p$ and $\Phi_p$ alone. See \citet{buhlmann2003boosting} and \citet{park2009l_2} for detailed discussion; \citet{newey2004twicing} makes a similar connection in the context of twicing kernels.  
\begin{proposition}\label{prop:l2aug-is-boosting}
Let $\check{Y}_p = Y_p - \Phi_p \hat{\beta}^\lambda_{\text{reg}}$ be the residuals from the base learner.  
Let  $\hat{\beta}_\text{boost}^\delta$ be the coefficients from the ridge regression of $\check{Y}_p$ on $\Phi_p$ with hyperparameter $\delta$. 
Then,
$\hat{\beta}_{\ell_2} = \hat{\beta}_\text{reg}^\lambda + \hat{\beta}_\text{boost}^\delta, $
and
$  \Vert Y_p - \Phi_p \hat{\beta}_{\ell_2} \Vert_2^2 \leq \Vert Y_p - \Phi_p \hat{\beta}_\text{reg}^\lambda\Vert_2^2. $
\end{proposition}

So for a fixed $\delta$, the augmented $\ell_2$ balancing estimator is equivalent to estimating a new outcome model coefficient estimator $\hat{\beta}_{\ell_2}$ that \emph{overfits} relative to $\hat{\beta}_\text{reg}^\lambda$ (in the sense of having smaller in-sample training error), and then applying that model to $\Phi_q$. 

Surprisingly --- and in contrast to the general result in \Cref{generalregularizationpath} --- the augmented coefficients $\hat{\beta}_{\ell_2}$ are the same for \emph{every} target covariate profile $\Phi_q$. To see this, note that \Cref{l2augment} shows that $\ell_2$ balancing weights are always linear in $\overbar{\Phi}_q$. Therefore, the corresponding regularization path $a_j^\delta$ does not depend on the target profile $\Phi_q$; it depends only on $\delta$ and the source distribution variances $\sigma_j^2$.
This property is closely related to \emph{universal adaptability} in the computer science literature on multi-group fairness \citep{kim2022universal}. The particular $\Phi_q$ may nonetheless impact the choice of $\delta$ in hyperparameter selection, e.g., via cross-validating imbalance, which in turn influences the degree of overfitting; we do find this to be the case theoretically in \Cref{sec:finite-sample-mse}.

\subsection{Ridge regression outcome model}
\Cref{l2augment} holds for arbitrary linear outcome model coefficient estimators $\hat{\beta}_\text{reg}^\lambda \in \mathbb{R}^d$; we now state the corresponding result for a ``double ridge'' estimator, where the base learner outcome model is itself fit via ridge regression. The key takeaway is that the implied augmented coefficients are \emph{undersmoothed} relative to the base learner ridge coefficients.

For this section, we will consider the following generalized ridge regression, sometimes known as ``adaptive'' ridge regression \citep{grandvalet1998least}. Let $\Lambda \in \mathbb{R}^{d \times d}$ be a diagonal matrix with $j$th diagonal entry $\lambda_j \geq 0$. Then the generalized ridge coefficients are: 
\begin{align*}
    \hat{\beta}_\text{ridge}^\Lambda &\coloneqq \underset{\beta \in \mathbb{R}^d}{\text{argmin}} \Vert \Phi_p \beta - Y_p \Vert_2^2 + \beta^\top \Lambda \beta\\
    &= (\Phi_p^\top\Phi_p + \Lambda)^{-1} \Phi_p^\top Y_p.
\end{align*}  
Standard ridge regression is the special case where the $\lambda_j$ all take the same value and so $\Lambda = \lambda I$. As above, the generalized ridge coefficients can be rewritten as shrinking the OLS coefficients:
\begin{align} \hat{\beta}_{\text{ridge},j}^\Lambda = \left(\frac{\sigma^2_j}{\sigma^2_j + \lambda_j}\right)\hat{\beta}_{\text{ols} , j }. \label{eq:adaptridgeform} \end{align}
We now demonstrate that the augmented $\ell_2$ balancing weights estimator with base learner $\hat{\beta}_\text{ridge}^\Lambda$ is equivalent to a plug-in estimator using generalized ridge with \emph{smaller} hyperparameters, $\hat{\beta}_\text{ridge}^\Gamma$, where $\Gamma$ is a diagonal matrix with $j$th diagonal entry $\gamma_j \in [0,\lambda_j]$.

\begin{proposition} \label{prop:double_ridge}
Let $\hat{\beta}^\Lambda_{\text{ridge}}$ denote the coefficients of a generalized ridge regression of $Y_p$ on $\Phi_p$ with hyperparameters $\Lambda$, and let $\hat{w}^\delta_{\ell_2}$ denote $\ell_2$ balancing weights with hyperparameter $\delta$ defined in \Cref{linearbalancingproblems}. Define the diagonal matrix $\Gamma$ with $j$th diagonal entry:
\[ \gamma_j \coloneqq \frac{\delta \lambda_j}{\sigma_j^2 + \lambda_j + \delta} \leq \lambda_j.\]
Then:
\begin{align*} 
    &\hat{\mathbb{E}}[\Phi_q\hat{\beta}^\Lambda_{\text{ridge}}] + \hat{\mathbb{E}}[ \hat{w}^\delta_{\ell_2}(Y_p - \Phi_p\hat{\beta}^\Lambda_{\text{ridge}}) ] = \hat{\mathbb{E}}[ \Phi_q  \hat{\beta}_{\text{ridge}}^\Gamma ] .
\end{align*} 
Furthermore, $\hat{\beta}_\text{ridge}^\Gamma$ are standard ridge regression coefficients (i.e., $\gamma_j$ is a constant for all $j$) when $\lambda_j = \lambda$ and $\sigma_j = \sigma$ for all $j$.
\end{proposition}

The same result holds for kernel ridge regression; see \Cref{sec:rkhs_appendix}.

In this setting, augmenting with balancing weights is equivalent to undersmoothing the original outcome model fit. In particular, we can use the expansion in \Cref{eq:adaptridgeform} to see the undersmoothing in $\hat{\beta}^\Gamma_\text{ridge}$ explicitly:
\[ 
\frac{\sigma^2_j}{\sigma^2_j + \gamma_j} = \underbrace{\left(\frac{\sigma^2_j}{\sigma^2_j + \lambda_j} \right)}_{\text{outcome model}} \underbrace{\left(\frac{\sigma_j^2 + \lambda_j + \delta}{\sigma_j^2 + \delta}\right)}_{\text{augmentation}},
\]
where the first term is the shrinkage from the original generalized ridge model alone, and the second term is due to augmenting with $\ell_2$ balancing weights. Importantly, the second term is in $[1, \frac{\sigma^2_j + \lambda_j}{\sigma^2_j}]$ and therefore partially reverses the shrinkage of the original estimate. In Section \ref{sec:kernel_ridge_asymototics}, we connect this to undersmoothing in the statistical sense. 

\section{Augmented $\ell_\infty$ balancing weights}
\label{sec:l8_aug_section}

In this section, we study $\ell_\infty$ balancing weights estimators, which are widely used in the balancing weights literature \citep{zubizarreta2015stable, athey2018approximate} and in the AutoDML literature \citep{chernozhukov2022automatic}. 
In the main text, we consider the special case where the covariance matrix $\Phi_p^\top\Phi_p$ is diagonal; that is, $ (\Phi_p^\top\Phi_p) = \text{diag}(\sigma^2_1, \sigma^2_2, ..., \sigma^2_d)$,
with $\sigma^2_j > 0$. Unlike with $\ell_2$ balancing, this is no longer without loss of generality. 
 We discuss this general case in \Cref{apx:linf-aug-correlated}.

For diagonal covariance, we first show that $\ell_\infty$ balancing has a closed form: it is equivalent to applying a soft-thresholding operator to the feature shift from $\overbar{\Phi}_p$ to $\overbar{\Phi}_q$. We then write the resulting augmented estimator as applying coefficients $\hat{\beta}_{\ell_\infty}$ to $\Phi_q$ and 
show that $\hat{\beta}_{\ell_\infty}$ is a sparse, element-wise convex combination of the base learner coefficients and OLS coefficients. 
When the outcome model is also fit via the lasso, we use the resulting representation to demonstrate a familiar ``double selection'' phenomenon \citep{belloni2014inference}, where $\hat{\beta}_{\ell_\infty}$ inherits the non-zero coefficients of both the base learner and the weighting model. This is a form of undersmoothing in the $\ell_0$ ``norm,'' in the sense that $\hat{\beta}_{\ell_\infty}$ always has at least as many non-zero coefficients as the base learner, $\hat{\beta}_{\text{reg}}$.

\subsection{Weighting alone}

We first define the soft-thresholding operator and show that the $\ell_\infty$ balancing problem has a closed form solution.

\begin{definition*}[Soft-thresholding operator]
    For $t>0$, define the soft-thresholding operator,
    \[     \mathcal{T}_t(z) \coloneqq  
    \begin{cases}
        0 & \text{if } |z| < t\\
        z - t & \text{if } z > t\\
        z + t & \text{if } z < -t
    \end{cases}.\] 
\end{definition*}

\begin{proposition}[$\ell_\infty$ Balancing]\label{orthonormalellinfty}
    If $\Phi_p^\top \Phi_p$ is diagonal, the solution $ w_{\ell_\infty}^\delta$ to the $\ell_\infty$ optimization problem (\ref{supnormproblem}) is:
    \begin{align*}
    \tfrac{1}{n} w_{\ell_\infty}^\delta &= \Phi_p (\Phi_p^\top\Phi_p)^{-1} \left[\overbar{\Phi}_p + \mathcal{T}_\delta(\overbar{\Phi}_q - \overbar{\Phi}_p) \right] \\
    &= \Phi_p (\Phi_p^\top\Phi_p)^{-1} \left[\overbar{\Phi}_p + \mathcal{T}_\delta(\Delta) \right]
    \end{align*}
    where $\Delta = \overbar{\Phi}_q - \overbar{\Phi}_p$, where we include $\overbar{\Phi}_p$ (equal to $0$ by assumption) to emphasize the dependence on feature shift, and with corresponding reweighted features, 
    $\hat{\Phi}_q^\delta = \overbar{\Phi}_p + \mathcal{T}_\delta(\overbar{\Phi}_q - \overbar{\Phi}_p)$.
\end{proposition}
 
For intuition, compare the (un-augmented) $\ell_\infty$ balancing weights estimator to the lasso-based coefficient estimates \citep{hastie2009elements}: 
\begin{align*}
     \hat{\mathbb{E}} [ w_{\ell_\infty}^\delta \circ Y_p] &= \mathcal{T}_\delta(\overbar{\Phi}_q)^\top \hat{\beta}_\text{ols}\\
      \hat{\mathbb{E}} [ \Phi_q \hat{\beta}_\text{lasso}^\lambda ] &= \overbar{\Phi}_q^\top \mathcal{T}_\lambda(\hat{\beta}_\text{ols}),
\end{align*}
where we simplify $\hat{\Phi}_q^\delta$ here to emphasize the connections between the methods. Whereas lasso performs soft-thresholding on the OLS coefficients (regularizing the outcome regression), $\ell_\infty$ balancing performs soft-thresholding on the implied feature shift to the target features.

\subsection{General linear outcome model}

We can then plug the closed-form solution for the weights into  \Cref{generalregularizationpath}.

\begin{proposition}\label{linfaugment}
    Let $\hat{w}_{\ell_\infty}^\delta$ be defined as above. Then the augmented $\ell_\infty$ balancing weights estimator with outcome model fit $\hat{\beta}_\text{reg}^\lambda \in \mathbb{R}^d$ has the form,
    \begin{align}
        &\hat{\mathbb{E}}[\Phi_q \hat{\beta}_\text{reg}^\lambda] + \hat{\mathbb{E}}[ \hat{w}_{\ell_\infty}^\delta (Y_p - \Phi_p \hat{\beta}_\text{reg}^\lambda)] =  \hat{\mathbb{E}}[\Phi_q \hat{\beta}_{\ell_\infty} ],\nonumber
    \end{align}
    where the $jth$ coefficient of $ \hat{\beta}_{\ell_\infty} $ equals:
   \[   \hat{\beta}_{\ell_\infty, j}  =
\begin{cases}
        \hat{\beta}_{\text{reg},j}^\lambda & \text{if } |\Delta_j| < \delta \\[0.5em]
        \left| \frac{\delta}{\Delta_j} \right| \hat{\beta}_{\text{reg},j}^\lambda + \left(1 - \left| \frac{\delta}{\Delta_j} \right| \right) \hat{\beta}_{\text{ols},j}  & \text{otherwise} \\
    \end{cases},
 \]
 where $\Delta_j = \overbar{\Phi}_{q,j} - \overbar{\Phi}_{p,j}$.
\end{proposition}

The augmented coefficients $\hat{\beta}_{\ell_\infty}$ are an element-wise convex combination of $\hat{\beta}_\text{reg}^\lambda$ and $\hat{\beta}_\text{ols}$. For features where the mean feature shift $\Delta_j$ is small (relative to $\delta$), $\hat{\beta}_{\ell_\infty}$ is equivalent to the base learner coefficient $\hat{\beta}_\text{reg}^\lambda$. The remaining coefficients are interpolated linearly toward the $\hat{\beta}_\text{ols}$ coefficients.

\Cref{fig:ortho-linf-regpath} summarizes these results and their implications for the augmented estimator. As with \Cref{fig:regpath}, we generate simple simulated data with $d = 10$. In the left panel, we plot the coefficients from lasso regression of $Y_p$ on $\Phi_p$ as a function of the lasso regularization parameter. The regularization path begins with the black dots, which represent the OLS coefficients. Each lasso coefficient (represented by a colored line) then shrinks linearly to exactly zero, due to the soft-thresholding operator. The middle panel plots the reweighted covariates using $\ell_\infty$ balancing weights between $\Phi_p$ and $\Phi_q$ solved in the constrained form. The black dots represent $\overbar{\Phi}_q$, corresponding to exact balance. Then as the weight regularization parameter increases, the reweighted covariates shrink linearly to exactly zero, just as in lasso. The right panel plots coefficients for the augmented estimator that combines a baseline outcome model fit $\hat{\beta}_\text{reg}^\lambda$ with $\ell_\infty$ balancing weights. The lines correspond to $\hat{\beta}_{\ell_\infty}$ as defined in \Cref{linfaugment}. 
The regularization path begins at the black dots, where $\hat{\beta}_{\ell_\infty}  = \hat{\beta}_\text{ols}$, and eventually converges to $\hat{\beta}_\text{reg}^\lambda$, showing the usual soft-thresholding behavior. The order at which the coefficients go to zero reflects the size of $\overbar{\Phi}_q$, because the regularization path depends on the weight coefficients from the middle panel. Thus, the augmented estimator shrinks $\hat{\beta}_\text{ols}$ toward $\hat{\beta}_\text{reg}^\lambda$ but via a soft-thresholding operator applied to the feature shift, $\Delta_j$. 

\begin{figure}[tb]
    \centering
    \begin{subfigure}{0.32\textwidth}
        \centering\includegraphics[width=\textwidth]{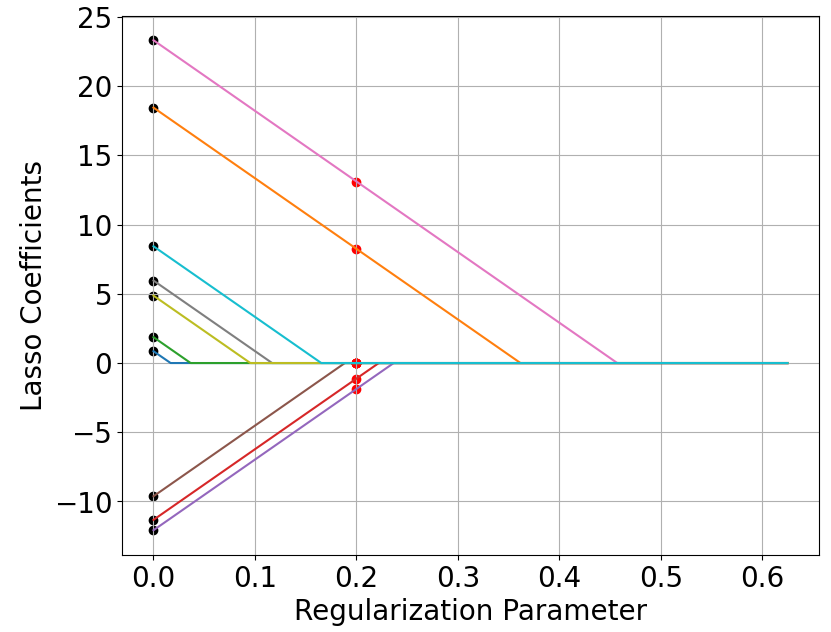}
        \caption{Outcome model}
    \end{subfigure}
    \begin{subfigure}{0.32\textwidth}
        \centering\includegraphics[width=\textwidth]{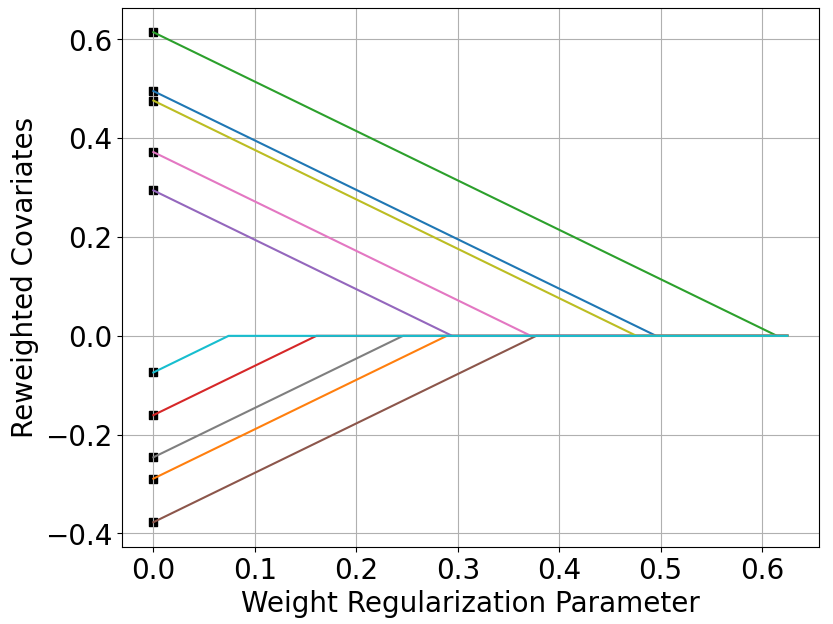}
        \caption{Weighting model}
    \end{subfigure}
        \begin{subfigure}{0.32\textwidth}
        \centering\includegraphics[width=\textwidth]{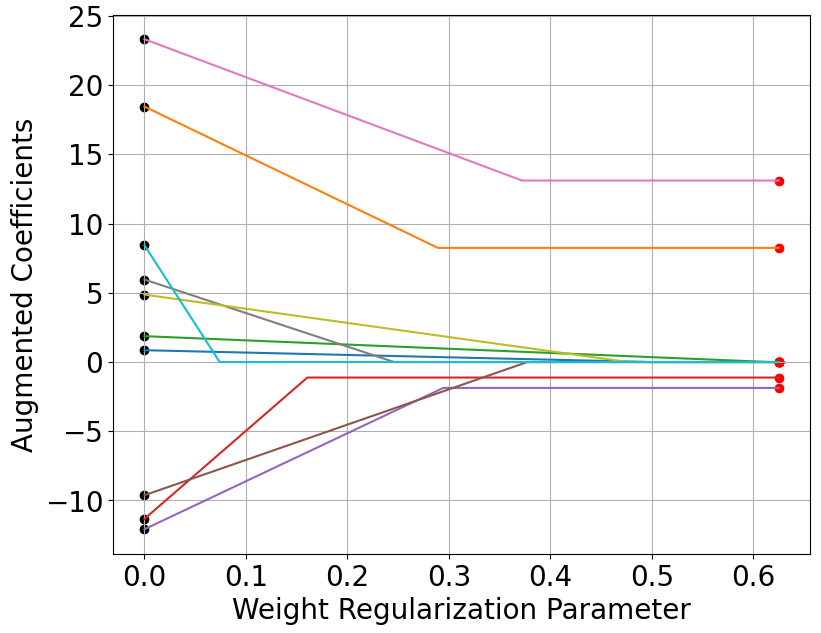}
        \caption{Augmented model}
    \end{subfigure}
    \caption{
    Regularization paths for ``double lasso'' augmented $\ell_\infty$ balancing weights. Panel (a) shows the coefficients $\hat{\beta}_{\text{reg}}^\lambda$ of a lasso regression of $Y_p$ on $\Phi_p$ with hyperparameter $\lambda$.
The black dots on the left are the OLS coefficients, with $\lambda = 0$. The red dots at $\lambda = 0.2$ illustrate the coefficients at a plausible hyperparameter value, $\hat{\beta}_{\text{reg}}^{0.2}$. 
Panel (b) shows re-weighted covariates, $\hat{\Phi}_q^\delta$, for the $\ell_\infty$ balancing weights problem with hyperparameter $\delta$; the black dots show exact balance, which corresponds to OLS. As $\delta$ increases, the weights converge to uniform weights and  $\hat{\Phi}_q^\delta$ converges to $\overline{\Phi}_p$, which we have centered at zero.
Panel (c) shows the augmented coefficients, $\hat{\beta}_{\ell_\infty}$ as a function of the weight regularization parameter $\delta$. The black dots on the left are the OLS coefficients. As $\delta \to \infty$, the coefficients converge to $\hat{\beta}_{\text{reg}}^{0.2}$.
All three regularization paths show the typical lasso ``soft thresholding'' behavior. The regularization path for the augmented estimator also shows ``double selection'' behavior.
    }
    \label{fig:ortho-linf-regpath}
\end{figure}

\subsection{Lasso outcome model}

In the case where $\hat{\beta}_\text{reg}^\lambda$ is itself fit via lasso, as studied in \cite{chernozhukov2022automatic}, then we recover a familiar double selection phenomenon \citep{belloni2014inference}. 

\begin{proposition}[Double Selection]\label{prop:double-selection}
    Let $\hat{\beta}_\text{lasso}^\lambda$ denote the coefficients of lasso regression of $Y_p$ on $\Phi_p$  with regularization parameter $\lambda$. Denote the indices of the non-zero coefficients as $I_\lambda$. Let $\hat{w}_{\ell_\infty}^\delta$ be $\ell_\infty$ balancing weights with parameter $\delta$ as in \Cref{orthonormalellinfty}. Let $I_\delta$ denote the non-zero entries of the reweighted covariates $\hat{\Phi}_q$. Assume that $\hat{\beta}_\text{ols}$ is dense. Then the indices of the non-zero entries of the augmented coefficients $\hat{\beta}_{\ell_\infty}$ are
    $ I_\text{aug} = I_\lambda \cup I_\delta. $
\end{proposition}

The lasso coefficients have a sparsity pattern generated by soft-thresholding the OLS coefficients. The augmented estimator then shrinks from OLS toward $\hat{\beta}_\text{reg}^\lambda$ by soft-thresholding the implied feature shift to the target features. As a result, wherever the lasso coefficients are non-zero \emph{or} the weight coefficients are non-zero, the final augmented coefficients are also non-zero. The ``included coefficients'' for the final estimator are then the union of the coefficients included in either individual model. 
Therefore, augmenting a lasso outcome model with $\ell_\infty$ balancing also exhibits a form of undersmoothing in the $\ell_0$ ``norm'', $\| \hat{\beta}_{\ell_\infty}\|_0$, in the sense that there are always at least as many non-zero coefficients as for the unaugmented lasso outcome model. However, this will not correspond to undersmoothing the base learner in the traditional sense, because in general there will not exist a lasso hyperparameter $\lambda$ that will produce sparsity pattern $I_\text{aug}$.

As noted by, for example, \citet{tang2023ultra}, the double selection estimator may suffer from imprecision due to adjustment for covariates that are associated with treatment but not outcome. One could in principle remove covariates that are only predictive of the treatment, but this can jeopardize statistical inference. We refer to \citet{moosavi2023costs} for a discussion on navigating this trade-off.

\section{Kernel Ridge Regression: Asymptotic and Finite Sample Analysis}
\label{sec:kernel_ridge_theory}

The results above are \emph{numerical}: they hold without any statistical or causal assumptions. However, the connection between augmented estimators and outcome models also presents \emph{statistical} insights that we discuss here. 
In particular, we leverage the numerical result that double (kernel) ridge regression --- which uses ridge regression for fitting both the outcome and weighting models --- is equivalent to a single, undersmoothed outcome ridge regression plug-in estimator.

First, we consider an asymptotic analysis in \Cref{sec:kernel_ridge_asymototics}: we use this equivalence to make explicit the connection between asymptotic results for augmented balancing weights with kernel ridge regression and prior results on optimal undersmoothing of a kernel ridge plug-in estimator.
As a result, optimally undersmoothed kernel ridge regression inherits whatever guarantees can be proven for augmented ridge regression. An implication is that we can generalize the insight from \citet{robins2007comment} that OLS is doubly robust to a wider class of non-parametric estimators.
This equivalence also suggests an appropriate hyperparameter scheme when the outcome regression is an element of an RKHS.

Second, we consider a finite sample analysis in \Cref{sec:finite-sample-mse}: we use this equivalence to derive the finite-sample design-conditional mean squared error of augmented kernel ridge regression. We then use this expression to characterize finite-sample-optimal hyperparameter tuning. We turn to hyperparameter tuning in practice in the next section.

\subsection{Asymptotic Results}
\label{sec:kernel_ridge_asymototics}

We now use our results in \Cref{prop:double_ridge} to make explicit the connection between two otherwise distinct sets of asymptotic results. First, \citet{wong2018kernel} and \citet{singh2021debiased} argue that double kernel ridge regression can deliver $\sqrt{n}$-consistent estimation of functionals in certain scenarios. \citet{wong2018kernel} also proposes an optimally undersmoothed $\ell_2$ balancing weights estimator.
Separately, \citet{hirshberg2019minimax} and \citet{mou2023kernel} propose optimally undersmoothed (single) kernel ridge outcome regression. 
Since, as we have shown in \Cref{prop:double_ridge} (see also Remark 2), these three procedures are equivalent, we can connect these results and show that plug-in estimators based on optimally undersmoothed kernel ridge regression or $\ell_2$ balancing weights can be $\sqrt{n}$-consistent. Moveover, results on RKHSs suggest a simple heuristic for hyperparameter choice. We give the high-level argument here and defer additional technical details to Appendix \ref{sec:asymptotics_appendix}.

To move from numerical results to statistical results, we must place some constraints on the data generating process. Assume that we observe $n$ iid samples of $(x_i,y_i,z_i)$ from $p$. Define $K \in \mathbb{R}^{n \times n}$ to be the kernel matrix with $i,j$-th entry $K_{ij} = k( (x_i, z_i), (x_j, z_j))$. Let $\sigma_j^2$ denote the eigenvalues of $K$. We assume that $\sigma_j^2 = \sigma^2>0$ is constant for all $j$; we can relax this at the cost of additional complexity. The ``single'' kernel ridge regression outcome regression estimator with parameter $\lambda$ has coefficient estimates: 
\[ \hat{\beta}_\text{ridge}^\lambda = (K + \lambda I)^{-1} y. \]
Applying \Cref{prop:double_ridge}, the augmented ``double kernel ridge'' estimator with hyperparameter $\delta$ is equivalent to a plug-in estimate for a new kernel ridge model:
\[  \hat{\beta}_\text{aug} = (K + \gamma I)^{-1} y, \qquad \mbox{with} \;\; \gamma = \frac{\lambda \delta}{\sigma^2 + \lambda + \delta}. \]
For statistical guarantees, we must typically allow the hyperparameters to change with $n$; let $\gamma_n$, $\lambda_n$ and $\delta_n$ then denote sequences of hyperparameters. In the doubly robust framework, one can choose $\lambda_n$ and $\delta_n$ in a way that is MSE-optimal for prediction purposes whilst ensuring that the bias of the augmented estimator is small. For two functions of $n$, $f_n$ and $g_n$, let $f_n \asymp g_n$ denote that $f_n = O(g_n)$ and $g_n = O(f_n)$. Then due to special properties of RKHS geometry, it follows  that $\delta_n$ can be of the same order as $\lambda_n$, that is $\delta_n \asymp \lambda_n$ \citep[][Theorem 5.2]{singh2020kernel}. In the next section, we consider setting $\delta = \lambda$ for hyperparameter tuning in practice; our \Cref{prop:double_ridge} then implies that  $\gamma_n\asymp \lambda_n^2$. We note more generally that \Cref{prop:double_ridge} implies that  $\gamma_n\asymp \lambda_n \delta_n$.

There are two important cases to consider.
When the RKHS is finite dimensional, the choice $\lambda_n = \delta_n = n^{-1/2}$ is optimal for controlling the prediction error for both the outcome and weighting models \citep{caponnetto2007optimal,singh2020kernel}. The augmented estimator is then equivalent to a single ridge regression with hyperparameter $\gamma_n \asymp n^{-1}$, which matches the rate of \citet{hirshberg2019minimax,mou2023kernel}.
Hence, this approach will always undersmooth relative to the MSE-optimal hyperparameter for a single ridge regression.

When the RKHS is infinite-dimensional, we find that the undersmoothed hyperparameter implied by the augmented procedure can take on a range of asymptotic rates, both faster and slower than $n^{-1}$, depending on effective dimension and smoothness; we give concrete examples in the Appendix. This somewhat contrasts with the results in \citet{hirshberg2019minimax,mou2023kernel}. 
In this sense, \Cref{prop:double_ridge} generalizes the standard undersmoothing arguments, which typically change the regularization schedule from $n^{-1/2}$ to $n^{-1}$. 

\begin{remark}[Single-model double robustness]
    Another interesting implication of the equivalence of these two procedures is that the single kernel ridge procedure is doubly robust, much the same way OLS is. Because estimating the coefficients from an OLS regression of $Y$ onto features of $(Z,X)$ is equivalent to a balancing weights or an IPW estimator based on a model for the inverse weights that is linear in the same features, this procedure is consistent whenever \emph{either} the weights or the outcome model is truly linear--that is, whenever either of these two linear models is correctly specified \citep{robins2007comment}. Similarly, the single kernel ridge procedure is doubly robust in that it is consistent if either the true outcome regression or the inverse propensity score is consistently estimated. However, valid inference in the case where the inverse weight model but \emph{not} the outcome model is truly linear will typically require different tuning parameter selection. 
\end{remark}

\subsection{Finite Sample Mean-Squared Error}\label{sec:finite-sample-mse}

We now use our numerical equivalences to write out the exact finite-sample mean squared error of the augmented kernel ridge estimator: by re-writing the augmented balancing weights estimator as a single outcome model, we can immediately leverage existing results from \citet{dobriban2018high}.

Following their setup, we define the diagonal matrix $\hat{\Sigma} \coloneqq \tfrac{1}{n} \Phi_p^T\Phi_p$; if $\hat{\Sigma}$ is not diagonal, we can apply the rotation in  \Cref{apx:l2-rot-wlog}. We consider ridge regression with rescaled hyperparameter $\lambda$ and solution $(\hat{\Sigma} + \lambda I)^{-1} \Phi_p Y_p/n$; this is equivalent to standard ridge regression above with hyperparameter $n \lambda$, and also accommodates kernel ridge regression with appropriate choice of $\Phi_p$. 
Assume that $Y_p = \Phi_p \beta_0 + \epsilon$ with $\beta_0 \in \mathbb{R}^d$, and where $\epsilon \in \mathbb{R}^n$ are iid with mean zero and variance $\sigma^2$. Then the exact, design-conditional, squared bias and variance of the ridge regression prediction applied to a new iid sample $(\Phi_\text{new} , Y_\text{new}) \sim p$ are: 
\begin{align*}
    B^2_{\text{p}}( \lambda ) &= \lambda^2 \beta_0 ^T (\hat{\Sigma} + \lambda I)^{-1} \mathbb{E}[\Phi_p^T \Phi_p] (\hat{\Sigma} + \lambda I)^{-1} \beta_0\\
    V_{\text{p}}( \lambda ) &= \frac{\sigma^2}{n} \text{tr} \left[ \hat{\Sigma} (\hat{\Sigma} + \lambda I)^{-1} \mathbb{E}[\Phi_p^T \Phi_p] (\hat{\Sigma} + \lambda I)^{-1} \right].
\end{align*}
Applying \Cref{prop:double_ridge}, we can simlarly derive the squared bias and variance of an augmented ridge estimator for our linear functional estimand; 
we denote these quantities $B_q^2$ and $V_q$ respectively. We express the bias and variance in terms of the two hyperparameters, $\lambda$ and $\delta$:
\begin{proposition}\label{prop:finite-sample-mse}
 Let $\sigma^2_j$ denote the eigenvalues of $\hat{\Sigma}$ and define $\Gamma_{\lambda, \delta}$ to be the diagonal matrix with non-zero entries $\gamma_j \coloneqq \frac{\delta \lambda}{\sigma^2_j + \delta + \lambda}$. Then,
\begin{align*}
    B^2_{\text{q}}( \lambda, \delta ) &= \beta_0 ^T (\hat{\Sigma} + \Gamma_{\lambda, \delta})^{-1}  \Gamma_{\lambda, \delta} \mathbb{E}[\Phi_q]^T \mathbb{E}[\Phi_q]  \Gamma_{\lambda, \delta} (\hat{\Sigma} +  \Gamma_{\lambda, \delta})^{-1} \beta_0\\
    V_{\text{q}}( \lambda, \delta) &= \frac{\sigma^2}{n} \text{tr} \left[ \hat{\Sigma} (\hat{\Sigma} + \Gamma_{\lambda, \delta})^{-1} \mathbb{E}[\Phi_q]^T \mathbb{E}[\Phi_q] (\hat{\Sigma} +  \Gamma_{\lambda, \delta})^{-1} \right].
\end{align*}
\end{proposition}

In the next section, we compare --- numerically and via simulation --- existing hyperparameter selection schemes to the optimal trade-off between $B_q^2$ and $V_q$. However, first we note that the analysis above opens up exciting new avenues for both theoretical and methodological work. One could theoretically analyze the mean squared error to understand how the optimal $\delta$ scales with the problem parameters; for example, by using proportionate asymptotics from random matrix theory as in the high-dimensional ridge regression literature \citep{hastie2022surprises}. Second, our analysis here suggests a novel, more complex hyperparameter selection scheme based directly on the finite sample analysis. We leave this to future work.

\section{Numerical illustrations and hyperparameter tuning}
\label{sec:numerical_illustrations}

This section illustrates our results in practice. We first explore hyperparameter tuning for double ridge regression, comparing practical methods to the optimal hyperparameter computed using our results from \Cref{prop:finite-sample-mse}. Following our asymptotic results in \ref{sec:kernel_ridge_asymototics}, we recommend equating the weighting and outcome model hyperparameters in practice.
We then apply both double ridge and lasso-augmented $\ell_\infty$-balancing to two versions of the canonical \citet{lalonde1986evaluating} application. 
An important theme throughout is that some approaches for hyperparameter selection frequently lead to $\delta = 0$, which collapses the augmented estimate to OLS alone --- even in settings where this is far from optimal. 
Overall, we take this as a warning that existing hyperparameter tuning schemes can be potentially misleading when applied naively.

\subsection{Hyperparameter tuning for ridge-augmented $\ell_2$ balancing}
\label{sec:hyperpar_tuning_sim}

We begin with practical hyperparameter tuning for the special case of double ridge, building on the MSE expression in Section \ref{sec:finite-sample-mse}. There is an active literature on selecting hyperparameters for augmented balancing weights estimators and double machine learning estimators more broadly \citep{kallus2020generalized, wang2020minimal, ben2021balancing, bach2024hyperparameter}. We contribute to this literature by comparing practical hyperparameter tuning schemes with an oracle hyperparameter tuning scheme based on \Cref{prop:finite-sample-mse}. 

Reflecting empirical practice, we focus here on choosing hyperparameters sequentially: we first select the outcome model hyperparameter $\lambda$ (e.g. by cross-validation) and then select the weighting model hyperparameter $\delta$. Ultimately, we find strong performance for both \emph{CV imbalance} and \emph{CV outcome} hyperparameters, as defined below.
We especially recommend the latter as a reasonable starting point in practice. 
In additional to theoretical support from our asymptotic analysis, the outcome model hyperparameter scheme does not require any additional algorithm or code after having fit the initial outcome model. 

\subsubsection{Oracle and practical hyperparameter tuning}

\paragraph{Oracle hyperparameter.} 

To compute oracle hyperparameters, we first compute the prediction-MSE-optimal $\lambda$ using the standard ridge regression MSE expression, and then we use \Cref{prop:finite-sample-mse} to compute the corresponding optimal $\delta$ for the linear functional estimand: 
\begin{align*}
    \lambda^* &\coloneqq \text{argmax}_{\lambda} \{ B^2_p(\lambda) + V_p(\lambda) \}\\
    \delta^* &\coloneqq \text{argmax}_{\delta} \{ B^2_q(\lambda^*, \delta) + V_q(\lambda^*, \delta) \}.
\end{align*}
While there is not a closed form for $\delta^*$, we can nonetheless directly compute this optimal hyperparameter and characterize its behavior under a range of scenarios.
We draw several conclusions about optimal $\delta^*$ for a wide range of DGPs of the form $Y_p = \Phi_p \beta_0 + \epsilon$.
First, $\delta^*$ is generally increasing in the noise, $\sigma^2$: larger $\sigma^2$ typically implies larger $\delta^*$.
Second, $\delta^*$ generally depends on the target mean, $\mathbb{E}[\Phi_q]$; that is, two DGPs that are identical except for $\mathbb{E}[\Phi_q]$ can have different values of $\delta^*$. 
The optimal hyperparameter, however, does \emph{not} depend on the magnitude of the shift in the target mean: 
replacing $\mathbb{E}[\Phi_q]$ with $c \mathbb{E}[\Phi_q]$ for $c \neq 0$, scales both the bias and variance by $c^2$, leaving $\delta^*$ unchanged.

\paragraph{Practical hyperparameter.}
We compare the oracle hyperparameter with three implementable practical proposals. In all cases, we first pick $\lambda$ by cross-validating the mean squared error of a ridge outcome model. 
\begin{itemize}
    \item \emph{CV imbalance.} Choose $\delta$ by cross-validating the estimated imbalance, $\Vert \tfrac{1}{n} \hat{w} \Phi_p - \bar{\Phi}_q \Vert_2^2$ , adapting a proposal from \citet{wang2020minimal}.
    \item \emph{CV Riesz loss.} Choose $\delta$ by cross-validating the Riesz loss in Equation \eqref{eq:autoform}, adapting a proposal from \citet{chernozhukov2022automatic}; this is the dual form of cross-validating the estimated imbalance.
    \item \emph{CV outcome.} Choose $\delta$ to be equal to the cross-validated ridge outcome $\lambda$, as inspired by the asymptotic theory in \citet{singh2020kernel}.
\end{itemize}

Before presenting simulation results, we provide a preliminary analytic discussion, comparing these practical schemes to the behavior of the oracle $\delta^*$.  For the first two proposals: just like the oracle, both depend on the target mean $\mathbb{E}[\Phi_q]$ and are invariant to re-scaling. However, these two approaches are mechanically independent of the outcomes $Y_p$, unlike the oracle $\delta^*$ which, in general, depends on the variance of the outcomes. On the other hand, the last proposal depends on the outcomes $Y_p$ but is mechanically independent of $\mathbb{E}[\Phi_q]$. 

This suggests that any one of these tuning parameter approaches cannot perform well across all DGPs. In future work, if we pursue a theoretical analysis of the oracle hyperparameter, e.g. in a proportionate asymptotics framework, we may be able predict when either the outcomes or the covariate shift is more important. In this work we begin by demonstrating that no one tuning scheme does uniformly best in simulations. 

\subsubsection{Simulation study}

\begin{table}[tb]
    \centering
    \begin{tabular}{ccccccccc}
        \hline
         & 
        \multicolumn{2}{c}{\# of DGPs} & & 
        \multicolumn{3}{c}{Relative MSE} && \\
        \cline{2-3} \cline{5-7} Method & Best & Worst & & Median & Best & Worst  & & $\text{Prop.}(\delta = 0)$\\ \hline
        CV Outcome & 10 & 3 & & 0.58 & 0.097 & $2 \times 10^5$  & & 0 \\ 
        CV Imbalance & 25 & 2 & & 0.39 & 0.001 & $2 \times 10^5$ & & 0\\  
        CV Riesz Loss & 1 & 31 & & 3,454 & 0.23 & $3 \times 10^7$ & & 0.56 \\ \hline  
    \end{tabular}
    \caption{\label{tab:sim_results}Mean-squared error (relative to the oracle) for four hyperparameter selection methods for \emph{double ridge regression} from a numerical investigation of 36 data generating processes (30 synthetic and 6 semi-synthetic). The final column is the proportion of draws where the hyperparameter $\delta = 0$.}
\end{table}

To assess the behavior of these hyperparameter tuning schemes, we conduct an extensive simulation study using 36 distinct data-generating processes, 30 synthetic and 6 semi-synthetic; see \Cref{apx:simulation-details} for a detailed discussion. 
For each DGP, we directly compute the oracle hyperparameter using the results in Section \ref{sec:finite-sample-mse}. We then compute values from the three practical hyperparameter tuning methods discussed above. The mean squared error that we consider is design-conditional, and so we draw samples of the covariates for each DGP only once. 

Table \ref{tab:sim_results} presents a summary of the MSE for the three methods across the 36 DGPs. Overall, we find that the \emph{CV outcome} approach of choosing $\delta = \lambda$ and the \emph{CV imbalance} approach both perform well in practice: these two achieve the lowest MSE in 35 of the 36 DGPs, with CV imbalance performing slightly better on average. By contrast, selecting $\delta$ via CV for the Riesz loss has numerical stability problems that compromises performance. The performance for the \emph{outcome} and \emph{balance} approaches, on the other hand, seem to degrade gracefully and rarely perform catastrophically. 
Taken together, these preliminary findings suggest researchers should begin with these two tuning methods as defaults.

\paragraph{Recovering the OLS point estimate.} As we discuss above (see, e.g., \Cref{fig:regpath}), when $\delta = 0$ the point estimate for the augmented balancing weights estimator is numerically identical to the OLS point estimate.
Thus, when a hyperparameter tuning procedure chooses $\delta = 0$ in practice, researchers are simply estimating the equivalent of OLS --- even if they are unaware they are doing so.
This is especially problematic in settings where OLS is far from optimal \citep[though see][for counterexamples]{kobak2020optimal, hastie2022surprises}. In our synthetic and semi-synthetic DGPs, $\delta = 0$ is never optimal, and is usually associated with a very large error driven by extreme variance --- see for example, \Cref{fig:bias-var-plot} in the Appendix. Thus the fact that hyperparameter tuning procedures can return $\delta = 0$ in these DGPs represents a pathological case. 

In our simulation study, we find that, when cross validating the Riesz loss,  over half of all draws returned $\delta = 0$. By contrast, none of the other methods returned $\delta = 0$ in the synthetic DGPs, though, as we discuss below, we do observe exact zeros for $\delta$ occasionally when cross-validating imbalance in the standard LaLonde dataset. This further highlights the numeric instability of hyperparameter tuning via CV for the Riesz loss, at least in the settings we consider here. We further suggest that in these cases, practitioners assess the sensitivity of the $\delta = 0$ results to the particular tuning procedure used or to the random choice of cross-validation splits.

\subsection{Application to \citet{lalonde1986evaluating}} \label{sec:numerical_illustration}

We now illustrate our equivalence and hyperparameter tuning results on real-world datasets. Following \citet{chernozhukov2022automatic}, we focus on the canonical \citet{lalonde1986evaluating} data set evaluating a job training program in the National Supported Work (NSW) Demonstration. The primary outcome of interest is annual earnings in 1978 dollars.

For these illustrations, we estimate the Average Treatment Effect on the Treated (ATT), $\mathbb{E}[Y(1) - Y(0) \mid Z = 1]$. We recover the missing conditional mean $\mathbb{E}[Y(0) \mid Z = 1]$ using the setup from Example \ref{ex:dist_shift} in Appendix \ref{sec:examples_appendix}, where the source and target populations are the control and treated units respectively. Thus $\Phi_p$ and $\Phi_q$ correspond to the feature expansion $\phi(X)$ applied to the covariates in the control group and treated group respectively. 
We consider two different features expansions of the original covariates: (1) a ``short'' set of 11 covariates used in \cite{dehejia1999causal};\footnote{These are: age, years of education, Black indicator, Hispanic indicator, married indicator, 1974 earnings, 1975 earnings, age squared, years of education squared, 1974 earnings squared, and 1975 earnings squared.} and (2) an expanded, ``long'' set of 171 interacted features used in \citet{farrell2015robust}.

Our goal is to explicate how augmented estimators under different hyperparameter tuning schemes undersmooth in practice in both low and high-dimensional settings. In some cases,  the augmented estimator collapses to exactly OLS as we document above. 
\Cref{apx:numerical-experiment-details} contains extensive additional analyses, including dataset summaries, additional results from the Infant Health Development Program (IHDP), and sensitivity of these numerical results to cross-fitting.

\subsubsection{High-dimensional setting}\label{sec:numerical-high-dim-results}

Following \citet{chernozhukov2022automatic}, we first consider the expanded set of 171 features for \citet{lalonde1986evaluating} used in \cite{farrell2015robust}. \Cref{fig:lalonde-long} shows estimates for ridge-augmented $\ell_2$ balancing (top row) and lasso-augmented $\ell_\infty$ balancing (bottom row). We explicitly characterize these results in terms of undersmoothing in \Cref{apx:undersmooth-lalonde}. The left two panels of each row show the cross-validation curves for the outcome regression and balancing weights, respectively. 
The right panels show the point estimate as a function of the weighting hyperparamter $\delta$, holding the outcome model hyperparameter $\lambda$ fixed; the black triangle represents the OLS plug-in point estimate. For context, the corresponding experimental estimate is \$1,794 \citep[see][]{dehejia1999causal}.
The green and red dotted lines correspond to hyperparameters chosen by cross-validating balance and the Riesz loss, respectively. For the double ridge estimate, the purple line corresponds to $\delta = \hat{\lambda}$, the outcome hyperparameter selected via cross validation. 

\Cref{fig:lalonde-long} highlights that both the imbalance and the point estimate are highly nonlinear close to zero. Thus, even small departures from OLS (at $\delta = 0$) lead to large changes in the point estimate --- in \Cref{apx:semisynth-biasvsvar} we give some suggestive evidence that the variance blows up relative to the bias in this range. We can also assess the sensitivity of the point estimate to the hyperparameter selection scheme. In this case, choosing $\delta$ via CV balance leads to meaningfully larger choices than via other methods. 

Finally, the selected $\delta$ is always strictly greater than zero for this high-dimensional dataset. However, we find this is sensitive to small perturbations in the problem parameters. For example, when we perturb $\mathbb{E}[\Phi_q]$ by adding a small value to all the even elements, then the cross-validated $\ell_2$ Riesz loss chooses $\delta = 0$ in 38\% of draws of the cross-validation splits. As suggested by \Cref{apx:semisynth-biasvsvar} and our simulation results, this is likely to result in extremely large mean squared error.

\begin{figure}[htb!]
    \centering
    \begin{subfigure}{0.32\textwidth}
        \centering 
        \includegraphics[width=\textwidth]{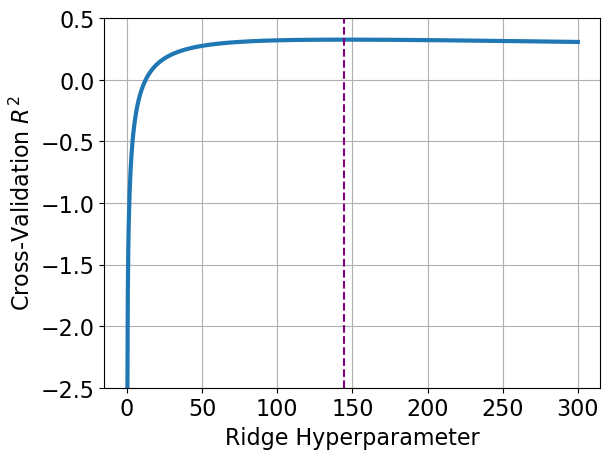}
        \caption{Ridge outcome model}
            \label{fig:lalonde-long-ridge-panel-a}
    \end{subfigure}
    \begin{subfigure}{0.325\textwidth}
        \centering 
        \includegraphics[width=\textwidth]{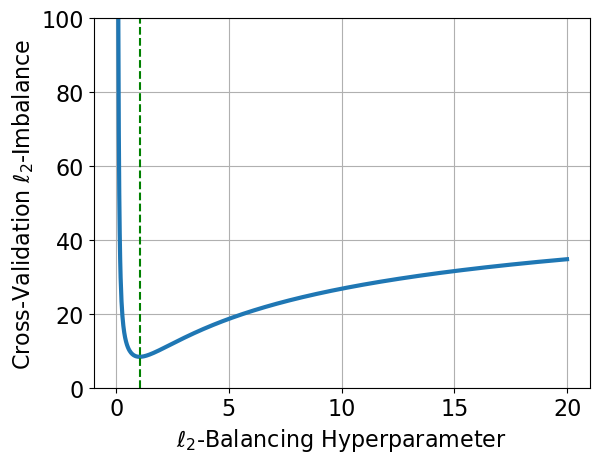}
        \caption{$\ell_2$ balancing}
         \label{fig:lalonde-long-ridge-panel-b}
    \end{subfigure}
        \begin{subfigure}{0.31\textwidth}
        \centering 
        \includegraphics[width=\textwidth]{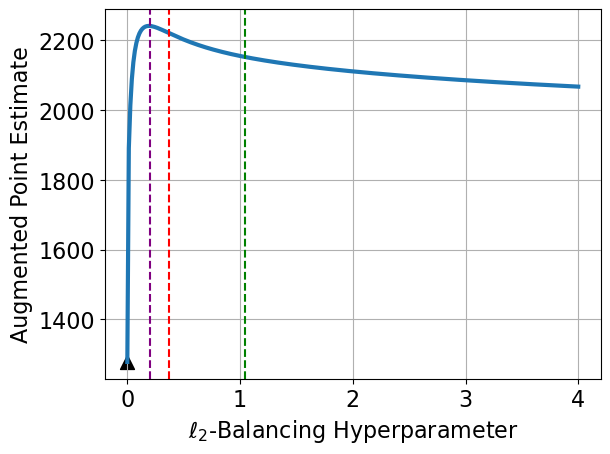}
        \caption{Estimate from ``double ridge''}
        \label{fig:lalonde-long-ridge-panel-c}
    \end{subfigure} \\[1em]
    \begin{subfigure}{0.32\textwidth}
        \centering 
        \includegraphics[width=\textwidth]{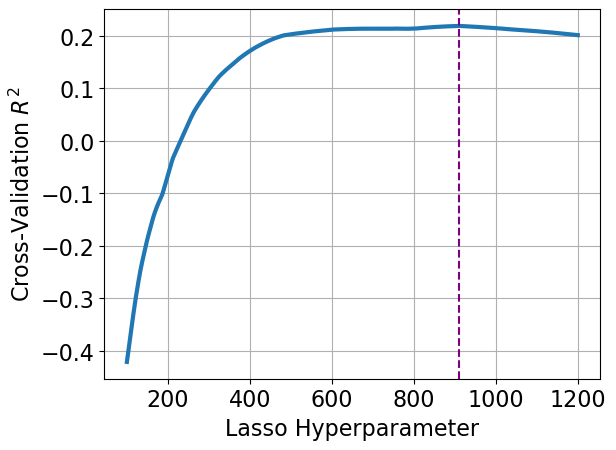}
        \caption{Lasso outcome model}
        \label{fig:lalonde-long-lasso-panel-a}
    \end{subfigure}
    \begin{subfigure}{0.325\textwidth}
        \centering 
        \includegraphics[width=\textwidth]{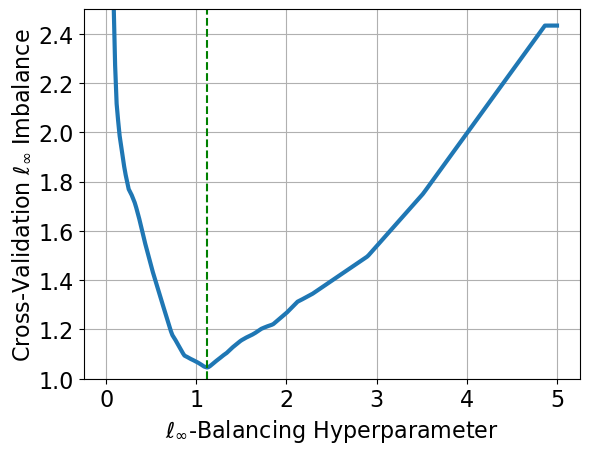}
        \caption{$\ell_\infty$ balancing}
        \label{fig:lalonde-long-lasso-panel-b}
    \end{subfigure}
        \begin{subfigure}{0.31\textwidth}
        \centering 
        \includegraphics[width=\textwidth]{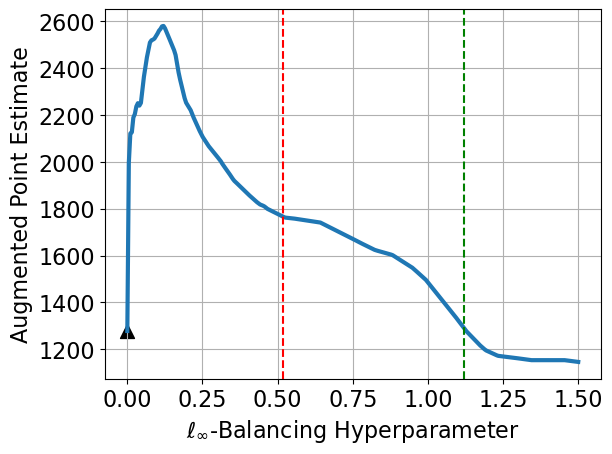}
        \caption{Estimate from ``double lasso''}
        \label{fig:lalonde-long-lasso-panel-c}
    \end{subfigure}

\caption{Augmented balancing weights estimates for the \citet{lalonde1986evaluating} data set with the expanded set of 171 features used in \citet{farrell2015robust}; the top row shows ridge-augmented $\ell_2$ balancing, and the bottom row shows lasso-augmented $\ell_\infty$ balancing.
Panels (a) and (d) show the 3-fold cross-validated $R^2$ for the ridge- and lasso-penalized regression of $Y_p$ on $\Phi_p$ among control units across the hyperparameter $\lambda$; the purple dotted lines show the CV-optimal value for each. 
Panel (b) and (e) show the 3-fold cross-validated imbalance for $\ell_2$ and $\ell_\infty$ balancing weights across the hyperparameter $\delta$; the green dotted lines show the CV-optimal value for each. 
Panels (c) and (f) show the point estimates for the augmented estimators across the weighting hyperparameter $\delta$; the black triangles correspond to the OLS point estimate; the green and red dotted lines correspond to the cross-validated balance and Riesz loss respectively; the purple line corresponds to the cross-validated ridge hyperparameter (for $\delta = \hat{\lambda})$. The variance-based hyperparameter for ridge is $\hat{\sigma}^2/n^2 = 104.8$ and for lasso is $137.5$. The corresponding point estimates are $1923.6$ and $725.8$ respectively, essentially equal to the plug-in outcome model estimates. }
    \label{fig:lalonde-long} 
\end{figure}

\subsubsection{Low-dimensional setting: Recovering OLS}\label{sec:short-regression}

\begin{figure}[tb]
    \centering
    \begin{subfigure}{0.32\textwidth}
        \centering 
        \includegraphics[width=\textwidth]{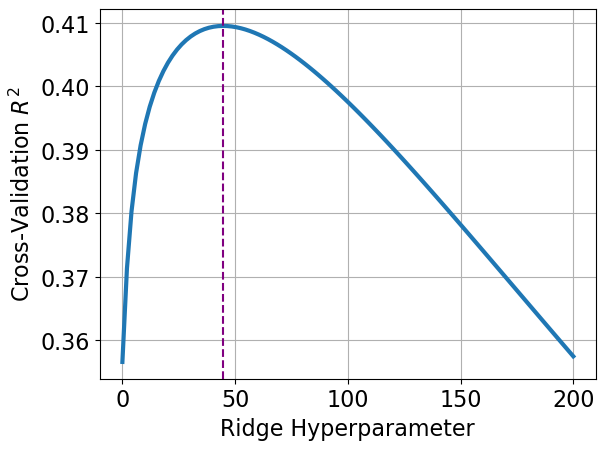}
        \caption{Ridge outcome model}
                \label{fig:lalonde-short-ridge-panel-a}
    \end{subfigure}
    \begin{subfigure}{0.325\textwidth}
        \centering 
        \includegraphics[width=\textwidth]{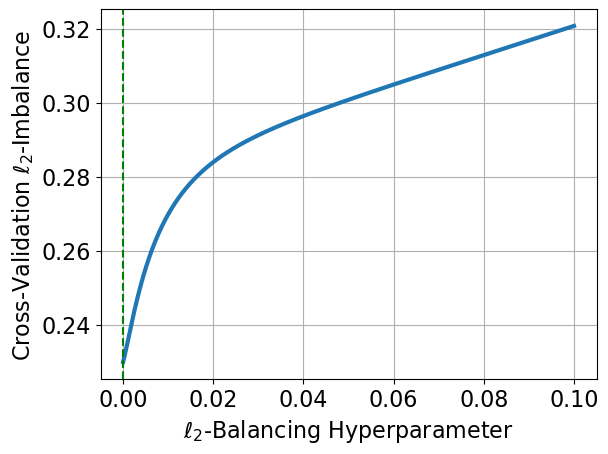}
        \caption{$\ell_2$ balancing}
                \label{fig:lalonde-short-ridge-panel-b}
    \end{subfigure}
        \begin{subfigure}{0.31\textwidth}
        \centering 
        \includegraphics[width=\textwidth]{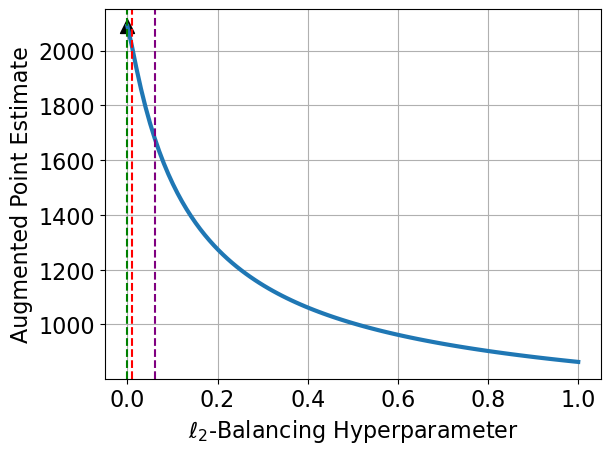}
        \caption{Augmented estimate}
        \label{fig:lalonde-short-ridge-panel-c}
    \end{subfigure}
\caption{Ridge-augmented $\ell_2$ balancing weights (``double ridge'') for \citet{lalonde1986evaluating} with the original 11 covariates. Panel (a) shows the 3-fold cross-validated $R^2$ for the Ridge-penalized regression of $Y_p$ on $\Phi_p$ among control units across the hyperparameter $\lambda$; the purple dotted line shows the CV-optimal value, $\hat{\lambda}$. Panel (b) shows the 3-fold cross-validated imbalance for $\ell_2$ balancing weights across the hyperparameter $\delta$; the green dotted line shows the CV-optimal value, which is $\delta = 0$ or exact balance. Panel (c) shows the point estimate for the augmented estimator across the weighting hyperparameter $\delta$; the black triangle corresponds to the OLS point estimate, the green dotted line corresponds to cross-validated balance, the red dotted line corresponds to cross-validated Riesz loss, and the purple dotted line corresponds to the ridge outcome hyperparameter.}
    \label{fig:lalonde-short-ridge}
\end{figure}

Finally, we apply double ridge to the ``short'' version of the \citet{lalonde1986evaluating} data set with 11 features. \Cref{fig:lalonde-short-ridge}
shows the cross-validation curves for the outcome and weighting models, as well as the point estimate as a function of the balance hyperparameter, with the OLS estimate given by the black triangle.
As above, the green, red, and purple dotted lines correspond to hyperparameters chosen by cross-validating balance, cross-validating the Riesz loss, and choosing $\delta=\lambda$ respectively.  

Unlike for the ``long'' dataset in Figure \ref{fig:lalonde-long}, \Cref{fig:lalonde-short-ridge} does not display quite as stark nonlinearity around zero. Importantly, however, setting $\delta$ by cross-validating imbalance or the Riesz loss yields $\delta = 0$ (up to numerical imprecision), which reduces the augmented estimator to exactly the estimate from a simple OLS regression --- even though the base learner ridge outcome model is heavily regularized. 
By contrast, our preferred hyperparameter tuning scheme of choosing $\delta = \lambda$ results in an estimate that is roughly \$400 dollars smaller than the OLS estimate.

\section{Discussion}
\label{sec:discussion}

We have shown that augmenting a plug-in regression estimator with linear balancing weights results in a new plug-in estimator with coefficients that are shrunk towards --- in some cases all the way to --- the estimates from OLS fit on the same observations. We generalize this equivalence for different choices of outcome and weighting regressions. In the asymptotic setting, we draw the explicit connection between augmented estimators and undersmoothing for the special case of kernel ridge regression. Then we derive the design-conditional finite sample MSE for the double ridge estimator, and use it to solve numerically for oracle hyperparameters. We compare the oracle hyperparameters with three practical tuning schemes and then illustrate our results on the canonical LaLonde data set. In the Appendix also explore many extensions, including to nonlinear weights and to high-dimensional features.

There are many promising avenues for future research. The fundamental connection between doubly robust estimation and undersmoothing opens up several theory directions. While we focus on the special case of kernel ridge regression in Section \ref{sec:kernel_ridge_asymototics}, we anticipate that these connections will hold more broadly. Similarly, while our focus in this paper has been on interpreting balancing weights as a form of linear regression, the converse is also valid: we could instead focus on how many outcome regression-based plug-in estimators are, in fact, a form of balancing weights; see \citet{lin2022regression} for connections between outcome modeling and density ratio estimation.

We also anticipate that the MSE we derive in \Cref{sec:finite-sample-mse} is a starting place for future theoretical analysis that can inform practice. We demonstrate in our simulation study that existing hyperparameter selection methods cannot perform uniformly well over all DGPs. We expect that analyzing the optimal hyperparameters, for example in a proportionate asymptotics regime, can either help devise new tuning schemes or inform which tuning method will work best on the dataset at hand.

We conjecture that these results may provide new insights into the estimation of causal effects in the proximal causal inference framework \citep{tchetgen2020introduction}. This framework uses proxy variables to identify causal effects in the presence of unmeasured confounding. Estimation has been complicated by the fact that, in the absence of strong parametic assumptions,  estimators of proximal causal effects are solutions to ill-posed Fredholm integral equations. \cite{ghassami2022minimax} and \cite{kallus2021causal} recently proposed tractable nonparametric estimators in this setting. They use an ``adversarial'' version of double kernel ridge regression --- allowing the weighting and outcome models to have different bases --- to estimate the solution to the required Fredholm integral equations. Our results apply immediately to standard augmented estimators with different bases for the outcome and weighting models, either via a union basis \citep{chernozhukov2022automatic} or by applying an appropriate projection as in \cite{hirshberg2021augmented}, and extending these results to proximal causal effect estimators might help in constructing new proximal balancing weights, matching, or regression estimators with attractive asymptotic properties.

Finally, many common panel data estimators are forms of augmented balancing weight estimation \citep{abadie2010synthetic, ben2021augmented, arkhangelsky2021synthetic}. We plan to use the numeric results here, especially the results for simplex-constrained weights in Appendix \ref{sec:simplex}, to better understand connections between methods and to inform inference.

\clearpage
\bibliographystyle{abbrvnat}
\bibliography{equiv}

\begin{thebibliography}{103}
\providecommand{\natexlab}[1]{#1}
\providecommand{\url}[1]{\texttt{#1}}
\expandafter\ifx\csname urlstyle\endcsname\relax
  \providecommand{\doi}[1]{doi: #1}\else
  \providecommand{\doi}{doi: \begingroup \urlstyle{rm}\Url}\fi

\bibitem[Abadie et~al.(2010)Abadie, Diamond, and
  Hainmueller]{abadie2010synthetic}
A.~Abadie, A.~Diamond, and J.~Hainmueller.
\newblock Synthetic control methods for comparative case studies: Estimating
  the effect of california’s tobacco control program.
\newblock \emph{Journal of the American statistical Association}, 105\penalty0
  (490):\penalty0 493--505, 2010.

\bibitem[Agarwal and Singh(2021)]{agarwal2021causal}
A.~Agarwal and R.~Singh.
\newblock Causal inference with corrupted data: Measurement error, missing
  values, discretization, and differential privacy.
\newblock \emph{arXiv preprint arXiv:2107.02780}, 2021.

\bibitem[Agarwal et~al.(2022)Agarwal, Tan, Ronen, Singh, and
  Yu]{agarwal2022randomforests}
A.~Agarwal, Y.~S. Tan, O.~Ronen, C.~Singh, and B.~Yu.
\newblock Hierarchical shrinkage: Improving the accuracy and interpretability
  of tree-based models.
\newblock In K.~Chaudhuri, S.~Jegelka, L.~Song, C.~Szepesvari, G.~Niu, and
  S.~Sabato, editors, \emph{Proceedings of the 39th International Conference on
  Machine Learning}, volume 162 of \emph{Proceedings of Machine Learning
  Research}, pages 111--135. PMLR, 17--23 Jul 2022.
\newblock URL \url{https://proceedings.mlr.press/v162/agarwal22b.html}.

\bibitem[Arbour and Feller(2024)]{arbour2024simplex}
D.~Arbour and A.~Feller.
\newblock The role of simplex constraints in regularizing treatment effect
  estimates.
\newblock 2024.

\bibitem[Arkhangelsky et~al.(2021)Arkhangelsky, Athey, Hirshberg, Imbens, and
  Wager]{arkhangelsky2021synthetic}
D.~Arkhangelsky, S.~Athey, D.~A. Hirshberg, G.~W. Imbens, and S.~Wager.
\newblock Synthetic difference-in-differences.
\newblock \emph{American Economic Review}, 111\penalty0 (12):\penalty0
  4088--4118, 2021.

\bibitem[Armagan et~al.(2011)Armagan, Clyde, and
  Dunson]{armagan2011generalized}
A.~Armagan, M.~Clyde, and D.~Dunson.
\newblock Generalized beta mixtures of gaussians.
\newblock \emph{Advances in neural information processing systems}, 24, 2011.

\bibitem[Athey et~al.(2018)Athey, Imbens, and Wager]{athey2018approximate}
S.~Athey, G.~W. Imbens, and S.~Wager.
\newblock Approximate residual balancing: debiased inference of average
  treatment effects in high dimensions.
\newblock \emph{Journal of the Royal Statistical Society: Series B (Statistical
  Methodology)}, 80\penalty0 (4):\penalty0 597--623, 2018.

\bibitem[Bach et~al.(2024)Bach, Schacht, Chernozhukov, Klaassen, and
  Spindler]{bach2024hyperparameter}
P.~Bach, O.~Schacht, V.~Chernozhukov, S.~Klaassen, and M.~Spindler.
\newblock Hyperparameter tuning for causal inference with double machine
  learning: A simulation study.
\newblock \emph{arXiv preprint arXiv:2402.04674}, 2024.

\bibitem[Bai and Ghosh(2017)]{bai2017inverse}
R.~Bai and M.~Ghosh.
\newblock The inverse gamma-gamma prior for optimal posterior contraction and
  multiple hypothesis testing.
\newblock \emph{arXiv preprint arXiv:1710.04369}, 2017.

\bibitem[Bartlett et~al.(2020)Bartlett, Long, Lugosi, and
  Tsigler]{bartlett2020benign}
P.~L. Bartlett, P.~M. Long, G.~Lugosi, and A.~Tsigler.
\newblock Benign overfitting in linear regression.
\newblock \emph{Proceedings of the National Academy of Sciences}, 117\penalty0
  (48):\penalty0 30063--30070, 2020.

\bibitem[Bauer et~al.(2007)Bauer, Pereverzev, and
  Rosasco]{bauer2007regularization}
F.~Bauer, S.~Pereverzev, and L.~Rosasco.
\newblock On regularization algorithms in learning theory.
\newblock \emph{Journal of complexity}, 23\penalty0 (1):\penalty0 52--72, 2007.

\bibitem[Belloni et~al.(2014)Belloni, Chernozhukov, and
  Hansen]{belloni2014inference}
A.~Belloni, V.~Chernozhukov, and C.~Hansen.
\newblock Inference on treatment effects after selection among high-dimensional
  controls.
\newblock \emph{The Review of Economic Studies}, 81\penalty0 (2):\penalty0
  608--650, 2014.

\bibitem[Ben-Michael et~al.(2021{\natexlab{a}})Ben-Michael, Feller, and
  Hartman]{ben2021multilevel}
E.~Ben-Michael, A.~Feller, and E.~Hartman.
\newblock Multilevel calibration weighting for survey data.
\newblock \emph{arXiv preprint arXiv:2102.09052}, 2021{\natexlab{a}}.

\bibitem[Ben-Michael et~al.(2021{\natexlab{b}})Ben-Michael, Feller, Hirshberg,
  and Zubizarreta]{ben2021balancing}
E.~Ben-Michael, A.~Feller, D.~A. Hirshberg, and J.~R. Zubizarreta.
\newblock The balancing act in causal inference.
\newblock \emph{arXiv preprint arXiv:2110.14831}, 2021{\natexlab{b}}.

\bibitem[Ben-Michael et~al.(2021{\natexlab{c}})Ben-Michael, Feller, and
  Rothstein]{ben2021augmented}
E.~Ben-Michael, A.~Feller, and J.~Rothstein.
\newblock The augmented synthetic control method.
\newblock \emph{Journal of the American Statistical Association}, 116\penalty0
  (536):\penalty0 1789--1803, 2021{\natexlab{c}}.

\bibitem[Benkeser and Van Der~Laan(2016)]{benkeser2016hal}
D.~Benkeser and M.~Van Der~Laan.
\newblock The highly adaptive lasso estimator.
\newblock In \emph{2016 IEEE international conference on data science and
  advanced analytics (DSAA)}, pages 689--696. IEEE, 2016.

\bibitem[Bruns-Smith and Feller(2022)]{bruns2022outcome}
D.~A. Bruns-Smith and A.~Feller.
\newblock Outcome assumptions and duality theory for balancing weights.
\newblock In \emph{International Conference on Artificial Intelligence and
  Statistics}, pages 11037--11055. PMLR, 2022.

\bibitem[B{\"u}hlmann and Yu(2003)]{buhlmann2003boosting}
P.~B{\"u}hlmann and B.~Yu.
\newblock Boosting with the l 2 loss: regression and classification.
\newblock \emph{Journal of the American Statistical Association}, 98\penalty0
  (462):\penalty0 324--339, 2003.

\bibitem[Caponnetto and De~Vito(2007)]{caponnetto2007optimal}
A.~Caponnetto and E.~De~Vito.
\newblock Optimal rates for the regularized least-squares algorithm.
\newblock \emph{Foundations of Computational Mathematics}, 7:\penalty0
  331--368, 2007.

\bibitem[Carvalho et~al.(2010)Carvalho, Polson, and
  Scott]{carvalho2010horseshoe}
C.~M. Carvalho, N.~G. Polson, and J.~G. Scott.
\newblock The horseshoe estimator for sparse signals.
\newblock \emph{Biometrika}, 97\penalty0 (2):\penalty0 465--480, 2010.

\bibitem[Chattopadhyay and Zubizarreta(2021)]{chattopadhyay2021implied}
A.~Chattopadhyay and J.~R. Zubizarreta.
\newblock On the implied weights of linear regression for causal inference.
\newblock \emph{arXiv preprint arXiv:2104.06581}, 2021.

\bibitem[Chattopadhyay et~al.(2020)Chattopadhyay, Hase, and
  Zubizarreta]{chattopadhyay2020balancing}
A.~Chattopadhyay, C.~H. Hase, and J.~R. Zubizarreta.
\newblock Balancing vs modeling approaches to weighting in practice.
\newblock \emph{Statistics in Medicine}, 39\penalty0 (24):\penalty0 3227--3254,
  2020.

\bibitem[Chernozhukov et~al.(2018{\natexlab{a}})Chernozhukov, Chetverikov,
  Demirer, Duflo, Hansen, Newey, and Robins]{chernozhukov2018double}
V.~Chernozhukov, D.~Chetverikov, M.~Demirer, E.~Duflo, C.~Hansen, W.~Newey, and
  J.~Robins.
\newblock {Double/debiased machine learning for treatment and structural
  parameters}.
\newblock \emph{The Econometrics Journal}, 21\penalty0 (1):\penalty0 C1--C68,
  01 2018{\natexlab{a}}.

\bibitem[Chernozhukov et~al.(2018{\natexlab{b}})Chernozhukov, Newey, and
  Singh]{chernozhukov2018learning}
V.~Chernozhukov, W.~K. Newey, and R.~Singh.
\newblock Learning l2-continuous regression functionals via regularized riesz
  representers.
\newblock \emph{arXiv preprint arXiv:1809.05224}, 8, 2018{\natexlab{b}}.

\bibitem[Chernozhukov et~al.(2022{\natexlab{a}})Chernozhukov, Escanciano,
  Ichimura, Newey, and Robins]{chernozhukov2022locally}
V.~Chernozhukov, J.~C. Escanciano, H.~Ichimura, W.~K. Newey, and J.~M. Robins.
\newblock Locally robust semiparametric estimation.
\newblock \emph{Econometrica}, 90\penalty0 (4):\penalty0 1501--1535,
  2022{\natexlab{a}}.

\bibitem[Chernozhukov et~al.(2022{\natexlab{b}})Chernozhukov, Newey,
  Quintas-Martinez, and Syrgkanis]{chernozhukov2022riesznet}
V.~Chernozhukov, W.~Newey, V.~M. Quintas-Martinez, and V.~Syrgkanis.
\newblock Riesznet and forestriesz: Automatic debiased machine learning with
  neural nets and random forests.
\newblock In \emph{International Conference on Machine Learning}, pages
  3901--3914. PMLR, 2022{\natexlab{b}}.

\bibitem[Chernozhukov et~al.(2022{\natexlab{c}})Chernozhukov, Newey, Singh, and
  Syrgkanis]{chernozhukov2022dynamic}
V.~Chernozhukov, W.~Newey, R.~Singh, and V.~Syrgkanis.
\newblock Automatic debiased machine learning for dynamic treatment effects and
  general nested functionals.
\newblock \emph{arXiv preprint arXiv:2203.13887}, 2022{\natexlab{c}}.

\bibitem[Chernozhukov et~al.(2022{\natexlab{d}})Chernozhukov, Newey, and
  Singh]{chernozhukov2022automatic}
V.~Chernozhukov, W.~K. Newey, and R.~Singh.
\newblock Automatic debiased machine learning of causal and structural effects.
\newblock \emph{Econometrica}, 90\penalty0 (3):\penalty0 967--1027,
  2022{\natexlab{d}}.

\bibitem[Chernozhukov et~al.(2022{\natexlab{e}})Chernozhukov, Newey, and
  Singh]{chernozhukov2022riesz}
V.~Chernozhukov, W.~K. Newey, and R.~Singh.
\newblock Debiased machine learning of global and local parameters using
  regularized riesz representers.
\newblock \emph{The Econometrics Journal}, 2022{\natexlab{e}}.

\bibitem[Chernozhukov et~al.(2024)Chernozhukov, Newey, Quintas-Martinez, and
  Syrgkanis]{chernozhukov2024riesz_regression}
V.~Chernozhukov, W.~K. Newey, V.~Quintas-Martinez, and V.~Syrgkanis.
\newblock Automatic debiased machine learning via riesz regression, 2024.

\bibitem[Dehejia and Wahba(1999)]{dehejia1999causal}
R.~H. Dehejia and S.~Wahba.
\newblock Causal effects in nonexperimental studies: Reevaluating the
  evaluation of training programs.
\newblock \emph{Journal of the American statistical Association}, 94\penalty0
  (448):\penalty0 1053--1062, 1999.

\bibitem[Deville and S{\"a}rndal(1992)]{deville1992calibration}
J.-C. Deville and C.-E. S{\"a}rndal.
\newblock Calibration estimators in survey sampling.
\newblock \emph{Journal of the American statistical Association}, 87\penalty0
  (418):\penalty0 376--382, 1992.

\bibitem[Dobriban and Wager(2018)]{dobriban2018high}
E.~Dobriban and S.~Wager.
\newblock High-dimensional asymptotics of prediction: Ridge regression and
  classification.
\newblock \emph{The Annals of Statistics}, 46\penalty0 (1):\penalty0 247--279,
  2018.

\bibitem[Duchi et~al.(2021)Duchi, Glynn, and Namkoong]{duchi2021el}
J.~C. Duchi, P.~W. Glynn, and H.~Namkoong.
\newblock Statistics of robust optimization: A generalized empirical likelihood
  approach.
\newblock \emph{Mathematics of Operations Research}, 46\penalty0 (3):\penalty0
  946--969, 2021.

\bibitem[Farrell(2015)]{farrell2015robust}
M.~H. Farrell.
\newblock Robust inference on average treatment effects with possibly more
  covariates than observations.
\newblock \emph{Journal of Econometrics}, 189\penalty0 (1):\penalty0 1--23,
  2015.

\bibitem[Fischer and Steinwart(2020)]{fischer2020sobolev}
S.~Fischer and I.~Steinwart.
\newblock Sobolev norm learning rates for regularized least-squares algorithms.
\newblock \emph{The Journal of Machine Learning Research}, 21\penalty0
  (1):\penalty0 8464--8501, 2020.

\bibitem[Fuller(2002)]{fuller2002regression}
W.~A. Fuller.
\newblock Regression estimation for survey samples.
\newblock \emph{Survey Methodology}, 28\penalty0 (1):\penalty0 5--24, 2002.

\bibitem[Gao et~al.(2022)Gao, Yang, and Kim]{gao2022soft}
C.~Gao, S.~Yang, and J.~K. Kim.
\newblock Soft calibration for selection bias problems under mixed-effects
  models.
\newblock \emph{arXiv preprint arXiv:2206.01084}, 2022.

\bibitem[Ghassami et~al.(2022)Ghassami, Ying, Shpitser, and
  Tchetgen]{ghassami2022minimax}
A.~Ghassami, A.~Ying, I.~Shpitser, and E.~T. Tchetgen.
\newblock Minimax kernel machine learning for a class of doubly robust
  functionals with application to proximal causal inference.
\newblock In \emph{International Conference on Artificial Intelligence and
  Statistics}, pages 7210--7239. PMLR, 2022.

\bibitem[Goldstein and Messer(1992)]{goldstein1992optimal}
L.~Goldstein and K.~Messer.
\newblock Optimal plug-in estimators for nonparametric functional estimation.
\newblock \emph{The annals of statistics}, pages 1306--1328, 1992.

\bibitem[Graham et~al.(2012)Graham, de~Xavier~Pinto, and
  Egel]{graham2012inverse}
B.~S. Graham, C.~C. de~Xavier~Pinto, and D.~Egel.
\newblock Inverse probability tilting for moment condition models with missing
  data.
\newblock \emph{The Review of Economic Studies}, 79\penalty0 (3):\penalty0
  1053--1079, 2012.

\bibitem[Grandvalet(1998)]{grandvalet1998least}
Y.~Grandvalet.
\newblock Least absolute shrinkage is equivalent to quadratic penalization.
\newblock In \emph{International Conference on Artificial Neural Networks},
  pages 201--206. Springer, 1998.

\bibitem[Gretton et~al.(2012)Gretton, Borgwardt, Rasch, Sch{\"o}lkopf, and
  Smola]{gretton2012kernel}
A.~Gretton, K.~M. Borgwardt, M.~J. Rasch, B.~Sch{\"o}lkopf, and A.~Smola.
\newblock A kernel two-sample test.
\newblock \emph{The Journal of Machine Learning Research}, 13\penalty0
  (1):\penalty0 723--773, 2012.

\bibitem[Hainmueller(2012)]{hainmueller2012entropy}
J.~Hainmueller.
\newblock Entropy balancing for causal effects: A multivariate reweighting
  method to produce balanced samples in observational studies.
\newblock \emph{Political analysis}, 20\penalty0 (1):\penalty0 25--46, 2012.

\bibitem[Harshaw et~al.(2019)Harshaw, S{\"a}vje, Spielman, and
  Zhang]{harshaw2019balancing}
C.~Harshaw, F.~S{\"a}vje, D.~Spielman, and P.~Zhang.
\newblock Balancing covariates in randomized experiments with the gram--schmidt
  walk design.
\newblock \emph{arXiv preprint arXiv:1911.03071}, 2019.

\bibitem[Hastie et~al.(2009)Hastie, Tibshirani, Friedman, and
  Friedman]{hastie2009elements}
T.~Hastie, R.~Tibshirani, J.~H. Friedman, and J.~H. Friedman.
\newblock \emph{The elements of statistical learning: data mining, inference,
  and prediction}, volume~2.
\newblock Springer, 2009.

\bibitem[Hastie et~al.(2022)Hastie, Montanari, Rosset, and
  Tibshirani]{hastie2022surprises}
T.~Hastie, A.~Montanari, S.~Rosset, and R.~J. Tibshirani.
\newblock Surprises in high-dimensional ridgeless least squares interpolation.
\newblock \emph{Annals of statistics}, 50\penalty0 (2):\penalty0 949, 2022.

\bibitem[Hazlett(2020)]{hazlett2020kernel}
C.~Hazlett.
\newblock Kernel balancing.
\newblock \emph{Statistica Sinica}, 30\penalty0 (3):\penalty0 1155--1189, 2020.

\bibitem[Hellerstein and Imbens(1999)]{hellerstein1999el}
J.~K. Hellerstein and G.~W. Imbens.
\newblock Imposing moment restrictions from auxiliary data by weighting.
\newblock \emph{Review of Economics and Statistics}, 81\penalty0 (1):\penalty0
  1--14, 1999.

\bibitem[Hill(2011)]{hill2011bayesian}
J.~L. Hill.
\newblock Bayesian nonparametric modeling for causal inference.
\newblock \emph{Journal of Computational and Graphical Statistics}, 20\penalty0
  (1):\penalty0 217--240, 2011.

\bibitem[Hirano et~al.(2003)Hirano, Imbens, and Ridder]{hirano2003efficient}
K.~Hirano, G.~W. Imbens, and G.~Ridder.
\newblock Efficient estimation of average treatment effects using the estimated
  propensity score.
\newblock \emph{Econometrica}, 71\penalty0 (4):\penalty0 1161--1189, 2003.

\bibitem[Hirshberg and Wager(2021)]{hirshberg2021augmented}
D.~A. Hirshberg and S.~Wager.
\newblock Augmented minimax linear estimation.
\newblock \emph{The Annals of Statistics}, 49\penalty0 (6):\penalty0
  3206--3227, 2021.

\bibitem[Hirshberg et~al.(2019)Hirshberg, Maleki, and
  Zubizarreta]{hirshberg2019minimax}
D.~A. Hirshberg, A.~Maleki, and J.~R. Zubizarreta.
\newblock Minimax linear estimation of the retargeted mean.
\newblock \emph{arXiv preprint arXiv:1901.10296}, 2019.

\bibitem[Imai and Ratkovic(2014)]{imai2014cbps}
K.~Imai and M.~Ratkovic.
\newblock Covariate balancing propensity score.
\newblock \emph{Journal of the Royal Statistical Society: Series B: Statistical
  Methodology}, pages 243--263, 2014.

\bibitem[Jacot et~al.(2018)Jacot, Gabriel, and Hongler]{jacot2018neural}
A.~Jacot, F.~Gabriel, and C.~Hongler.
\newblock Neural tangent kernel: Convergence and generalization in neural
  networks.
\newblock \emph{Advances in neural information processing systems}, 31, 2018.

\bibitem[Kallus(2020)]{kallus2020generalized}
N.~Kallus.
\newblock Generalized optimal matching methods for causal inference.
\newblock \emph{J. Mach. Learn. Res.}, 21:\penalty0 62--1, 2020.

\bibitem[Kallus et~al.(2021)Kallus, Mao, and Uehara]{kallus2021causal}
N.~Kallus, X.~Mao, and M.~Uehara.
\newblock Causal inference under unmeasured confounding with negative controls:
  A minimax learning approach.
\newblock \emph{arXiv preprint arXiv:2103.14029}, 2021.

\bibitem[Kennedy(2022)]{kennedy2022semiparametric}
E.~H. Kennedy.
\newblock Semiparametric doubly robust targeted double machine learning: a
  review.
\newblock \emph{arXiv preprint arXiv:2203.06469}, 2022.

\bibitem[Kim et~al.(2022{\natexlab{a}})Kim, Niknam, and
  Zubizarreta]{kim2022kernel}
K.~Kim, B.~A. Niknam, and J.~R. Zubizarreta.
\newblock Scalable kernel balancing weights in a nationwide observational study
  of hospital profit status and heart attack outcomes.
\newblock 2022{\natexlab{a}}.

\bibitem[Kim et~al.(2022{\natexlab{b}})Kim, Kern, Goldwasser, Kreuter, and
  Reingold]{kim2022universal}
M.~P. Kim, C.~Kern, S.~Goldwasser, F.~Kreuter, and O.~Reingold.
\newblock Universal adaptability: Target-independent inference that competes
  with propensity scoring.
\newblock \emph{Proceedings of the National Academy of Sciences}, 119\penalty0
  (4):\penalty0 e2108097119, 2022{\natexlab{b}}.

\bibitem[Kline(2011)]{kline2011oaxaca}
P.~Kline.
\newblock Oaxaca-blinder as a reweighting estimator.
\newblock \emph{American Economic Review}, 101\penalty0 (3):\penalty0 532--37,
  2011.

\bibitem[Kobak et~al.(2020)Kobak, Lomond, and Sanchez]{kobak2020optimal}
D.~Kobak, J.~Lomond, and B.~Sanchez.
\newblock The optimal ridge penalty for real-world high-dimensional data can be
  zero or negative due to the implicit ridge regularization.
\newblock \emph{The Journal of Machine Learning Research}, 21\penalty0
  (1):\penalty0 6863--6878, 2020.

\bibitem[Krishna et~al.(2009)Krishna, Bondell, and Ghosh]{krishna2009bayesian}
A.~Krishna, H.~D. Bondell, and S.~K. Ghosh.
\newblock Bayesian variable selection using an adaptive powered correlation
  prior.
\newblock \emph{Journal of statistical planning and inference}, 139\penalty0
  (8):\penalty0 2665--2674, 2009.

\bibitem[LaLonde(1986)]{lalonde1986evaluating}
R.~J. LaLonde.
\newblock Evaluating the econometric evaluations of training programs with
  experimental data.
\newblock \emph{The American economic review}, pages 604--620, 1986.

\bibitem[Lin et~al.(2022)Lin, Zhang, and Zhang]{lin2022reproducing}
R.~R. Lin, H.~Z. Zhang, and J.~Zhang.
\newblock On reproducing kernel banach spaces: Generic definitions and unified
  framework of constructions.
\newblock \emph{Acta Mathematica Sinica, English Series}, 38\penalty0
  (8):\penalty0 1459--1483, 2022.

\bibitem[Lin and Han(2022)]{lin2022regression}
Z.~Lin and F.~Han.
\newblock On regression-adjusted imputation estimators of the average treatment
  effect.
\newblock \emph{arXiv preprint arXiv:2212.05424}, 2022.

\bibitem[Lumley et~al.(2011)Lumley, Shaw, and Dai]{lumley2011connections}
T.~Lumley, P.~A. Shaw, and J.~Y. Dai.
\newblock Connections between survey calibration estimators and semiparametric
  models for incomplete data.
\newblock \emph{International Statistical Review}, 79\penalty0 (2):\penalty0
  200--220, 2011.

\bibitem[Menon and Ong(2016)]{menon2016linking}
A.~Menon and C.~S. Ong.
\newblock Linking losses for density ratio and class-probability estimation.
\newblock In \emph{International Conference on Machine Learning}, pages
  304--313. PMLR, 2016.

\bibitem[Moosavi et~al.(2023)Moosavi, H{\"a}ggstr{\"o}m, and
  de~Luna]{moosavi2023costs}
N.~Moosavi, J.~H{\"a}ggstr{\"o}m, and X.~de~Luna.
\newblock The costs and benefits of uniformly valid causal inference with
  high-dimensional nuisance parameters.
\newblock \emph{Statistical Science}, 38\penalty0 (1):\penalty0 1--12, 2023.

\bibitem[Mou et~al.(2023)Mou, Ding, Wainwright, and Bartlett]{mou2023kernel}
W.~Mou, P.~Ding, M.~J. Wainwright, and P.~L. Bartlett.
\newblock Kernel-based off-policy estimation without overlap: Instance
  optimality beyond semiparametric efficiency, 2023.
\newblock URL \url{https://arxiv.org/abs/2301.06240}.

\bibitem[Murray and Feller(2024)]{murray2024bayes_weights}
J.~Murray and A.~Feller.
\newblock Bayesian causal models from a weighting perspective: Balance, bias,
  and double robustness.
\newblock 2024.

\bibitem[Newey(1994)]{newey1994asymptotic}
W.~K. Newey.
\newblock The asymptotic variance of semiparametric estimators.
\newblock \emph{Econometrica: Journal of the Econometric Society}, pages
  1349--1382, 1994.

\bibitem[Newey and Robins(2018)]{newey2018cross}
W.~K. Newey and J.~R. Robins.
\newblock Cross-fitting and fast remainder rates for semiparametric estimation.
\newblock \emph{arXiv preprint arXiv:1801.09138}, 2018.

\bibitem[Newey and Smith(2004)]{newey2004el}
W.~K. Newey and R.~J. Smith.
\newblock Higher order properties of gmm and generalized empirical likelihood
  estimators.
\newblock \emph{Econometrica}, 72\penalty0 (1):\penalty0 219--255, 2004.

\bibitem[Newey et~al.(1998)Newey, Hsieh, and Robins]{newey1998undersmoothing}
W.~K. Newey, F.~Hsieh, and J.~Robins.
\newblock Undersmoothing and bias corrected functional estimation.
\newblock 1998.

\bibitem[Newey et~al.(2004)Newey, Hsieh, and Robins]{newey2004twicing}
W.~K. Newey, F.~Hsieh, and J.~M. Robins.
\newblock Twicing kernels and a small bias property of semiparametric
  estimators.
\newblock \emph{Econometrica}, 72\penalty0 (3):\penalty0 947--962, 2004.

\bibitem[Park et~al.(2009)Park, Lee, and Ha]{park2009l_2}
B.~Park, Y.~Lee, and S.~Ha.
\newblock $ l\_2 $ boosting in kernel regression.
\newblock \emph{Bernoulli}, 15\penalty0 (3):\penalty0 599--613, 2009.

\bibitem[Piironen and Vehtari(2017)]{piironen2017sparsity}
J.~Piironen and A.~Vehtari.
\newblock {Sparsity information and regularization in the horseshoe and other
  shrinkage priors}.
\newblock \emph{Electronic Journal of Statistics}, 11\penalty0 (2):\penalty0
  5018 -- 5051, 2017.
\newblock \doi{10.1214/17-EJS1337SI}.
\newblock URL \url{https://doi.org/10.1214/17-EJS1337SI}.

\bibitem[Polson and Scott(2010)]{polson2010global_local}
N.~G. Polson and J.~G. Scott.
\newblock Shrink globally, act locally: Sparse bayesian regularization and
  prediction.
\newblock \emph{Bayesian statistics}, 9\penalty0 (501-538):\penalty0 105, 2010.

\bibitem[Qin and Zhang(2007)]{qin2007el}
J.~Qin and B.~Zhang.
\newblock Empirical-likelihood-based inference in missing response problems and
  its application in observational studies.
\newblock \emph{Journal of the Royal Statistical Society: Series B (Statistical
  Methodology)}, 69\penalty0 (1):\penalty0 101--122, 2007.

\bibitem[Robins et~al.(2007)Robins, Sued, Lei-Gomez, and
  Rotnitzky]{robins2007comment}
J.~Robins, M.~Sued, Q.~Lei-Gomez, and A.~Rotnitzky.
\newblock Comment: Performance of double-robust estimators when" inverse
  probability" weights are highly variable.
\newblock \emph{Statistical Science}, 22\penalty0 (4):\penalty0 544--559, 2007.

\bibitem[Robins et~al.(2008)Robins, Li, Tchetgen, van~der Vaart,
  et~al.]{robins2008higher}
J.~Robins, L.~Li, E.~Tchetgen, A.~van~der Vaart, et~al.
\newblock Higher order influence functions and minimax estimation of nonlinear
  functionals.
\newblock In \emph{Probability and statistics: essays in honor of David A.
  Freedman}, volume~2, pages 335--422. Institute of Mathematical Statistics,
  2008.

\bibitem[Robins et~al.(1994)Robins, Rotnitzky, and Zhao]{robins1994estimation}
J.~M. Robins, A.~Rotnitzky, and L.~P. Zhao.
\newblock Estimation of regression coefficients when some regressors are not
  always observed.
\newblock \emph{Journal of the American statistical Association}, 89\penalty0
  (427):\penalty0 846--866, 1994.

\bibitem[Rotnitzky et~al.(2021)Rotnitzky, Smucler, and
  Robins]{rotnitzky2021characterization}
A.~Rotnitzky, E.~Smucler, and J.~M. Robins.
\newblock Characterization of parameters with a mixed bias property.
\newblock \emph{Biometrika}, 108\penalty0 (1):\penalty0 231--238, 2021.

\bibitem[Rubin(1980)]{rubin1980randomization}
D.~B. Rubin.
\newblock Randomization analysis of experimental data: The fisher randomization
  test comment.
\newblock \emph{Journal of the American statistical association}, 75\penalty0
  (371):\penalty0 591--593, 1980.

\bibitem[Rubinstein et~al.(2021)Rubinstein, Haviland, and
  Choi]{rubinstein2021balancing}
M.~Rubinstein, A.~Haviland, and D.~Choi.
\newblock Balancing weights for region-level analysis: the effect of medicaid
  expansion on the uninsurance rate among states that did not expand medicaid.
\newblock \emph{arXiv preprint arXiv:2105.02381}, 2021.

\bibitem[Shen et~al.(2022)Shen, Ding, Sekhon, and Yu]{shen2022panel}
D.~Shen, P.~Ding, J.~Sekhon, and B.~Yu.
\newblock A tale of two panel data regressions.
\newblock \emph{arXiv preprint arXiv:2207.14481}, 2022.

\bibitem[Shen(1997)]{shen1997methods}
X.~Shen.
\newblock On methods of sieves and penalization.
\newblock \emph{The Annals of Statistics}, 25\penalty0 (6):\penalty0
  2555--2591, 1997.

\bibitem[Singh(2021)]{singh2021debiased}
R.~Singh.
\newblock Debiased kernel methods.
\newblock \emph{arXiv preprint arXiv:2102.11076}, 2021.

\bibitem[Singh et~al.(2020)Singh, Xu, and Gretton]{singh2020kernel}
R.~Singh, L.~Xu, and A.~Gretton.
\newblock Kernel methods for causal functions: Dose, heterogeneous, and
  incremental response curves.
\newblock \emph{arXiv preprint arXiv:2010.04855}, 2020.

\bibitem[Singh et~al.(2022)Singh, Sun, et~al.]{singh2022late}
R.~Singh, L.~Sun, et~al.
\newblock Double robustness for complier parameters and a semiparametric test
  for complier characteristics.
\newblock Technical report, 2022.

\bibitem[Speckman(1979)]{speckman1979minimax}
P.~Speckman.
\newblock Minimax estimates of linear functionals in a hilbert space.
\newblock \emph{Unpublished Manuscript}, 1979.

\bibitem[Tan(2020)]{tan2020regularized}
Z.~Tan.
\newblock Regularized calibrated estimation of propensity scores with model
  misspecification and high-dimensional data.
\newblock \emph{Biometrika}, 107\penalty0 (1):\penalty0 137--158, 2020.

\bibitem[Tang et~al.(2023)Tang, Kong, Pan, and Wang]{tang2023ultra}
D.~Tang, D.~Kong, W.~Pan, and L.~Wang.
\newblock Ultra-high dimensional variable selection for doubly robust causal
  inference.
\newblock \emph{Biometrics}, 79\penalty0 (2):\penalty0 903--914, 2023.

\bibitem[Tchetgen~Tchetgen et~al.(2020)Tchetgen~Tchetgen, Ying, Cui, Shi, and
  Miao]{tchetgen2020introduction}
E.~J.~T. Tchetgen~Tchetgen, A.~Ying, Y.~Cui, X.~Shi, and W.~Miao.
\newblock An introduction to proximal causal learning.
\newblock \emph{arXiv preprint arXiv:2009.10982}, 2020.

\bibitem[Van Der~Laan and Rubin(2006)]{tmle}
M.~J. Van Der~Laan and D.~Rubin.
\newblock Targeted maximum likelihood learning.
\newblock \emph{The international journal of biostatistics}, 2\penalty0 (1),
  2006.

\bibitem[Van~der Laan et~al.(2011)Van~der Laan, Rose, et~al.]{van2011targeted}
M.~J. Van~der Laan, S.~Rose, et~al.
\newblock \emph{Targeted learning: causal inference for observational and
  experimental data}, volume~4.
\newblock Springer, 2011.

\bibitem[Vermeulen and Vansteelandt(2015)]{vermeulen2015bias}
K.~Vermeulen and S.~Vansteelandt.
\newblock Bias-reduced doubly robust estimation.
\newblock \emph{Journal of the American Statistical Association}, 110\penalty0
  (511):\penalty0 1024--1036, 2015.

\bibitem[Wang and Zubizarreta(2020)]{wang2020minimal}
Y.~Wang and J.~R. Zubizarreta.
\newblock Minimal dispersion approximately balancing weights: asymptotic
  properties and practical considerations.
\newblock \emph{Biometrika}, 107\penalty0 (1):\penalty0 93--105, 2020.

\bibitem[Wong and Chan(2018)]{wong2018kernel}
R.~K. Wong and K.~C.~G. Chan.
\newblock Kernel-based covariate functional balancing for observational
  studies.
\newblock \emph{Biometrika}, 105\penalty0 (1):\penalty0 199--213, 2018.

\bibitem[Zhao(2019)]{zhao2019covariate}
Q.~Zhao.
\newblock Covariate balancing propensity score by tailored loss functions.
\newblock \emph{The Annals of Statistics}, 47\penalty0 (2):\penalty0 965--993,
  2019.

\bibitem[Zhao and Percival(2017)]{zhao2017entropy}
Q.~Zhao and D.~Percival.
\newblock Entropy balancing is doubly robust.
\newblock \emph{Journal of Causal Inference}, 5\penalty0 (1), 2017.

\bibitem[Zubizarreta(2015)]{zubizarreta2015stable}
J.~R. Zubizarreta.
\newblock Stable weights that balance covariates for estimation with incomplete
  outcome data.
\newblock \emph{Journal of the American Statistical Association}, 110\penalty0
  (511):\penalty0 910--922, 2015.

\end{thebibliography}

\clearpage
\appendix
\renewcommand\thefigure{\thesection.\arabic{figure}}    
\renewcommand\thetable{\thesection.\arabic{table}}    
\setcounter{figure}{0}
\setcounter{table}{0}    

\tableofcontents

\section{Additional background and examples}
\label{sec:examples_appendix}

\subsection{Examples of general linear functionals via the Riesz representer}

\begin{example}[Counterfactual mean] \label{ex:cf_mean}
Let $Z = \{0,1\}$ and 
$
\psi(m) = \mathbb{E}[ m(X, 1)].
$
Under SUTVA and conditional ignorability, this estimand is equal to $E[Y(1)]$. The Riesz representer is the IPW, $\alpha(X,Z)=Z/e(X)$.  
\end{example}

\begin{example}[Average derivative] \label{ex:avg_deriv}
Let $Z \in \mathbb{R}$ and 
$
\psi(m) =  \mathbb{E}\left[\frac{\partial}{\partial z}m(X,Z)\right]. 
$
Under an appropriate generalization of SUTVA and conditional ignorability, this estimand corresponds to the average derivative effect of a continuous treatment. Under regularity conditions the Riesz representer is given by $\alpha(X,Z)=-\frac{\frac{d}{dz}p(Z|X)}{p(Z|X)}$ where $p(z|x)$ is the conditional density of $Z$ given $X$. 
\end{example}

\begin{example}[Distribution shift]  \label{ex:dist_shift}
Consider an example without $Z$, following the machine learning literature on covariate shift. 
Let $p$ denote the \emph{source} distribution of the observed data, and let $p^\ast$ over $\mathcal{X}$ denote the \emph{target} distribution. The estimand is then $\psi(m) = \int_\mathcal{X} m(x) dp^\ast(x)$. In a causal inference setting, this can recover the Average Treatment Effect on the Treated (ATT) under SUTVA and conditional ignorability; i.e., let $p$ be the distribution of covariates and outcomes for units assigned to control and $p^\ast$ be the distribution of the covariates for units assigned to treatment.
The Riesz representer is the density ratio, $\alpha(X) = \frac{dp^\ast}{dp}(X)$.
\end{example}

\begin{example}[Exact balancing weights] \label{ex:exact_bal_wts} The most common balancing weights estimation problem finds the minimum weights that exactly balance each element of $\Phi$.
In the constrained form, exact balancing solves
\begin{align}
    &  \min_{w\in\mathbb{R}^n} \Vert w \Vert_2^2\label{eq:exact_balance} \\
    &   \text{such that }  \tfrac{1}{n} w\phi_{pj} = \bar{\phi}_{qj} \quad\text{ for all } j\nonumber
\end{align}
\end{example}

\subsection{Examples of balancing weights}\label{apx:balancing-examples}

\begin{example}[$\ell_2$ balancing] \label{ex:l2_bal_wts}  The $\ell_2$ balancing weights problem is usually expressed via  its penalized form:
\begin{align}
    \min_{w\in\mathbb{R}^n} \Big\{ \Vert\tfrac{1}{n} w \Phi_p - \overbar{\Phi}_q \Vert^2_2 + \delta \Vert w\Vert_2^2 \Big\} \label{eq:l2_balance}.
\end{align}
The automatic form is a ridge-penalized regression for the Riesz representer.
\end{example}

\begin{example}[$\ell_\infty$ balancing]  \label{ex:sup_norm_bal_wts}
The constrained form of the $\ell_\infty$ balancing weights problem is
\begin{align}
    &  \min_{w\in\mathbb{R}^n} \Vert w \Vert_2^2\label{supnormproblem} \\
    &   \text{such that } \Vert\tfrac{1}{n} w \Phi_p - \overbar{\Phi}_q \Vert_\infty \leq \delta \nonumber
\end{align}
The automatic form is a lasso-penalized regression for the Riesz representer, sometimes known as the \emph{Minimum Distance Lasso} \citep{chernozhukov2022automatic}. 
\end{example}

\begin{example}[Kernel balancing] \label{ex:kernel_bal_wts}
    As a brief preview of the balancing problem in the infinite-dimensional setting, we provide an example where $\mathcal{F} = \mathcal{H}$ is a reproducing kernel Hilbert space on $\mathcal{X} \times \mathcal{Z}$ with norm $\Vert \cdot \Vert_\mathcal{H}$ and kernel $\mathcal{K} : (\mathcal{X} \times \mathcal{Z}) \times (\mathcal{X} \times \mathcal{Z}) \rightarrow \mathbb{R}$. Then for any $x_i \in \mathcal{X}, z_i \in \mathcal{Z}$, the representer $\phi(x_i,z_i) \coloneqq \mathcal{K}(x_i, z_i ,\cdot , \cdot) \in \mathcal{H}$. Using infinite-dimensional matrix notation, we denote $\Phi_p \in \mathcal{H}^n$ and $\overbar{\Phi}_q \in \mathcal{H}$ as above. The penalized balancing weights problem for $\mathcal{F} = \mathcal{H}$ is:
 \begin{align}
    \min_{w\in\mathbb{R}^n} \Big\{ \Vert \tfrac{1}{n} w \Phi_p - \overbar{\Phi}_q \Vert^2_\mathcal{H} + \delta \Vert w\Vert_2^2 \Big\}. \label{eq:kernel_balance}
\end{align}
 See \Cref{sec:high-dimensions} for details and references.
\end{example}

\subsection{Causal inference}

We now return to Example 1 above.
Here the goal is estimating the unobserved potential outcomes in an observational study. Let $Y$ be the potential outcome under control \citep[with appropriate restrictions, such as SUTVA;][]{rubin1980randomization}, let $p$ be the population of individuals in the control condition, and let $q$ be the population of individuals in the treatment condition. Then $Y$ is observed for population $p$ but not for population $q$, and the missing mean, $\mathbb{E}_q[Y]$, is the average potential outcome under control for the individuals who in fact were treated. Letting $Y$ be the potential outcome under treatment, $p$ the population of individuals in the treatment condition, and $q$ the population of individuals in the control condition, 
$\mathbb{E}_q[Y]$ is the average potential outcome under treatment for the individuals who in fact received control.

For both examples, the crucial assumption for identification is \emph{conditional ignorability}: the conditional distribution of $Y$ given $X$ is the same in the source and target populations. This is also known as ``conditional exchangeability,"  ``selection on observables," or ``no unmeasured confounding." For our purposes, we will require the mean, but not distributional, version of this assumption:
    $ \mathbb{E}_p[Y|X] = \mathbb{E}_q[Y|X].$ 

Since we assume the conditional expectations are the same in the two populations, we occasionally denote the common conditional mean functional without subscripts, $\mathbb{E}[Y|X]$. Under this assumption, we can identify $\mathbb{E}_q[Y]$ with the \emph{regression functional}, also known as the \emph{adjustment formula} or \emph{g-formula}:
\begin{align}\label{g-form}
     \mathbb E_q \left[ \mathbb{E}_p[Y|X] \right] = \mathbb E_q \left[ \mathbb{E}_q[Y|X] \right] = \mathbb{E}_q[Y].
\end{align}

A complementary approach instead relies on the density ratio between the marginal covariate distributions in the source and target populations, $\frac{dq}{dp}(X)$, also known as the Radon-Nikodym derivative, importance sampling weights, or inverse propensity score weights (IPW). \footnote{Using Bayes Rule, we can equivalently express $\frac{dq}{dp}(X)$ via the \emph{propensity score} $P(1_p|X)$, where $1_p$ is the indicator that an observation from the size-proportional mixture distribution of $p$ and $q$ is from population $p$: $\frac{dq}{dp}(X) = \frac{1-P(1_p | X)}{P(1_p | X)} \frac{P(1_p)}{1-P(1_p)}$} 
This is also a special case of a Riesz representer \citep{chernozhukov2022automatic}.
Under an additional \emph{population overlap assumption} that  $q(x)$ is absolutely continuous with respect to $p(x)$, we can identify $\mathbb{E}_q[Y]$ via the \emph{weighting functional}, also known as the \emph{IPW functional}:
\begin{align} \label{IPW}
    \mathbb{E}_p \left[ \frac{dq}{dp}(X) \hspace{0.1cm} Y \right] &= \mathbb{E}_p \left[ \frac{dq}{dp}(X)\hspace{0.1cm}\mathbb{E}_p[Y|X] \right] = \mathbb{E}_q[ \mathbb{E}_p[Y|X]] = \mathbb{E}_q[Y].
\end{align}

Finally, we can combine the regression and weighting functionals to create a third identifying functional, known as the \emph{doubly robust functional} \citep{robins1994estimation}:
\begin{align}
    \mathbb E_q \left[ \mathbb{E}_p[Y|X] \right] + \mathbb{E}_p\left[ \frac{dq}{dp} (X) \left\{Y - \mathbb{E}_p[Y|X]\right\}\right]. \label{eq:doubly_robust_functional_causal_appendix}
\end{align}
This functional has the attractive property of being equal to $\mathbb{E}_q[Y]$ even if either one of $\frac{dq}{dp} (X)$ or $\mathbb{E}_p[Y|X]$ is replaced with an arbitrary function of $X$, hence the term ``doubly robust."\footnote{This functional is equal to $\mathbb{E}_q[Y]$ if $\mathbb{E}_p[Y|X]$ is replaced with an arbitrary well-behaved functional of $X_p$, because the first and last terms cancel and we are left with the weighting functional $\mathbb{E}_p[\frac{dq}{dp}(X)Y]$. It is also equal to  $\mathbb{E}_q[Y]$ if $\frac{dq}{dp}(X)$ is replaced with an arbitrary well-behaved functional of $X_p$, because the $\mathbb{E}_p[ h (X) (Y - \mathbb{E}_p[Y|X])]$ is equal to $0$ for any $h$ and therefore we are left with the regression functional $\mathbb E_q \left[ \mathbb{E}_p[Y|X] \right]$.} 
See \citet{chernozhukov2018double, kennedy2022semiparametric} for recent overviews of the active literature in causal inference and machine learning focused on estimating versions of Equation \eqref{eq:doubly_robust_functional}.

The \emph{augmented} estimators that we analyze in this paper are based on estimating this doubly robust functional. These estimators \emph{augment} an estimator of the regression functional based on an outcome regression (or \textit{base learner}) with appropriately weighted residuals. Alternatively, they \emph{augment} an estimator of the weighting functional with an outcome regression-based estimator of the regression functional (subtracting off the implied estimator of $\frac{dq}{dp}(X) \mathbb E_p[Y|X]$).

Recall that under linearity the imbalance over all $f \in \mathcal{F}$ has a simple closed form. 
For any $f(x,z) = \theta ^\top \phi(x,z)  \in \mathcal{F}$, 
$
    \mathbb{E} [ h(X,Z,f) ] =  \theta^\top \mathbb{E} [h(X,Z,\phi) ],
$
where $h(X,Z,\phi)$ is short-hand for the vector with $j$th entry $ h(X,Z,\phi_j)$. 
We can then write the imbalance in terms of a transformed feature space $h(X,Z\phi)$, giving a closed form that we can readily calculate by applying the linear functional $\psi$ from (\ref{eq:general_linear_functional}) to the features $\phi$ :
\begin{align*}
    \text{Imbalance}_\mathcal{F}(w) &\coloneqq \sup_{f \in \mathcal{F}} \Big\{ \mathbb{E}[ w(X,Z) f(X,Z)] - \mathbb{E}[h(X,Z,f)] \Big\}\\
    &= \sup_{\Vert \theta \Vert \leq 1} \Big\{ \theta^\top\mathbb{E}[ w(X,Z) \phi(X,Z) ] - \theta^\top \mathbb{E} [h(X,Z,\phi)] \Big\}\\
    &=  \big\Vert \mathbb{E}[ w(X,Z) \phi(X,Z) ] - \mathbb{E} [h(X,Z,\phi)] \big\Vert_*.
\end{align*}

Consider the counterfactual mean estimand $\psi(m) = \mathbb{E}[m(X,0) | Z=1]$. We have
\[ \text{Imbalance}_\mathcal{F}(w) = \big\Vert \mathbb{E}[ w(X,Z) \phi(X,Z) ] - \mathbb{E} [\phi(X,0) | Z=1] \big\Vert_*. \]
For simplicity let $\phi(x,z) = x$. Now we get $\text{Imbalance}_\mathcal{F}(w) = | \mathbb{E}[ w(X,Z) X] - \mathbb{E}[X|Z=1] | $ and therefore, the balancing optimization problem finds weights $w$ that reweight the total mean $\mathbb{E}[X]$ to approximate the conditional mean $\mathbb{E}[X|Z=1]$.

\subsection{Equivalences of outcome regression and balancing weighting methods}\label{apx:equiv-conditions}

For the special case of $\ell_2$ kernel balancing, the balancing weights problem is numerically equivalent to directly estimating the conditional expectation $\mathbb{E}[Y_p | \Phi_p]$ via kernel ridge regression and applying the estimated coefficients to $\bar{\Phi}_q$. We begin with the special case of unregularized linear regression and then present the more general setting. We initially present the results assuming $d < n$ and that $\Phi_p$ has rank $d$, turning to the high-dimensional case with $d > n$ in Appendix \ref{sec:high-dimensions}.

\paragraph{Linear regression.} Ordinary least squares regression is equivalent to a weighting estimator that exactly balances the feature means. See \citet{fuller2002regression} for discussion in the survey sampling literature; see \citet{robins2007comment}, \citet{abadie2010synthetic}, \citet{kline2011oaxaca}, and \citet{chattopadhyay2020balancing} for relevant discussions in the causal inference literature.

In particular, let $\hat{w}_\text{exact}$ be the solution to the the exact balancing weights problem in Example \ref{ex:exact_bal_wts} in the main text.
Let $\hat{\beta}_\text{ols} =  (\Phi_p^\top\Phi_p)^{-1} \Phi_p^\top Y_p$ be the OLS coefficients from the regression of $Y_p$ on $\Phi_p$.
We then have the following numerical equivalence:
\begin{align} 
\hat{\mathbb{E}}[ \Phi_q \hat{\beta}_\text{ols} ] \quad &= \quad \hat{\mathbb{E}}[ \hat{w}_\text{exact} \circ Y_p ] \label{ols_equiv} \\
\overbar{\Phi}_q \underbrace{ (\Phi_p^\top\Phi_p)^{-1} \Phi_p^\top Y_p}_{\hat{\beta}_\text{ols}}  \quad &= \quad \underbrace{\overbar{\Phi}_q  (\Phi_p^\top\Phi_p)^{-1} \Phi_p^\top}_{ \tfrac{1}{n} \hat{w}_\text{exact}} Y_p ,\nonumber
\end{align}
where the weights have the closed form $\tfrac{1}{n} \hat{w}_\text{exact} = \overbar{\Phi}_q  (\Phi_p^\top\Phi_p)^{-1} \Phi_p^\top$.

\paragraph{Ridge regression.} This equivalence immediately extends to ridge regression \citep{speckman1979minimax,hirshberg2019minimax,kallus2020generalized}.\footnote{See \citet{harshaw2019balancing} for an interesting connection of this equivalence to experimental design. See \citet{ben2021augmented} and \citet{shen2022panel} for related applications in the panel data setting.}
Let $\hat{w}^\delta_{\ell_2}$ be the minimizer of the $\ell_2$ balancing weights problem in Example \ref{ex:l2_bal_wts} in the main text, with hyperparameter $\delta$. Let
\begin{align}\label{eq:ridge}
\hat{\beta}^\delta_{\text{ridge}} \coloneqq \underset{\beta \in \mathbb{R}^d}{\text{argmin}} \Big\{ \Vert Y_p - \Phi_p\beta \Vert_2^2 + \delta \Vert \beta \Vert_2^2\Big\}
\end{align}
be the ridge regression coefficients from least squares regression of $Y_p$ on $\Phi_p$. 
We then have the following numerical equivalence: 
\begin{align}
\hat{\mathbb{E}}[ \Phi_q \hat{\beta}^\delta_{\text{ridge}} ] \quad &= \quad \hat{\mathbb{E}}[ \hat{w}^\delta_{\ell_2} \circ Y_p ] \label{ridgeequiv} \\
\overbar{\Phi}_q \underbrace{( \Phi_p^\top\Phi_p + \delta I)^{-1} \Phi_p^\top Y_p}_{\hat{\beta}^\delta_{\text{ridge}}} \quad &= \quad \underbrace{\overbar{\Phi}_q ( \Phi_p^\top\Phi_p + \delta I)^{-1} \Phi_p^\top}_{ \tfrac{1}{n}\hat{w}^\delta_{\ell_2}} Y_p,\nonumber
\end{align}
where the weights have the closed form $\tfrac{1}{n} \hat{w}^\delta_{\ell_2} = \overbar{\Phi}_q ( \Phi_p^\top\Phi_p + \delta I)^{-1} \Phi_p^\top$.
Thus, the estimate from  ridge regression is identical to the estimate using the $\ell_2$ balancing weights. We leverage this equivalence in Section \ref{sec:newresults} below.

\paragraph{Kernel ridge regression.} In general, the same equivalence holds in the non-parametric setting where $\phi$ is the feature map induced by an RKHS. In particular, let $\mathcal{F} = \{ f \in \mathcal{H} : \Vert f \Vert_\mathcal{H} \leq r \}$, where $\mathcal{H}$ is a reproducing kernel Hilbert space (RKHS) on $\mathcal{X} \times \mathcal{Z}$ with kernel $\mathcal{K}$, $\Vert \cdot \Vert_\mathcal{H}$ denotes the norm of the RKHS, and $r > 0$. Then the equivalence above holds for $\phi(x,z) \coloneqq \mathcal{K}(x,z,\cdot,\cdot)$. Although  $\phi$ is typically infinite-dimensional, the  Riesz Representer Theorem  shows that the least squares regression and, equivalently, the balancing optimization problem have closed-form solutions. The least squares regression  approach is  \emph{kernel ridge regression} and the weighting estimator is \emph{kernel balancing weights} \citep[see][]{hazlett2020kernel, kim2022kernel}. 
\citet{hirshberg2019minimax} leverage this equivalence to analyze the asymptotic bias of kernel balancing weights. For further discussion of this equivalence  see \citet{gretton2012kernel, kallus2020generalized}.

\medskip
Finally, we briefly mention some additional papers that discuss relevant equivalences. In the context of panel data, \citet{shen2022panel} establish connections between different forms of regression, which is especially relevant for our discussion of high-dimensional features in Appendix \ref{sec:high-dimensions}.
In addition, \citet{lin2022regression} provide an interesting alternative perspective by demonstrating that a large class of outcome regression estimators can be viewed as implicitly estimating the density ratio of the covariate distributions in the two treatment groups. Our results generalize and unify many of these existing numeric equivalences.


\section{Details for when $d > n$}\label{sec:high-dimensions}

In this section, we extend our results to the high-dimensional setting. In all that follows we will assume that $d > n$ and we will assume $\Phi_p^\top$ has rank $n$.\footnote{Alternatively, we could follow \citet{bartlett2020benign} and assume that, almost surely, the projection of $\Phi_p$ on the space orthogonal to any eigenvector of $\mathbb{E}[\Phi_p\Phi_p^\top]$ spans a space of dimension $n$. But as our results are numerical this has no real advantage.} For $d = \infty$, we replace $\mathbb{R}^d$ with any infinite-dimensional Hilbert space $\mathcal{H}$ and we require the norm defining $\mathcal{F}$ to be the norm of the Hilbert space. In this case, it should be understood that $\Phi_p \in \mathcal{H}^n$.

\subsection{Balancing weights when $d > n$}

In the main text, recall that there are three equivalent versions of the balancing weights problem: the penalized, constrained, and automatic form with hyperparameters $\delta_1, \delta_2, \delta_3 \geq 0$ respectively. When $\Phi_p^\top \Phi_p$ is no longer invertible, a unique solution may fail to exist for certain values of these hyperparameters. We provide the relevant technical caveats here.

We begin by mentioning that for $\delta_1 > 0$, the penalized form of the balancing weights optimization problem is strictly convex, and therefore a unique solution exists, regardless of whether $d > n$. However, when $\delta_1 = 0$, there could potentially be infinite many solutions. In this setting, we choose the one with the minimum norm:
\begin{align}
    \min_{w \in \mathbb{R}^n} &\Vert w \Vert_2^2 \label{eq:min-norm-balancing}\\
    \text{such that } & \Vert w \Phi_p - \overbar{\Phi}_q \Vert_*^2 = \min_v \Vert v \Phi_p - \overbar{\Phi}_q \Vert_*^2.\nonumber
\end{align}

If we define $\delta_\text{min} \coloneqq \min_v \Vert v \Phi_p - \overbar{\Phi}_q \Vert_*^2$, we see that the minimum norm solution in \Cref{eq:min-norm-balancing} corresponds to a solution to the constrained form of balancing weights with $\delta_2 = \delta_\text{min}$. Importantly, no solution exists for $\delta_2 < \delta_\text{min}$, and we must make the additional restriction that $\delta_2 \geq \delta_\text{min}$. In particular, no solution exists for $\delta_2 = 0$ and we cannot achieve exact balance; that is, for all $w$, $w \Phi_p \neq \bar{\Phi}_q$.

As in the penalized form, the automatic form is strictly convex and a unique solution exists for $\delta_3 > 0$. When $\delta_3=0$ we choose the minimum norm solution: by duality this will be equivalent to the minimum norm solution to the penalized problem \citep[see][]{bruns2022outcome}.

Note that for $d = \infty$, each ``row'' of $\Phi_p$ is a vector in a Hilbert space $\mathcal{H}$. To solve the balancing weights problem computationally, we need a closed-form solution to the Hilbert space norm $\Vert \cdot \Vert_\mathcal{H}$. For example, this is a tractable computation when $\mathcal{H}$ is an RKHS.

\citet{singh2021debiased} gives the automatic form of this problem. See also \citet{wong2018kernel}, \citet{hazlett2020kernel}, and \citet{kallus2020generalized}.

\subsection{Equivalences from \Cref{apx:equiv-conditions} when $d > n$}\label{sec:equiv-d-bigger-n}

We now extend the equivalences from \Cref{apx:equiv-conditions} to the high-dimensional case.
Let $\delta \geq 0$ be the hyperparameter for the penalized form of balancing weights --- as we note above, this is important to state explicitly, since the constrained form will not have a solution for all values of its hyperparameter. For hyperparameter $\delta>0$, the solutions to $\ell_2$ balancing weights and ridge regression are identical as in \Cref{eq:ridge} with no alterations; ridge regression works by default when $d > n$. On the other hand, when $\delta = 0$, there exist infinitely many solutions to the normal equations that define the solution to the OLS optimization problem. Since $(\Phi_p^\top \Phi_p)$ is not invertible, \Cref{eq:ridge} does not apply directly. Instead, we introduce the minimum norm solution to OLS: 
\begin{align*}
    \min_{\beta \in \mathbb{R}^d} &\Vert \beta \Vert_2^2\\
    \text{such that } & \Vert \Phi_p \beta - Y_p \Vert_2^2 = \min_{\beta'} \Vert \Phi_p \beta' - Y_p \Vert_2^2.
\end{align*}
See \cite{bartlett2020benign} for an extensive discussion of this optimization problem and its statistical properties as an OLS estimator. For $d > n$, the minimum norm solution is:
\begin{align*}
    \hat{\beta}_\text{ols} \coloneqq (\Phi_p^\top \Phi_p)^\dag \Phi_p^\top Y_p = \Phi_p^\top(\Phi_p \Phi_p^\top)^{-1} Y_p,
\end{align*}
where $A^\dag$ denotes the pseudoinverse of a matrix $A$. Note that the definition holds in general.\footnote{For example, when $\Phi_p \in \mathcal{H}^n$ for an infinite-dimensional Hilbert space $\mathcal{H}$, $(\Phi_p^\top \Phi_p)^\dag$ is guaranteed to exist, since it is bounded and has closed range.} In the low-dimensional setting in the main text, $(\Phi_p^\top \Phi_p)$ is invertible, and so $(\Phi_p^\top \Phi_p)^\dag = (\Phi_p^\top \Phi_p)^{-1}$. The second equality holds only when $d > n$. 

A version of \Cref{ols_equiv} holds between the minimum norm $\ell_2$ balancing weights and minimum norm OLS estimators. Because the minimum norm $\ell_2$ balancing weights do not achieve exact balance, we change the notation from $\hat{w}_\text{exact}$ to $\hat{w}_{\ell_2}^0$. In this setting, $\Vert \cdot \Vert_* = \Vert \cdot \Vert_2$ and the minimum-norm balancing weights problem in \Cref{eq:min-norm-balancing} is also a minimum norm linear regression, but of $\bar{\Phi}_q \in \mathbb{R}^d$ on $\Phi_p^\top \in \mathbb{R}^{d \times n}$:
   \[ \hat{w}_{\ell_2}^0 =\Phi_p^\top (\Phi_p^\top \Phi_p)^\dag \bar{\Phi}_q = (\Phi_p \Phi_p^\top)^{-1} \Phi_p \bar{\Phi}_q . \]
Therefore, \Cref{ols_equiv} holds by replacing the inverse with the pseudo-inverse:
\begin{align*} 
\hat{\mathbb{E}}[ \Phi_q \hat{\beta}_\text{ols} ] \quad &= \quad \hat{\mathbb{E}}[ \hat{w}_{\ell_2}^0 \circ Y_p ] \\
\hat{\mathbb{E}}[ \Phi_q \underbrace{ (\Phi_p^\top\Phi_p)^\dag \Phi_p^\top Y_p}_{\hat{\beta}_\text{ols}} ] \quad &= \quad\hat{\mathbb{E}}[ \underbrace{\overbar{\Phi}_q  (\Phi_p^\top\Phi_p)^\dag \Phi_p^\top}_{ \hat{w}_{\ell_2}^0} \circ Y_p  ],\\
\hat{\mathbb{E}}[ \Phi_q \underbrace{ \Phi_p^\top (\Phi_p\Phi_p^\top)^{-1} Y_p}_{\hat{\beta}_\text{ols}} ] \quad &= \quad\hat{\mathbb{E}}[ \underbrace{\overbar{\Phi}_q \Phi_p^\top  (\Phi_p\Phi_p^\top)^{-1} }_{ \hat{w}_{\ell_2}^0} \circ Y_p  ].
\end{align*}

\subsection{Propositions \ref{prop:bal_wt_OLS} and \ref{generalregularizationpath} when $d > n$}\label{sec:regpath-d-bigger-n}

The results in Propositions \ref{prop:bal_wt_OLS} and \ref{generalregularizationpath} apply to the setting where  $d > n$ without any further alteration using the pseudo-inverse.
\begin{proof}[Proof of \Cref{prop:bal_wt_OLS}]
    \begin{align*}
        Y_p^\top \Phi_p  \hat{\theta}^\delta
        = Y_p^\top \Phi_p \Phi_p^\dag \Phi_p \hat{\theta}^\delta
        = Y_p^\top \Phi_p (\Phi_p^\top\Phi_p)^\dag \Phi_p^\top \Phi_p \hat{\theta}^\delta = \hat{\beta}_{\text{ols}} \hat{\Phi}_q, \\
    \end{align*}
    where the first two equalities follow from the pseudoinverse identities $A = A A^\dag A$ and $A^\dag = (A^\top A)^\dag A^\top$ for any matrix $A$. 
\end{proof}

Likewise \Cref{generalregularizationpath} holds exactly for $\hat{\beta}_\text{ols}$ defined with the pseudoinverse.

\subsection{The RKHS Setting}
\label{sec:rkhs_appendix}

The results for $d=\infty$ can be computed efficiently for reproducing kernel Hilbert spaces. For notational simplicity, we will consider the distribution shift setting from \Cref{ex:dist_shift}. Let $\mathcal{H}$ be a possibly-infinite-dimensional RKHS on $\mathcal{X}$ with kernel $\mathcal{K}$ and induced feature map via the representer theorem, $\phi : \mathcal{X} \rightarrow \mathcal{H}$ with $\phi(x) = \mathcal{K}(x, \cdot)$. Let $\Vert \cdot \Vert_\mathcal{H}$ denote the norm of $\mathcal{H}$. Let $K_p$ be the matrix with entries $\mathcal{K}(x_i, x_j)$, where $x_i, x_j \in \mathcal{X}$ are the $i$th and $j$th entries of $X_p$. Then $\Phi_p \Phi_p^\top = K_p$ is invertible. 

We will write out the versions of the main results for $\mathcal{F} = \mathcal{H}$ to demonstrate how to compute the corresponding results for RKHSs even though $d = \infty$. Denote the solution to the regularized least squares problem in $\mathcal{H}$ with $\lambda \geq 0$:
\begin{align*}
    \hat{f}^\delta &\coloneqq \text{argmin}_{f \in \mathcal{H}} \Vert f(X_p) - Y_p \Vert_2^2 + \lambda \Vert f \Vert_\mathcal{H}^2 .
\end{align*}
This is equivalent to the following problem by the representer theorem: 
\begin{align*}
    \hat{\beta}^\delta_\mathcal{H}  &\coloneqq \text{argmin}_{\beta \in \mathbb{R}^n} \Vert K \beta - Y_p \Vert_2^2 + \lambda \beta K \beta\\
    &= (K_p +\lambda I )^{-1} Y_p.
\end{align*}
Let $K_{x,p} \in \mathbb{R}^n$ be the row vector with entries $\mathcal{K}(x, x_i)$ where $x$ is an arbitrary element of $\mathcal{X}$ and $x_i$ is the $i$th entry of $X_p$. Then for any element $x \in \mathcal{X}$, $\hat{f}^\delta(x) = K_{x,p} \hat\beta^\delta_\mathcal{H}$. In particular, let define $K_{q,p}$ as the matrix with $(i,j)$th entry $\mathcal{K}(x_{qi}, x_{pj})$ where $x_{qi}$ is the $i$th sample from the target population and  $x_{qj}$ is the $j$th entry of $X_p$. Furthermore, define $\bar{K}_{q,p} \coloneqq \hat{\mathbb{E}}[K_{p,q}] \in \mathbb{R}^n$. Then, for any solution $\hat{w}^\delta_\mathcal{H}$ to the penalized form of balancing weights with function class $\mathcal{F}$ and hyperparameter $\delta \geq 0$: 
\begin{align}
    \hat{w}_\mathcal{H}^\delta = (K_p + \lambda I)^{-1} \bar{K}_{q,p}.
\end{align}
The proof follows from the closed-form of $\text{Imbalance}_\mathcal{H}(w)$, known as the Maximum Mean Discrepancy (MMD) \citep{gretton2012kernel}; see, e.g., \citet{hirshberg2019minimax, kallus2020generalized, bruns2022outcome}. 

With these preliminaries, we immediately have the following equivalence from \cite{hirshberg2019minimax}, which generalizes \Cref{apx:equiv-conditions} to the RKHS case:
\begin{align}
\hat{\mathbb{E}}[ K_{q,p} \hat{\beta}^\delta_\mathcal{H} ] \quad &= \quad \hat{\mathbb{E}}[ \hat{w}^\delta_\mathcal{H} \circ Y_p ] \\
\hat{\mathbb{E}}[ K_{q,p} \underbrace{(K_p + \delta I) Y_p}_{\hat{\beta}^\delta_\mathcal{H}} ] \quad &= \quad\hat{\mathbb{E}}[ \underbrace{\bar{K}_{q,p} (K_p + \delta I)^{-1} }_{ \hat{w}^\delta_{\mathcal{H}}} \circ Y_p ].\nonumber
\end{align}
Likewise, we have the following form for \Cref{prop:bal_wt_OLS}. Define $\hat{K}_{q,p} \coloneqq \hat{w}^{\delta T}_\mathcal{H} K_p$. Then, for any $\delta \geq 0$:
\begin{align*}
    \hat{\mathbb{E}}[\hat{w}^\delta_\mathcal{H} \circ Y_p] = \hat{\mathbb{E}} [ \hat{K}_{q,p} \hat{\beta}_\mathcal{H}^0].
\end{align*}
The resulting expression for \Cref{generalregularizationpath} is:
\begin{align*}
       & \hat{\mathbb{E}}[ \hat{w}_\mathcal{H}^\delta \circ Y_p] + \hat{\mathbb{E}}\left[\left( K_{q,p} - \hat{K}_{q,p}^\delta\right)\hat{\beta}_\mathcal{H}^\lambda\right] 
        = \hat{\mathbb{E}}[ K_{q,p} \hat{\beta}_\text{aug}],
    \end{align*}
    where the $j$th element of $\hat{\beta}_\text{aug}$ is:
    \begin{align*}
        \hat{\beta}_{\text{aug} ,j} &\coloneqq \left(1-a_j^\delta\right) \hat{\beta}_{\mathcal{H},j}^\lambda + a_j^\delta \hat{\beta}_{\mathcal{H},j}^0\\[0.5em]
        a_j^\delta &\coloneqq \frac{\widehat{\Delta}_{j}^\delta}{ \Delta_j },   
    \end{align*}
    where $\Delta_j = \bar{K}_{q,p,j} - \bar{K}_{p,j}$ and $\widehat{\Delta}^\delta_j = \hat{K}_{q,p,j}^\delta - \bar{K}_{p,j}$ with $\bar{K}_{p,j} \coloneqq \hat{\mathbb{E}}[K_p]$. 

Identical versions for the RKHS setting apply to \Cref{sec:L2}. These follow directly from the expressions above so we will omit repeating them explicitly. Importantly, equivalent versions for $\ell_\infty$ balancing in \Cref{sec:l8_aug_section} do \emph{not} follow immediately because an infinite dimensional vector space equipped with the $\ell_1$ norm does not form a Hilbert space. We conjecture that such extensions could be constructed using the Reproducing Kernel Banach Space literature \citep{lin2022reproducing}.

\section{Sample Splitting and Cross-Fitting}\label{sec:sample-split}

We briefly discuss the application of our numerical results in the setting with sample splitting. In standard balancing weights, we fit the weighting model $\theta^\delta$ using data $\Phi_p$ from population $p$ --- and then also apply the weighting model at $\Phi_p$. Cross-fitting breaks possible dependencies by only applying parameters on samples that were not used for estimation; this is a core technique in AutoDML \citep{chernozhukov2022automatic} and has been widely studied in the recent literature on doubly robust estimation \citep[see, for example][]{newey2018cross, kennedy2022semiparametric}. With cross-fitting, our numerical results only hold approximately, though we argue that the overall behavior of the estimators is qualitatively similar. 

To illustrate this, let the $n$ samples from population $p$ be split into $S$ partitions or ``splits'' and assume for simplicity that each split has size $n' \coloneqq n/S$. Denote the split $s$ covariates $\Phi_{p,s}$ and outcomes $Y_{p,s}$. Let $\Phi_{p,-s}$ and $Y_{p,-s}$ denote covariates and outcomes that are not in split $s$. As a simple example, consider cross-fit, unaugmented $\ell_2$ balancing weights with parameter $\delta = 0$ (i.e., OLS). For each split, we first solve the balancing problem \emph{out-of-sample} by solving for the coefficients as in the example in Equation (\ref{ols_equiv}):
\begin{align*}
    \hat{\theta}^0_{-s} &\coloneqq \text{argmin}_{\theta \in \mathbb{R}^{d}} \Big\{ \Vert \tfrac{1}{n-n'}\theta \Phi_{p,-s}^\top \Phi_{p,-s} - \bar{\Phi}_q \Vert^2_2 \Big\}\\ 
    &= (n-n') \overbar{\Phi}_q  (\Phi_{p,-s}^\top\Phi_{p,-s})^\dag.
\end{align*}
Note that we have re-written the balancing problem to be in terms of the coefficients $\theta$ instead of the weights $w = \theta \Phi_{p,-s}$, in order to emphasize that the goal is to apply this weighting model to the split $s$ samples to obtain weights $\hat{\theta}^0_{-s} \Phi_{p,s}^\top$. The split $s$ balancing weights estimate is then $\tfrac{1}{n'}\hat{\theta}^0_{-s} \Phi_{p,s}^\top Y_{p,s}$ and the final cross-fit estimator averages over these splits:
\[ \frac{1}{S} \sum_{s=1}^S \hat{\theta}^0_{-s} \Phi_{p,s}^\top Y_{p,s}. \]
Note that the coefficients $\hat{\theta}^0_{-s}$ enforce exact balance --- but only for data outside split $s$. In general, these weights will not achieve exact balance for split $s$. That is:
\begin{align*}
     \Vert \tfrac{1}{n-n'} \hat{\theta}^0_{-s} \Phi_{p,-s}^\top \Phi_{p,-s} - \bar{\Phi}_q \Vert_2 &= 0, \\
      \Vert \tfrac{1}{n'} \hat{\theta}^0_{-s} \Phi_{p,-s}^\top \Phi_{p,s} - \bar{\Phi}_q \Vert_2 &\neq 0.
\end{align*}
For an augmented estimator, we would also fit an outcome model using data from outside split $s$. For example, we could fit OLS:
\[ \hat{\beta}_{\text{ols},-s} \coloneqq (\Phi_{p,-s}^\top\Phi_{p,-s} )^\dag \Phi_{p,-s}^\top Y_{p,-s}. \]
The augmented estimator combining $\hat{\theta}^0_{-s}$ with $\hat{\beta}_{\text{ols},-s}$ would give:
\begin{align*}
    \frac{1}{S} \sum_{s=1}^S \Big(\bar{\Phi}_q \hat{\beta}_{\text{ols},-s}  + \tfrac{1}{n'}\hat{\theta}^0_{-s} \Phi_{p,s}^\top (Y_{p,s} - \Phi_{p,s} \hat{\beta}_{\text{ols},-s}) \Big).
\end{align*}

\subsection{\Cref{generalregularizationpath} with Sample Splitting}\label{sec:aug-sample-split}
For \Cref{generalregularizationpath} with sample splitting, within a single split the numerical result is identical, but the substantive point of interest is that we always shift coefficients toward the \emph{in-sample} OLS coefficients.

Let the coefficients $\hat{\beta}_{-s}^\lambda$ and $\hat{\theta}^\delta_{-s}$ be fixed vectors; in practice they will be models fit using samples not in $s$. Define $\hat{\Phi}_{q,s}^\delta \coloneqq \tfrac{1}{n'} \hat{\theta}^\delta_{-s} \Phi_{p,s}^\top \Phi_{p,s}$, and $\hat{\beta}_{\text{ols},s} \coloneqq (\Phi_{p,s}^\top\Phi_{p,s} )^\dag \Phi_{p,s}^\top Y_{p,s}$. 
Then the augmented estimator in the $s$th partition follows immediately from \Cref{generalregularizationpath} in the main text:
    \begin{align*}
        \hat{\mathbb{E}}[ \Phi_q \hat{\beta}_{\text{aug},s}],
    \end{align*}
    where the $j$th element of $\hat{\beta}_{\text{aug},s}$ is:
    \begin{align*}
        \hat{\beta}_{\text{aug},s ,j} &\coloneqq \left(1-a_{j,s}^\delta\right) \hat{\beta}_{-s,j}^\lambda + a_{j,s}^\delta \hat{\beta}_{\text{ols},s,j}\\[0.5em]
        a_{j,s}^\delta &\coloneqq \frac{\widehat{\Delta}_{j,s}^\delta}{ \Delta_{j,s} },  
    \end{align*}
    where $\Delta_{j,s} = \overbar{\Phi}_{q,j} - \overbar{\Phi}_{p,s,j}$ and $\widehat{\Delta}^\delta_{j,s} = \hat{\Phi}_{q,s,j}^\delta - \overbar{\Phi}_{p,j}$.

\subsection{Unregularized Outcome Model}

Whereas \Cref{generalregularizationpath} is unchanged by sample splitting, some of the equivalence results are affected by sample splitting. For example, we know from \cite{robins2007comment} that when the base learner is unregularized, i.e., $\hat{\beta}_\text{reg}^\lambda = \hat{\beta}_\text{ols}$, then the entire estimator collapses to OLS alone. With sample splitting, this is only true if $\hat{\beta}_\text{reg}^\lambda = \hat{\beta}_{\text{ols},s}$. With cross-fitting, however, the outcome model would typically be estimated using only data from outside split $s$. 

For example, consider $\hat{\beta}_{\text{ols},-s}$ introduced above. Plugging this into the result in \Cref{sec:aug-sample-split} yields:
\[  \hat{\beta}_{\text{aug} , s, j} \coloneqq \left(1-a_{j,s}^\delta\right) \hat{\beta}_{\text{ols},-s,j}^\lambda + a_{j,s}^\delta \hat{\beta}_{\text{ols},s,j}.  \]
When the OLS coefficients are fit out of sample, this prevents overfitting to the $\ell$th split. Augmentation shifts the out-of-split OLS coefficents back toward the in-split OLS coefficients. As the sample size in each split goes to infinity, then both $\hat{\beta}_{\text{ols},s}$ and $\hat{\beta}_{\text{ols},-s}$ converge to the same population OLS coefficients and the augmented coefficients converge to standard OLS for any weighting model $\theta^\delta \in \mathbb{R}^d$. 

\subsection{Unregularized Weight Model}\label{apx:cross-fit-unreg-weight}

Consider the opposite case where $\hat{\beta}_\text{reg}^\lambda$ is arbitrary and the weight model $\theta^{0}_{-s}$ achieves exact balance between $\Phi_{p,-s}$ and $\Phi_q$, as defined above. Then, as suggested in \Cref{sec:aug-sample-split}, $\hat{\Phi}^0_{q,s} \coloneqq \hat{\theta}^0_{-s} \Phi_{p,s}^\top \Phi_{p,s} \neq \bar{\Phi}_q$ in general. Instead, we only have an approximation:
\begin{align*}
    a_{j,s}^\delta &\coloneqq \frac{\hat{\Phi}^0_{q,s,j} - \overbar{\Phi}_{p,s,j}}{  \overbar{\Phi}_{q,j} - \overbar{\Phi}_{p,s,j} } \approx 1,
\end{align*}
where the approximation becomes equality as the sample size in each split goes to infinity. As a result, $ \hat{\beta}_{\text{aug} , s} \approx \hat{\beta}_{\text{ols}, s}$, 
the in-split OLS coefficients, where again these coefficients are equal as the sample size in each split goes to infinity.

\subsection{``Double Ridge''}

Similarly, whereas \Cref{prop:double_ridge} reduced to a single ridge outcome model, with sampling splitting we instead obtain an affine combination of in-split and out-of-split ridge regressions. Let $\hat{\beta}^\lambda_{-s}$ and $\hat{\theta}^\delta_{-s}$ denote ridge and $\ell_2$ balancing coefficients respectively fit outside of the $\ell$th split. For notational simplicity, assume that $(\Phi^\top_{p,s} \Phi_{p,s}) = \text{diag}(\sigma^2_{1,s}, ..., \sigma^2_{d,s})$ and similarly for $-s$. Then the augmented estimator in the $s$th split equals $\hat{\mathbb{E}}[ \Phi_q \hat{\beta}_{\text{aug},s} ]$, where
\begin{align*}
    &\hat{\beta}_{\text{aug},s,j} \coloneqq \left( \frac{\sigma^2_{j,-s} - \sigma^2_{j,s} + \delta }{\sigma^2_{j,-s} + \delta} \right) \hat{\beta}^\lambda_{-s,j} + \left( \frac{\sigma^2_{j,s} + \delta}{\sigma^2_{j,-s} + \delta}\right) \hat{\beta}^\delta_{s,j}.
\end{align*}

\section{Nonlinear weights}
\label{sec:nonlinear_weights}

Our main results focus on \emph{linear} balancing weights, where $\hat{w} = \hat{\theta} \Phi_p^\top$. A rich tradition in survey statistics \citep[e.g.,][]{deville1992calibration}, machine learning \citep[e.g.,][]{menon2016linking}, and causal inference \citep[e.g.,][]{zhao2019covariate} focuses instead on \emph{non-linear} balancing weights, such as when the weights correspond to a specific \emph{link function} $g(\cdot)$ applied to the linear predictor, $\hat{w} = g(\hat{\theta} \Phi_p^\top)$, or, equivalently, when the balancing weights problem penalizes an alternative dispersion penalty. Other balancing weights estimators are inherently nonlinear, for instance with weights estimated via a neural network \citep{chernozhukov2022riesznet}. 

In this section, we extend our results to nonlinear weights. We first consider general, arbitrary weights. We then consider the special case of minimum-variance weights that are constrained to be non-negative, which includes many important examples \citep{zubizarreta2015stable, athey2018approximate, abadie2010synthetic}. For this special case, we show that the non-negativity constraint corresponds to sample trimming.

Before turning to these results, we briefly review several important examples of nonlinear weights.

\begin{itemize}

    \item \textbf{Entropy balancing and alternative dispersion measures.} There is a large literature exploring alternative dispersion penalties for balancing weights and other calibrated estimators; see \citet{deville1992calibration} for a discussion in survey sampling and \citet{zhao2019covariate} for a discussion in causal inference. 
    The most common alternative to the variance penalty is \emph{entropy balancing} \citep{hainmueller2012entropy}, which has a dual representation as linear weights passed through an exponential link function, $\hat{w} = \text{exp}(\hat{\theta}\Phi_p^\top)$. For estimands like the Average Treatment Effect on the Treated, this implies a (calibrated) logistic regression model for the propensity score:
$$
 \hat{w} = \text{exp}(\hat{\theta} \Phi_p^\top) = \text{logit}^{-1}(\hat{\theta} \Phi_p^\top)  \big/ \left(1 - \text{logit}^{-1}(\hat{\theta} \Phi_p^\top) \right).
$$
    See \citet{tan2020regularized} for further discussion.

    \item \textbf{Traditional IPW via maximum likelihood estimation.} Importantly, the form above is purely numeric and does not depend on the particular estimation procedure for the dual parameters, $\theta$. Thus, our results below also apply to the common practice of IPW with a maximum likelihood (or other) estimate of a logistic regression propensity score model, as well as to AIPW and \emph{Double Machine Learning} estimators \citep{chernozhukov2018double}.

    \item \textbf{Generalized Empirical Likelihood.} The balancing weights literature is closely related to the empirical likelihood approach, which also finds weights over units to estimate an outcome model. Examples of (generalized) empirical likelihood in similar settings include \citet{hellerstein1999el}, \citet{qin2007el}, \citet{newey2004el}, \citet{graham2012inverse}, and  \citet{duchi2021el}. See \citet[Sec 3]{hirano2003efficient} for a didactic example connecting inverse propensity score weighting, outcome modeling, and the empirical likelihood approach.
    Standard Empirical Likelihood estimators penalize the log of the weights, $\text{log}(w)$. Many other dispersion choices are common; see, for example, \citet{newey2004el}. 
    
\end{itemize}

Our results also apply to weighting methods that do not have this same structure, such as Covariate Balancing Propensity Scores \citep{imai2014cbps}. Finally, the implied weights could also arise from an outcome model that can be represented as a linear smoother, such as a random forest; see, for example, \citet{lin2022reproducing}.

\subsection{General nonlinear weights}

We now show that \Cref{prop:bal_wt_OLS} holds approximately for arbitrary, non-linear weights.

\begin{proposition} \label{prop:bal_wt_OLS_nonlin}
Let $\hat{w} \in \mathbb{R}^n$, be any arbitrary weights. Denote the corresponding weighted features $\hat{\Phi}_q \coloneqq \tfrac{1}{n} \hat{w} \Phi_p$. 
Let $\hat{\beta}_{\text{ols}} = ( \Phi_p^\top\Phi_p)^\dag \Phi_p^\top Y_p$ be the OLS coefficients of the regression of $Y_p$ on $\Phi_p$. Then for any $\eta \in \mathbb{R}^d$, $c \in \mathbb{R}$: 
\begin{align*}
\hat{\mathbb{E}}\left[\hat{w} \circ Y_p\right] &=\hat{\Phi}_q \hat{\beta}_{\text{ols}}  + \underbrace{(\hat{w} - \Phi_p \eta - c)^\top (Y_p - \Phi_p \hat{\beta}_\text{ols})}_{\text{\emph{approximation error}}}.
\end{align*}
\end{proposition}
Importantly, \Cref{prop:bal_wt_OLS_nonlin} holds for arbitrary weights and, as we discuss above, does not require estimating the weights via a balancing weights optimization problem. 

Because \Cref{prop:bal_wt_OLS_nonlin} holds for any $\eta$ and $c$, we have the following bound on the approximation error:
\begin{align}
    \big| \hat{\mathbb{E}}\left[\hat{w} \circ Y_p\right] - \hat{\Phi}_q \hat{\beta}_{\text{ols}}  \big| \leq \min_{c, \eta} \big\Vert \hat{w} - \Phi_p \eta - c \big\Vert_2 \big\Vert Y_p - \Phi_p \hat{\beta}_\text{ols} \big\Vert_2,
\end{align}
where we could replace the $2$-norms with any valid H{\"o}lder inequality pair. 

The approximation error will necessarily depend on the specific data and target functional. We can nonetheless offer some remarks. 
First, when $d > n$, we have $\Vert Y_p - \Phi_p \hat{\beta}_\text{ols} \Vert_2 = 0$, and we recover \Cref{prop:bal_wt_OLS} exactly. Thus, in the high dimensional setting, \Cref{generalregularizationpath} holds exactly for arbitrary non-linear weights. For $n > d$, our bound on the approximation error will generally be non-zero and depends on the size of the deviations of $\hat{w}$ from its linear projection --- a natural measure of the ``non-linearity'' of the weights. For a given weight vector, we can compute this measure by first fitting an unpenalized linear regression of the weights on $\Phi_p$ and then taking the norm of the residuals. By a simple Taylor approximation argument, this value will be small when $\Phi_p$ is tightly concentrated around its mean.

Finally, following the argument in Section \ref{sec:newresults}, we can extend this proposition to augmented estimators with nonlinear weights: the augmented estimator equals $\hat{\mathbb{E}}[ \Phi_q \hat{\beta}_\text{aug}]$ (just as before) plus the approximation error in \Cref{prop:bal_wt_OLS_nonlin}. Therefore if \emph{nonlinear} weights $\hat{w}$ give exact balance, i.e., if $\hat{\Phi}_q = \bar{\Phi}_q$, then $\hat{\beta}_\text{aug} = \hat{\beta}_\text{ols}$, and the final point estimate differs from OLS by the approximation error in \Cref{prop:bal_wt_OLS_nonlin}.

\subsubsection{Illustration: LaLonde}

As an illustration, we compute both entropy balancing weights and logistic regression IPW weights for the ``short'' LaLonde dataset with $d = 11$ features. When targeting the Average Treatment Effect on the Treated, both weights have the same exponential link form.

We use a (cross-validated) lasso-penalized regression for the outcome model, and compute the corresponding augmented estimators, $\hat{\psi} \coloneqq \bar{\Phi}_q^\top \hat{\beta}_\text{reg} + \hat{w}^\top (Y_p - \Phi_p \hat{\beta}_\text{reg})$, and augmented coefficients $\hat{\beta}_\text{aug}$ as defined in \Cref{generalregularizationpath}. We compare the augmented estimators to the plug-in OLS estimate via a triangle-inequality-type decomposition:
\begin{align*}
    \text{Difference to OLS plug-in: }& | \bar{\Phi}_q^\top \hat{\beta}_\text{ols} - \hat{\psi} |\\
    \text{Difference due to weight non-linearity: }& |\bar{\Phi}_q^\top \hat{\beta}_\text{aug} - \hat{\psi} |\\
    \text{Difference due to imbalance: }& | \bar{\Phi}_q^\top \hat{\beta}_\text{ols} - \bar{\Phi}_q^\top \hat{\beta}_\text{aug} |
\end{align*}

We plot these metrics as a function on the entropy balancing hyperparameter and an $\ell_1$ penalty for the propensity score model in \Cref{fig:ols-nonlin-approx}.

\begin{figure}[tb]
    \centering
    \begin{subfigure}{0.425\textwidth}
        \centering 
        \includegraphics[width=\textwidth]{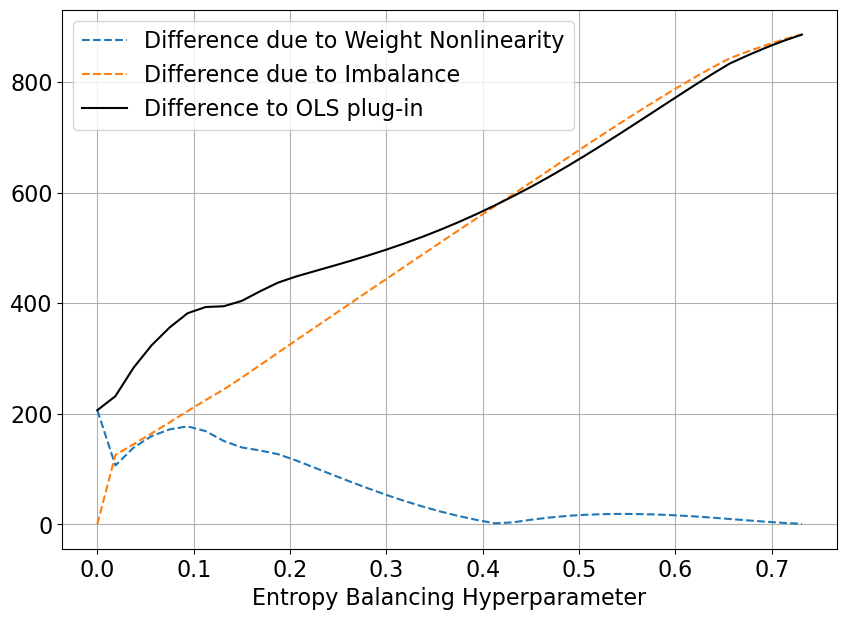}
        \caption{Entropy Balancing}
        \label{fig:ols-nonlin-approx-panel-a}
    \end{subfigure}
        \begin{subfigure}{0.425\textwidth}
        \centering 
        \includegraphics[width=\textwidth]{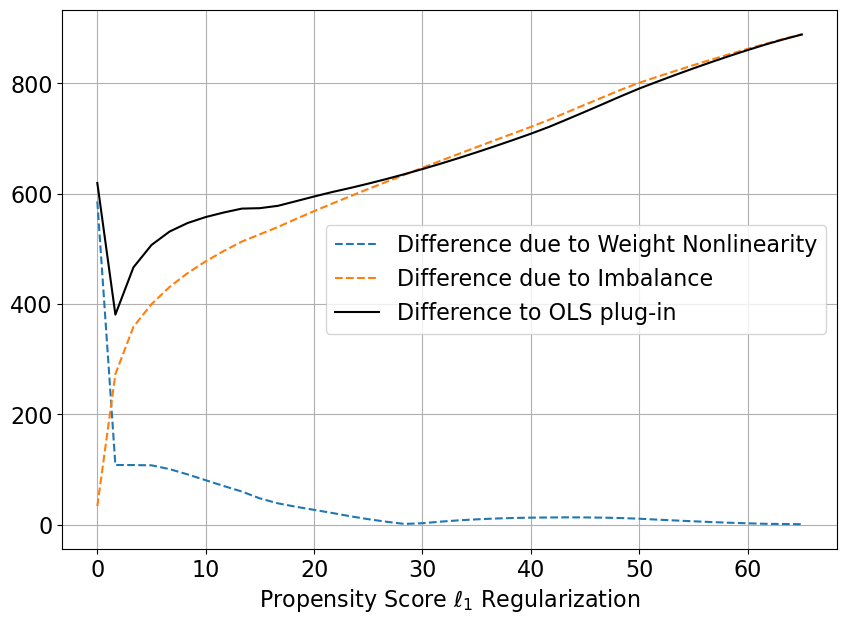}
        \caption{Traditional IPW}
        \label{fig:ols-nonlin-approx-panel-b}
    \end{subfigure}
    \caption{Decomposing the estimate from nonlinear augmented balancing weights for the ``short'' \citet{lalonde1986evaluating} example.}
    \label{fig:ols-nonlin-approx}
\end{figure}

Figure \ref{fig:ols-nonlin-approx} shows that there are two countervailing features at play here. First, as the weights become linear, $\hat{\psi} \rightarrow \bar{\Phi}_q \hat{\beta}_\text{aug}$ with the difference quantified by \Cref{prop:bal_wt_OLS_nonlin}. This occurs as the regularization strength increases because regularization pushes the weights to be uniform. Second, as the imbalance goes to zero, $\hat{\beta}_\text{aug} \rightarrow \hat{\beta}_\text{ols}$. For entropy balancing, this happens exactly as the hyperparamater goes to zero. The IPW weights typically yield substantially larger imbalance, and this imbalance does not vanish as the regularization parameter goes to zero.

\subsection{Constraining weights to be non-negative}
\label{sec:simplex}

A common modification of the (minimum variance) balancing weights problem is to constrain the estimated weights to be non-negative.\footnote{Recall that we already impose the constraint that the weights sum to 1. So imposing non-negativity is equivalent to constraining the weights to be on the simplex.} 
Such weights have a number of attractive practical properties: they limit extrapolation; they ensure that the final weighting estimator is sample bounded; and they are typically sparse, which can sometimes aid interpretability \citep{robins2007comment, abadie2010synthetic}.
Examples of constrained $\ell_\infty$ weights include Stable Balancing Weights and extensions \citep{zubizarreta2015stable, athey2018approximate, wang2020minimal}; examples of constrained $\ell_2$ weights include popular variants of the synthetic control method \citep{abadie2010synthetic, ben2021augmented}.

Using the dual form of the problem, \citet{ben2021balancing} show that (minimum variance) linear balancing weights with a non-negativity constraint have the form $\hat{w} = \{\Phi_p \hat{\theta}^\delta \}_+$, where $\{x\}_+ = \text{max}(x, 0)$ and where the coefficients $\hat{\theta}^\delta$ are generally different from the corresponding coefficients in the unconstrained model.\footnote{This ``positive part link'' representation is unique to the minimum variance weights. Other dispersion penalties, such as the entropy of the weights, imply a different link function. See \citet{zhao2019covariate} and \citet{ben2021balancing}.}
We can apply this insight to extend Proposition \ref{prop:bal_wt_OLS} to non-negative weights.

\begin{proposition} \label{prop:bal_wt_OLS_nonneg}
Let $\hat{w}^\delta_+ \coloneqq  \{ \Phi_p \hat{\theta}^\delta \}_+$, with $\hat{\theta}^\delta \in \mathbb{R}^d$ and where $\{x\}_+ = \text{max}(x, 0)$, be any linear balancing weights with a non-negativity constraint, with corresponding weighted covariates $\hat{\Phi}_q^\delta \coloneqq \hat{w}_+ \Phi_p$.  Let $\Phi_{p+},$ and $Y_{p+}$ denote the respective quantities restricted to those data points where $\hat{w}^\delta_+ >0$.  Let $\hat{\beta}_{\text{ols}}^+ \coloneqq (\Phi_{p+}^\top\Phi_{p+})^{-1} \Phi_{p+}^\top Y_{p+}$ be the OLS coefficients of the regression of $Y_{p+}$ on $\Phi_{p+}$. Then: 
$$\hat{\mathbb{E}}\left[\hat{w}_+^\delta \circ Y_p\right] =  \hat{\Phi}_q^\delta \hat{\beta}_{\text{ols}}^+.$$
\end{proposition}
So with the non-negativity constraint, we recover \Cref{prop:bal_wt_OLS} almost exactly, except that the usual OLS coefficients from the entire population $p$, $\hat\beta_{\text{ols}}$, are replaced with the OLS coefficients from only those units with positive weight, $\hat{w}^\delta_+ > 0$, $\hat\beta_{\text{ols}}^+$. 
We can therefore view the non-negativity constraint as a form of sample trimming. 
In particular, we can think of the data points where $\hat{w}^\delta_+ = 0$ as a set of outliers --- too dissimilar from the target population $\Phi_q$ --- that we trim before applying OLS. But the key is that the definition of ``outlier'' depends the choice of $\delta$ and $\mathcal{F}$, and even the target covariates $\Phi_q$: in general, changing $\delta$, $\mathcal{F}$, or $\Phi_q$ will change the set where $\hat{w}^\delta_+ > 0$.
Defining and characterizing how this set changes is a promising avenue for future research.
Finally, note that it is not generally the case that re-estimating the balancing weights problem after sample trimming will lead to non-negative weights (though this does hold for OLS). 

\Cref{prop:bal_wt_OLS_nonneg} simplifies further when the balancing weights achieve exact balance; see \citet[][Proposition 10]{rubinstein2021balancing} for a discussion of this special case. When $\hat{\Phi}_q = \Phi_q$, the balancing weight estimator with the non-negativity constraint is equivalent to
$ \hat{\mathbb{E}} [\Phi_q \hat{\beta}_{\text{ols}}^+ ]$.
Thus, linear balancing weights with a simplex constraint is equivalent to trimming the control group and applying standard OLS (note that the trimming does not affect the target covariate profile, $\Phi_q$). 

Finally, we can extend the results in  \Cref{generalregularizationpath} for augmented balancing weights to incorporate a non-negativity constraint. Many popular augmented balancing weights estimators have the form of \Cref{generalregularizationpath}, including \citet{athey2018approximate} and \citet{ben2021augmented}. Understanding the implications of this connection is an interesting direction for future work.

\subsection{Non-Linear Differentiable Outcome Models}
\label{sec:beyond_linear_appendix}

In this section, we provide a very preliminary sketch of how the results might be extended to the case where $\mathcal{F}$ is non-linear but still differentiable in its parameters.

Let $\mathcal{F} = \{ f(X, \theta) : \theta \in \mathbb{R}^d, \nabla_\theta f(X,\theta) \text{ exists} \}$. Then just like \Cref{prop:bal_wt_OLS} relates any $w$ that are linear in $X$ to the OLS coefficients, we can relate any $w \in \mathcal{F}$ to the least squares regressor in the function class $\mathcal{F}$.

First, let $\theta_\text{LS}$ be the unregularized least squares regressor (where we choose the least norm $\theta$ to break ties):
\[ \theta_\text{LS} \coloneqq \min_\theta \Vert Y_p - f(\theta, X_p) \Vert_2^2 \]

We have the first-order condition:

\[ \nabla_\theta f(\theta_\text{LS},X_p)^\top (Y_p - f(\theta_\text{LS},X_p)) = 0 \]

Now we can get a version of Proposition 2.1 for $w \in \mathcal{F}$ by considering the following taylor expansion:

\[ w(X_p) \coloneqq f(\theta_w, X_p) \approx f(\theta_\text{LS}, X_p) + \nabla_\theta f(\theta_\text{LS},X_p)(\theta_w - \theta_\text{LS}) \]

In which case, applying the first-order condition above, we get:

\[ w(X_p)^\top y \approx \underbrace{w(X_p)^\top f(\theta_\text{LS}, X_p)}_{\text{identical to 3.1}} + \underbrace{f(\theta_\text{LS}, X_p)^\top(Y_p - f(\theta_\text{LS}, X_p))}_{\text{this term is zero in linear case}}. \]


\section{Results with Correlated Features}\label{apx:correlated-features}

\subsection{Overview}

Throughout the main text, we have made the assumption that $\Phi_p^\top\Phi_p$ is a diagonal matrix. This is a useful simplifying assumption for illustrating the intuition behind our results. Our central result, \Cref{generalregularizationpath}, holds for arbitrary $\Phi_p^\top \Phi_p$. However, when $\Phi_p^\top \Phi_p$ is not diagonal, i.e. the covariates are correlated, there are some important additional subtleties for interpreting the results.

In \Cref{apx:l2-rot-wlog}, we show that the results in the main text for $\ell_2$ balancing weights still hold for non-diagonal $\Phi_p^\top \Phi_p$ by applying an additional rotation step. By constrast, in \Cref{apx:linf-aug-correlated}, we demonstrate that diagonal $\Phi_p^\top \Phi_p$ is \emph{not} without loss of generality for $\ell_\infty$ balancing. Not only do the weights not have a closed form, but the resulting augmented coefficients will not be sparse.

\subsection{Augmented $\ell_2$ Balancing Weights with Correlated Features} \label{apx:l2-rot-wlog}

To begin, we provide the following characterization of $\ell_2$ augmenting with correlated features \emph{directly}, that is, without using \Cref{generalregularizationpath}:

\begin{proposition}\label{prop:corrolated-l2-augment}
        Let $\hat{w}_{\ell_2}^\delta$ denote the solution to the penalized linear form of balancing weights with parameter $\delta$ and $\mathcal{F} = \{ f(x) = \theta^\top \phi(x) : \Vert \theta \Vert_2 \leq r \}$. Then,
    \[ \hat{w}_{\ell_2}^\delta = \Phi_p (\Phi_p^\top\Phi_p + \delta I)^{-1} \overbar{\Phi}_q. \]
     Therefore, the augmented $\ell_2$ balancing weights estimator with outcome model $\hat{\beta}_\text{reg} \in \mathbb{R}^d$ has the form,
    \begin{align}
        &\hat{\mathbb{E}}[\Phi_q \hat{\beta}_\text{reg}] + \hat{\mathbb{E}}[ \hat{w}_{\ell_2}^\delta (Y_p - \Phi_p \hat{\beta}_\text{reg})] =  \hat{\mathbb{E}}[\Phi_q \hat{\beta}_{\ell_2}],\nonumber\\
        &\hat{\beta}_{\ell_2} \coloneqq (I - A_\delta) \hat{\beta}_\text{reg} + A_\delta \hat{\beta}_\text{ols}\label{eq:corr-l2}\\
        &A_\delta \coloneqq (\Phi_p^\top\Phi_p + \delta I)^{-1}(\Phi_p^\top\Phi_p).\nonumber
    \end{align}
    For comparison, ridge regression with parameter $\delta$ has closed-form coefficients:
    \[ \hat{\beta}_\text{ridge}^\delta = (\Phi_p^\top\Phi_p + \delta I)^{-1}\Phi_p^\top Y_p = A_\delta \beta_\text{ols}. \]
\end{proposition}

This is a generalization of \Cref{l2augment} to the correlated setting. The result follows by simply plugging in the closed-form of $\hat{w}^\delta_{\ell_2}$ into the augmented estimator. We will now show that the resulting $\hat{\beta}_{\ell_2}$ is equivalent to applying \Cref{generalregularizationpath} to rotated versions of the outcome and weighting coefficients. 

First, we define the rotated version of \Cref{generalregularizationpath}:

\begin{proposition}\label{prop:beta-aug-rotated}
    Let $\Phi_p^\top\Phi_p$ be arbitrary, with eigendecomposition $\Phi_p^\top\Phi_p = VD^2V^\top$ where $D^2$ is a diagonal matrix with $j$th entry $\sigma_j^2$. For any $\hat{\beta}_\text{reg}^\lambda \in \mathbb{R}^d$, and any linear balancing weights estimator with estimated coefficients $ \hat{\theta}^\delta \in \mathbb{R}^d$, and with $\hat{w}^\delta \coloneqq \hat{\theta}^\delta \Phi_p^\top$ and $\hat{\Phi}_q^\delta \coloneqq \hat{w}^\delta \Phi_p$, the resulting augmented estimator
    \begin{align*}
       & \hat{\mathbb{E}}[ \hat{w}^\delta \circ Y_p] + \hat{\mathbb{E}}\left[\left(\Phi_q - \hat{\Phi}_q^\delta\right)\hat{\beta}_\text{reg}^\lambda\right] \\
        &= \hat{\mathbb{E}}[ \Phi_q  V \hat{\beta}_\text{aug}^\text{rot}],
    \end{align*}
    where the $j$th element of $\hat{\beta}_\text{aug}^\text{rot}$ is:
    \begin{align*}
        \hat{\beta}_{\text{aug} ,j}^\text{rot} &\coloneqq \left(1-a_j^\text{rot}\right) \hat{\beta}_{\text{reg},j}^\text{rot} + a_j^\text{rot} \hat{\beta}_{\text{ols},j}^\text{rot}\\[0.5em]
        a_j^\text{rot} &\coloneqq \frac{\widehat{\Delta}_{j}^\text{rot}}{ \Delta_j^\text{rot} }, 
    \end{align*}
    defined in terms of the rotated quantities:
    \begin{align*}
        \hat{\beta}_\text{reg}^\text{rot} &\coloneqq V^\top \hat{\beta}_\text{reg}^\lambda,\\
        \hat{\beta}_\text{ols}^\text{rot} &\coloneqq V^\top \hat{\beta}_\text{ols},\\
        \widehat{\Delta}^\text{rot} &\coloneqq (\hat{\Phi}_q - \bar{\Phi}_p) V,\\
        \Delta^\text{rot} &\coloneqq (\bar{\Phi}_q - \bar{\Phi}_p) V.
    \end{align*}
\end{proposition}
\begin{proof}
First notice that:
\begin{align*}
       & \hat{\mathbb{E}}[ \hat{w}^\delta \circ Y_p] + \hat{\mathbb{E}}\left[\left(\Phi_q - \hat{\Phi}_q^\delta\right)\hat{\beta}_\text{reg}^\lambda\right] \\
       &= \hat{\mathbb{E}}\left[ \hat{\Phi}_q^\delta \hat{\beta}_{\text{ols}} + \left(\Phi_q - \hat{\Phi}_q^\delta\right)\hat{\beta}_\text{reg}^\lambda\right] \\
       &= \hat{\mathbb{E}}\left[ \hat{\Phi}_q^\delta V V^\top \hat{\beta}_{\text{ols}} + \left(\Phi_q - \hat{\Phi}_q^\delta\right) 
V V^\top \hat{\beta}_\text{reg}^\lambda\right] \\
\end{align*}
Then the result follows immediately from \Cref{generalregularizationpath} applied to the covariates $\Phi_p V$ and $\Phi_q V$ with weights $\hat{w}$ and outcome model $V^\top\hat{\beta}_\text{reg}^\lambda$. 
\end{proof}

Note that we could also compute $\hat{\beta}_\text{aug}$ directly as in \Cref{generalregularizationpath} because the original derivation does not depend on $\Phi_p^\top\Phi_p$ being diagonal. Combining the two propositions we have that:
\[ \hat{\mathbb{E}}[ \Phi_q \hat{\beta}_\text{aug}] = \hat{\mathbb{E}}[ \Phi_q  V \hat{\beta}_\text{aug}^\text{rot}]. \]
However, \emph{it is not the case} that $\hat{\beta}_\text{aug} = V \hat{\beta}_\text{aug}^\text{rot}$ in general. For $\ell_2$ balancing weights, $V \hat{\beta}^\text{rot}_\text{aug}$ is the correct generalization of \Cref{generalregularizationpath} to the correlated setting as we demonstrate in the next few propositions. 

\begin{proposition}\label{prop:rotated-combination-is-right}
    Let $\hat{\beta}_{\ell_2}$ be defined as in \Cref{eq:corr-l2}. For the same weights $\hat{w}_{\ell_2}^\delta$, let $\hat{\beta}_\text{aug}$ and $\hat{\beta}_\text{aug}^\text{rot}$ be defined as in Propositions \ref{generalregularizationpath} and \ref{prop:beta-aug-rotated} respectively. Then for any $\Phi_p^\top\Phi_p$ (diagonal or not):
    \[ \hat{\mathbb{E}} [ \Phi_q \hat{\beta}_{\ell_2} ] = \hat{\mathbb{E}} [ \Phi_q V \hat{\beta}_\text{aug}^\text{rot} ] = \hat{\mathbb{E}} [ \Phi_q \hat{\beta}_\text{aug} ]. \]
    Additionally, when $\Phi_p^\top\Phi_p$ is diagonal, then
    \[ \hat{\beta}_{\ell_2} = V \hat{\beta}_\text{aug}^\text{rot} = \hat{\beta}_\text{aug}. \]
    However, when $\Phi_p^\top\Phi_p$ is not diagonal, then in general
    \[ \hat{\beta}_{\ell_2} = V \hat{\beta}_\text{aug}^\text{rot} \neq \hat{\beta}_\text{aug}. \]
\end{proposition}
\begin{proof}
    The fist claims follows from Propositions \ref{generalregularizationpath}, \ref{prop:beta-aug-rotated}, and \ref{prop:corrolated-l2-augment}. For the second claim, if $\Phi_p^\top\Phi_p$ is diagonal, then $V=I$ and $\hat{\beta}_\text{aug} = \hat{\beta}_\text{aug}^\text{rot}$. Next we will show for general $\Phi_p^\top\Phi_p$ that $\hat{\beta}_{\ell_2} = V \hat{\beta}_\text{aug}^\text{rot}$. Note that
    \[ A_\delta = V (D^2 + \delta I)^{-1} D^2 V^\top, \]
    and that
    \begin{align*}
        \hat{\Delta}^\text{rot} &=  \overbar{\Phi}_q (\Phi_p^\top\Phi_p + \delta I)^{-1} (\Phi_p^\top \Phi_p) V\\
        &= \overbar{\Phi}_q V (D^2 + \delta I)^{-1} D^2
    \end{align*}  
    and so $\text{diag}(a^\text{rot}) = (D^2 + \delta I)^{-1} D^2$. Therefore,
    \begin{align*}
        V \hat{\beta}_\text{aug}^\text{rot} &= V (I - \text{diag}(a^\text{rot})) V^\top \hat{\beta}_\text{reg} + V \text{diag}(a^\text{rot}) V^\top\hat{\beta}_\text{ols}\\
        &= (VV^\top - A_\delta) \hat{\beta}_\text{reg} + A_\delta \hat{\beta}_\text{ols}\\
        &= (I - A_\delta) \hat{\beta}_\text{reg} + A_\delta \hat{\beta}_\text{ols}.
    \end{align*}
    Finally, for a counter-example where $V \hat{\beta}^\text{rot}_\text{aug} \neq \hat{\beta}_\text{aug}$, see \Cref{apx:l2-undersmooth-rotated-test}. 
\end{proof}

\begin{proposition}\label{prop:double-ridge-norm-order}
    Consider the setting in \Cref{prop:rotated-combination-is-right}, but now assume that the outcome model $\hat{\beta}_\text{reg} = \hat{\beta}_\text{ridge}^\lambda$, the ridge regression coefficients of $Y_p$ on $\Phi_p$ with hyperparameter $\lambda$. Then:
    \[ \hat{\beta}_{\ell_2} = V (D^2 + \text{diag}(\gamma))^{-1} V^\top \Phi_p^\top Y_p, \]
    where $\text{diag}(\gamma)$ is the diagonal matrix with $j$th entry
    \[ \gamma_j \coloneqq \frac{\delta \lambda}{\sigma_j^2 + \lambda + \delta}. \]
    Furthermore, 
    \[ \Vert \hat{\beta}_\text{ridge}^\lambda \Vert_2 \leq \Vert \hat{\beta}_{\ell_2} \Vert_2 = \Vert V \hat{\beta}_\text{aug}^\text{rot} \Vert_2 \leq \Vert \hat{\beta}_\text{ols} \Vert_2,  \]
    but in general it is possible that:
    \[ \Vert \hat{\beta}_\text{ols} \Vert_2 < \Vert \hat{\beta}_\text{aug} \Vert_2. \]
\end{proposition}
\begin{proof}
    \begin{align*}
        \hat{\beta}_{\ell_2} &= (I - A_\delta) (\Phi_p^\top\Phi_P + \lambda I)^{-1} \Phi_p^\top Y_p + A_\delta  (\Phi_p^\top\Phi_P)^{-1} \Phi_p^\top Y_p\\
        &= V \Big[ (I - (D^2 + \delta I)^{-1} D^2) (D^2 +\lambda I)^{-1} + (D^2 + \delta I)^{-1} D^2 (D^2)^{-1} \Big] V^\top \Phi_p^\top Y_p\\
        &= V (D^2 + (\delta+\lambda) I) (D^2 + \delta I)^{-1} (D^2 + \lambda I)^{-1} V^\top\Phi_p^\top Y_p\\
        &= V (D^2 + \text{diag}(\gamma))^{-1} V^\top \Phi_p^\top Y_p.
    \end{align*}
    We now demonstrate the two norm inequalities. First,
    \begin{align*}
        &\Vert \hat{\beta}_\text{ridge}^\lambda \Vert_2^2 \leq \Vert \hat{\beta}_{\ell_2} \Vert_2^2\\
        \iff & \Vert V (D^2 + \lambda I)^{-1} V^\top \Phi_p^\top Y_p \Vert_2^2 \leq \Vert V (D^2 + \diag(\gamma))^{-1} V^\top \Phi_p^\top Y_p \Vert_2^2\\
        \iff & Y_p^\top \Phi_p (  V \left[(D^2 + \lambda I)^{-1} (D^2 + \lambda I)^{-1} - (D^2 + \diag(\gamma))^{-1} (D^2 + \diag(\gamma))^{-1}\right] V^\top ) \Phi_p^\top Y_p \geq 0\\
        \iff & (D^2 + \lambda I)^{-1} (D^2 + \lambda I)^{-1} \preceq (D^2 + \diag(\gamma))^{-1} (D^2 + \diag(\gamma))^{-1}\\
        \iff & (D^2 + \lambda I)^{-1} \preceq (D^2 + \diag(\gamma))^{-1}. 
    \end{align*}
    Note that
    \begin{align} \frac{1}{\sigma^2_j + \gamma_j} = \left(\frac{1}{\sigma_j^2 + \lambda}\right) \left(\frac{\sigma_j^2 + \lambda + \delta}{\sigma_j^2 + \delta}\right),\label{eq:fact1} \end{align}
    and that
    \begin{align} \frac{\sigma_j^2 + \lambda + \delta}{\sigma_j^2 + \delta} \geq 1.\label{eq:fact2} \end{align}
    Combining facts (\ref{eq:fact1}) and (\ref{eq:fact2}), we have $1/(\sigma^2_j + \gamma_j) >  1/(\sigma^2_j + \lambda)$, and we're done. 

    Similarly, for the second inequality: 
    \begin{align*}
        &\Vert \hat{\beta}_{\ell_2} \Vert_2^2 \leq \Vert \hat{\beta}_\text{ols} \Vert_2^2 \\
        \iff & \Vert V (D^2 + \text{diag}(\gamma))^{-1} V^\top \Phi_p^\top Y_p \Vert_2^2 \leq \Vert V (D^2)^{-1} V^\top \Phi_p^\top Y_p \Vert_2^2\\
        \iff & (D^2 + \diag(\gamma))^{-1} \preceq (D^2)^{-1}
    \end{align*}
    which follows because $\gamma \geq 0$.

    Finally, for a counter-example where $\Vert \hat{\beta}_\text{ols} \Vert_2 < \Vert \hat{\beta}_\text{aug} \Vert_2$, see \Cref{tab:l2-norm-empirical-results}. 
\end{proof}

Even when $\Phi_p^\top \Phi_p$ is diagonal, if $\hat{\beta}_\text{reg}$ can be arbitrary --- and in particular, adversarially chosen --- then we are not guaranteed to see undersmoothing. Consider an arbitrary $\hat{\beta}_\text{reg}$ that $\ell_2$-smoothed with respect to OLS. That is, $\Vert \hat{\beta}_\text{reg} \Vert_2 \leq \Vert \hat{\beta}_\text{ols} \Vert_2$. If there are no other constraints on $\hat{\beta}_\text{reg}$, then there exist problem instances where $\hat{\beta}_{\ell_2}$ has $\ell_2$-norm larger than $\hat{\beta}_\text{ols}$ or smaller than $\hat{\beta}_\text{reg}$. For a simple example, consider $a^\delta = [0, 1/(1+\delta)]$, and $\hat{\beta}_\text{ols} = [0,1]$ --- note that these are jointly realizable by some $\Phi_p, Y_p$. Then $\hat{\beta}_\text{reg} = [1,0]$ will have $\ell_2$-norm bigger than $1$ for sufficiently large $\delta$.

\paragraph{Empirical Results for augmented $\ell_2$-balancing weights.}\label{apx:l2-undersmooth-rotated-test}
To see in practice that the rotated procedure obtains the properly undersmoothed augmented coefficients, we report norm values for $\ell_2$ augmented estimators using both a ridge and a lasso base learner in the high-dimensional LaLonde and IHDP settings. The key takeaway is that $\hat{\beta}_{\ell_2} = V \hat{\beta}^\text{rot}_\text{aug}$ has the correct undersmoothing properties.

\begin{table*}[h]
\centering
\begin{tabular}{ | c | c | c | c | c |  }
 \hline
 $\hat{\beta}_\text{reg}$ & $\Vert \hat{\beta}_\text{reg} \Vert_2$ & $\Vert \hat{\beta}_{\ell_2} \Vert_2$ & $\Vert \hat{\beta}_\text{ols} \Vert_2$ & $\Vert \hat{\beta}_\text{aug} \Vert_2$ \\
 \hline\hline
 LaLonde Ridge & 6.495e3 & 1.160e4 & 4.469e7 & 4.633e7 \\
 \hline
 LaLonde Lasso & 7.273e3 & 1.154e4 & 4.469e7 & 4.633e7 \\
 \hline
 IHDP Ridge & 1.038 & 8.569 & 9.299 & 9.722 \\
 \hline
 IHDP Lasso & 1.262 & 8.666 & 9.299 & 9.626 \\
 \hline
\end{tabular}
\caption{The $\ell_2$ norm of base learner coefficients, unrotated augmented coefficients, OLS coefficients, and properly rotated augmented coefficients. Results are reported for the LaLonde and IHDP datasets using their respective high-dimensional feature sets and with both ridge and lasso base learners.  }\label{tab:l2-norm-empirical-results}
\end{table*}

\subsection{Augmented $\ell_\infty$ balancing weights}\label{apx:linf-aug-correlated}

When we use $\ell_\infty$ balancing for the weighting model (or equivalently the minimum distance lasso estimator for the Riesz representer), the story becomes more complicated. When $\Phi_p^\top\Phi_p$ is not diagonal, both the Lasso outcome model and the $\ell_\infty$ balancing weights lack a closed-form solution. 

For Lasso with correlated features, the regularization path is qualitatively similar to the diagonal case. The coefficients start at OLS and shrink exactly to zero at different rates. However, due to the correlation between the features, each coefficient no longer linearly (or even monotonically) moves toward zero. Likewise, for $\ell_\infty$ balancing weights, we see similar behavior for the coefficients of the Riesz representer model. These coefficients are sparse, with regularization paths that move toward zero (but not monotonically). 

In the diagonal setting, the regularization path for the augmented estimator, $a_j^\delta$, is proportionate (element-wise) to $\hat{\theta}$. Therefore, because $\hat{\theta}$ is sparse, the augmented estimator coefficients defined in \Cref{generalregularizationpath} is also sparse. For general design, the impact of $\hat{\theta}$ on the augmented coefficients also depends on the empirical covariance matrix. 

\begin{lemma}
    Let $\hat{w}$ be the solution to balancing weights with hyperparameter $\delta$ for $\mathcal{F} = \{ f(x) = \theta^\top\phi(x) : \Vert \theta \Vert \leq r \}$, where $\Vert \cdot \Vert$ is any norm. Let $\hat{\theta}$ be the corresponding solution to the dual form of balancing weights, so that $\hat{w} = \Phi_p \hat{\theta}$. Then,
    \[ \hat{\Phi}_q =  \hat{\theta}^\top \Sigma, \]
    and
    \[ a_j^\delta = \frac{\hat{\theta}^\top\Sigma_j}{ \overbar{\Phi}_{q,j}}. \]
\end{lemma}

In the case of $\ell_\infty$ balancing, since $\hat{\theta}$ is sparse, $a_j^\delta$ is a sparse combination of the elements of the $j$th column of the empirical covariance matrix $\Sigma$. But unless this combination is exactly zero (typically a measure zero event), the resulting $a^\delta$ will not inherit sparsity from $\hat{\theta}$.

\section{Undersmoothing with general weights as viewed by the normal equations}

Our key numerical result, \Cref{generalregularizationpath}, allows us to write augmented weighting estimators as a plug-in estimator using coefficients $\hat{\beta}_\text{aug}$. Notably, this result holds for weights $\hat{w}$ as long as $\hat{w} = \hat{\theta} \Phi_p^\top$ for \emph{any} vector $\hat{\theta} \in \mathbb{R}^d$. The numerical result holds even if $\hat{\theta}$ doesn't result in a reasonable estimator of the Riesz representer. 

It is interesting to see under what conditions on $\hat{w}$ does $\hat{\beta}_\text{aug}$ still behave as an ``undersmoothed'' estimator. We discussed some ways to define this in \Cref{apx:l2-rot-wlog}; for example requiring that the $\ell_2$-norm of the coefficieents be between that of the base learner $\hat{\beta}$ and $\hat{\beta}_\text{ols}$. In this section, we develop a different notion of ``undersmoothing'' in terms of the size of the normal equation violations. In penalized linear regression, the smaller the penalty, the closer the LHS of the normal equation gets to zero. Therefore one way to view ``undersmoothing'' is to measure the size of the normal equation violations. Formally, the first order condition of (unpenalized) linear regression with coefficients $\beta$ is:

\[ \mathcal{N}(\hat{\beta}_\text{ols}) \coloneqq \Phi_p^\top ( Y_p - \Phi_p \hat{\beta}_\text{ols} ) = 0. \]

Define $A$ to be the diagonal matrix with $j$th entry $a_j$ as defined in \Cref{generalregularizationpath}. Then the corresponding normal equation violation (the gradient of the squared loss) for $\hat{\beta}_\text{aug}$ is:

\begin{align*}
    \mathcal{N}(\hat{\beta}_\text{aug}) &= \Phi_p^\top (Y_p - \Phi_p A \hat{\beta}_\text{ols} - \Phi_p (I - A)\hat{\beta})\\
    &= \Phi_p^\top (Y - \Phi_p \hat{\beta}) + (\Phi_p^\top \Phi_p)A(\hat{\beta} - \hat{\beta}_\text{ols})\\
    &= \Phi_p^\top(Y - \Phi_p \hat{\beta}_\text{ols}) + (\Phi_p^\top \Phi_p)(A-I)(\hat{\beta} - \hat{\beta}_\text{ols})\\
    &= (\Phi_p^\top \Phi_p)(I - A)(\hat{\beta}_\text{ols} - \hat{\beta})
\end{align*}  

Note that when $A = 0$, we recover the original outcome model, e.g. $\hat{\beta}_\text{aug} = \hat{\beta}$ and so the normal equation violations of $\hat{\beta}$ can be written:

\[ \mathcal{N}(\hat{\beta}) = (\Phi_p^\top \Phi_p)(\hat{\beta}_\text{ols} - \hat{\beta}).\]

In other words, how close $\hat{\beta}$ is to optimizing the squared loss (in terms of the size of the gradient) is just equal to how close the coefficients are to the OLS coefficients, but mediated via the sample covariance matrix. E.g. large eigenvalues of $\Phi_p^\top \Phi_p$ will amplify the difference more, and small eigenvalues less. The only difference with the augmented coefficients $\hat{\beta}_\text{aug}$ is the introduction of the scaling diagonal matrix $(I-A)$. If $A$ is uniformly close to zero, it will damp the eigenvalues of the covariance matrix and result in small normal equation violations. This is most clear when the sample covariance matrix is diagonal with entries $\sigma^2_j$. In this case, the $j$th element of the squared loss gradient becomes:

\[ \sigma^2_j (1 - a_j) (\hat{\beta}_{\text{ols},j} - \hat{\beta}_j). \]

Whenever all $a_j$ are between $0$ and $2$, then the violations of the normal equations for $\hat{\beta}_\text{aug}$ are guarantee to be uniformly at least as small as that of $\hat{\beta}$.


\section{Comparison to TMLE with Balancing Weights}\label{apx:tmle-connection}

As pointed out in \cite{chernozhukov2022automatic}, the balancing weights estimator of the Riesz representer can be applied to any doubly-robust meta algorithm. In the main text, we used to standard AIPW form, but we can also consider the TMLE form \cite{tmle, van2011targeted}. In the linear setting, for an outcome model $\hat{\beta}$ and weights $\hat{w} \coloneqq \Phi_p \hat{\theta}$ this takes the form:
\begin{align*}
    \bar\Phi_q^T (\hat{\beta} + \hat{\epsilon} \hat{\theta} ),
\end{align*}
where
\begin{align*}
    \hat{\epsilon} &\coloneqq \text{argmin}_\epsilon \Big\{ \Vert Y_p - \Phi_p \hat{\beta} - \epsilon \hat{w} \Vert_2^2 \Big\}\\
    &= \frac{\hat{w}^T ( Y_p - \Phi_p \hat{\beta} )}{\hat{w}^T \hat{w}}.
\end{align*}

In this section, we show the explicit connection between the AIPW and TMLE forms which are closely related. This connection can be stated in two ways, an ``AIPW-centric'' form and a ``TMLE-centric'' form.

\textbf{AIPW-centric version of stating the connection:}
\begin{align*}
    \text{AIPW: } \hspace{3mm} &\bar{\Phi}_q^T \hat{\beta} + \hspace{20.5mm} \frac{1}{n} \hspace{20.5mm} \hat{w}^T(Y_p - \Phi_p \hat{\beta})\\
    \text{TMLE: } \hspace{3mm} &\bar{\Phi}_q^T \hat{\beta} + \underbrace{\left(\frac{\bar{\Phi}_q^T \hat{\theta}}{ \hat{\theta}^T \Phi_p^T \Phi_p \hat{\theta}}\right)}_{\text{scaled ratio of terms in Riesz loss}} \hat{w}^T(Y_p - \Phi_p \hat{\beta})\\
\end{align*}

\textbf{TMLE-centric version of stating the connection:}
\begin{align*}
    \text{TMLE: } \hspace{3mm} &\bar{\Phi}_q^T \hat{\beta} + \hspace{6.5mm} \hat{\epsilon} \bar{\Phi}_q^T \hat{\theta}\\
    \text{AIPW: } \hspace{3mm} &\bar{\Phi}_q^T \hat{\beta} + \underbrace{\hat{\epsilon} \hat{\Phi}_q^T \hat{\theta}}_{\hat{\Phi}_q \text{ instead of } \bar{\Phi}_q}\\
\end{align*}

From the ``AIPW-centric'' version of the connection, we see the familiar $\hat{w}^T Y_p$ term appear in the TMLE, and can apply our main results to rewrite the TMLE estimator around $\hat{\beta}_\text{ols}$. However, because we now weight the average by a term that is different from $\frac{1}{n}$, the results may have different variance properties. We think this is a promising direction for future research.

\section{Simulation Study Details}\label{apx:simulation-details}

\subsection{Setup}

We consider 36 different data generating processes (DGPs) for our simulation study. For each of them, we compare an oracle baseline with three feasible hyperparameter tuning schemes.

For the remainder of this section, we we say ``numerically optimize'', we mean using the scipy.optimize.minimize solver with tolerance 1e-12. In particular, we pre-compute the SVD of the covariance matrix before solving, so that we can compute the pseudoinverse for all involved expressions in closed form without performing traditional matrix inversion during optimization. 

For the oracle hyperparameters, we compute $\lambda^*$ by numerically optimizing the in-distribution mean squared error for ridge regression from \Cref{sec:finite-sample-mse}. We then fix $\lambda=\lambda^*$ and then numerically optimize the expression in \Cref{prop:finite-sample-mse} to get the MSE-optimal $\delta^*$

The three feasible hyperparameter tuning schemes we consider, by contrast, use a particular draw of $Y_p$. In all cases, we choose $\lambda$ by cross-validating a ridge outcome model. Then we considering choosing $\delta$ by: (1) cross-validating balance, (2) cross-validating the Riesz loss, and (3) setting $\delta$ equal to the cross-validated ridge $\lambda$. In all cases, cross-validation is performed 5-fold via numerical optimization as described above instead of using a grid. 

Note that the oracle hyperparameters are only optimal in a setting where we have to pick a single $\lambda$ and $\delta$ for each draw of $Y_p$. Cross-validation picks a new $\lambda$ for each $Y_p$ and so can theoretically outperform the oracle. This happens very rarely, but a non-zero amount of the time. Overall, the results in \Cref{tab:sim_results} suggest that this is still a very good baseline. Finally, note that we could have picked $\lambda$ by cross validation for the oracle, and then solved for the optimal $\delta^*$ separately for each $Y_p$, fixing the CV value of $\lambda$. However, we cannot guarantee that the resulting $\delta$ will have any optimality properties, since the mean squared error expression is derived by averaging over $Y_p$ draws for fixed values of $\lambda$ and $\delta$. But this could be an interesting follow up experiment. 

In all cases, we take 1000 draws of $Y_p$ and then compute the squared error and take their average as a monte carlo estimate of the mean squared error. 

\subsection{Synthetic DGPs}

To compute the oracle $\lambda^*$ and $\delta^*$ from \Cref{sec:finite-sample-mse}, we need: the true coefficients, $\beta_0$, the population covariance matrix, $\mathbb{E}[ \Phi_p^T \Phi_p]$, a sample covariance matrix $\hat{\Sigma}$, the conditional variance $\sigma^2$, and the target covariance mean, $\mathbb{E}[\Phi_q]$. So for each DGP, we need to specify these five objects. 

We consider synthetic DGPs with three basic setups. They all use $n=2000$ and $d = 50$. For each of the three setups, we draw a random $\beta_0$, that is the absolute value of a $d$-dimensional standard normal, that is then normalized to have 2-norm equal to 1. The three setups each generate a popuation covariance matrix in roughly the same way. In each case, we choose a maximum and minimum eigenvalue, $\eta_\text{min}$ and $\eta_\text{max}$ respectively. We then generate an equally space grid between $\eta_\text{min}^{1/c}$ and $\eta_\text{min}^{1/c}$ for some curvature constant $c$. We choose the eigenvalues of the covariance matrix to be the numbers in this grid raised to the $c$th power. We then draw a random eigenvector matrix from the special orthogonal group, $U$, and form the covariance matrix from the eigenvectors and eigenvalues in the standard way. Next, we draw an $\Phi_p$ with $n=2000$ samples from a mean-zero normal distribution with this covariance matrix, and compute $\hat{\Sigma} = \Phi_p^T\Phi_p/n$.

The three basic setups differ in teh choice of $\eta_\text{min}, \eta_\text{max},$ and $c$. For setting 1 we choose $\eta_\text{min} = 1e-4$, $\eta_\text{max}=3$ and $c= 5000$. For setting 2 we choose $\eta_\text{min} = 1e-8, \eta_\text{max} = 3$, and $c=5000$. For setting 3, we choose $\eta_\text{min}=1e-10, \eta_\text{max}=5, c=10$.

Then for each of these basic setups, we create 10 DGPs, by consider all combinations of a list of $\mathbb{E}[\Phi_q]$ and $\sigma^2$. We use $\sigma^2 \in \{ 0.1, 2 \}$. For $\mathbb{E}[\Phi_q]$ we use: the vector of all $0.1$, the vector of all $2$, and then 3 vectors chosen randomly uniformly between $-1$ and $1$ which are then scaled to have norm 1. Thus the total of $2 \times 5 = 10$ DGPs for each setup.

\subsection{Semi-Synthetic DGPs}

We then also use semi-synthetic DGPs based on Lalonde Long and IHDP Long (see \Cref{apx:numerical-experiment-details} for details on these datasets). For each of these datasets, we recenter $\Phi_p$ and $Y_p$ to have mean-zero. Then we choose $\beta_0$ to be the coefficients from cross-validated ridge regression of $Y_p$ on $\Phi_p$. We let $\sigma^2$ be the variance of the residuals from this regression, and we choose the population covariance matrix to be $\Phi_p^T\Phi_p/n$. Next, we \emph{redraw} a new matrix of samples from a mean-zero normal distribution with this covariance matrix, and compute the sample covariance $\hat{\Sigma}$ from these \emph{new} samples. For targets, we use the actual $\mathbb{E}[\Phi_q]$, but also include two perturbed versions where all even elements of $\mathbb{E}[\Phi_q]$ are either increased or decreased by a proportion of the norm of $\mathbb{E}[\Phi_q]$. For IHDP, the perturbation is $1/10$ times the norm, for Lalonde, the perturbation is $1/100$ times the norm. 

So since for each semi-synthetic setting we have three values of $\mathbb{E}[\Phi_q]$ and one value of $\sigma^2$, that corresponds to 6 DGPs each, for a total of 36 together with the synthetic DGPs.

\section{Additional Numerical Experiments}\label{apx:numerical-experiment-details}

\subsection{Dataset Summaries}

The following table summarizes the number of samples and features in the datasets used for numerical illustrations. In the main text, we presented results for ``LaLonde Short'' and ``LaLonde Long'', the \cite{lalonde1986evaluating} data with the original 11 features and the expanded 171 features from \cite{farrell2015robust}.

In \Cref{apx:additional-experiments}, we also provide results for the Infant Health and Development Program (IHDP) dataset, a standard observational causal inference benchmark from \cite{hill2011bayesian}. IHDP is based on data from a randomized control trial of an intensive home visiting and childcare intervention for low birth weight infants born in 1985. For all children, we have a range of baseline covariates (contained in ``IHDP Short''), including both categorical covariates, like the mother’s educational attainment, and continuous covariates, like the child’s birth weight. Our goal is to estimate the average outcome (a standardized test score) in the absence of the intensive intervention. We additionally create an ``IHDP Long'' extended feature set that includes all pair-wise interaction terms between the discrete covariates and a second-order polynomial expansion of all continuous covariates. 

\begin{table*}[h]
\centering
\begin{tabular}{ | c | c | c | c |  }
 \hline
 Dataset & Features & Train Samples & Test Samples   \\
 \hline\hline
 LaLonde Short & 11 & 727 & 185  \\
 \hline
 LaLonde Long & 171 & 727 & 185 \\
 \hline
 IHDP Short & 25 & 608 & 139 \\
 \hline
 IHDP Long & 193 & 608 & 139\\
 \hline
\end{tabular}
\caption{Summary of the four datasets used for numerical illustrations including the number of features, the number of control/train observations (used 
 as samples from $p$) and the number of treated/test observations (used as samples from $q$). }\label{tab:data-summary}
\end{table*}

\subsection{Numerical Example with cross-fitting}\label{apx:lalonde-cross-fit}

In \Cref{sec:short-regression}, we demonstrated that for the low-dimensional NSW setting, the cross-validated weighting hyperparameter is $\delta = 0$. As a result, for any linear base learner, the augmented estimator is numerically equivalent to a simple OLS plug-in estimate. As discussed in \Cref{apx:cross-fit-unreg-weight}, with cross-fitting this should only remain approximately true. In this section, we numerically assess the sensitivity of this result to cross-fitting. 

We repeat the experiment with the 11-dimensional NSW covariates using a lasso base learner and $\ell_\infty$ balancing weights. We perform 5-fold cross-fitting with random splits. In each split, we compute the augmented estimate using models fit with data from the other 4 splits as described at the beginning of \Cref{sec:sample-split}. Each model has its hyperparameter chosen separately by cross-validation. Thus for 5-fold cross-fitting, we fit 5 lasso regressions and 5 balancing weights estimators each with $4/5$ths the sample size and a separate hyperparameter. In the manner of bootstrapping, we re-run this end-to-end procedure 1000 times for different random choices of the splits. 

Note that there are two main ways that cross-fitting might cause the final estimate to deviate from simple OLS. The first is the variation due to only applying the models out of sample. The second is that with smaller sample sizes, cross-validation might select larger values for the hyperparameter. We find that for 80\% of the $5 \times 1000$ weighting models, $\delta = 0$. In the remaining 20\% some of the hyperparameters take on a very small but non-zero value. The mean $\delta$ is 0.0037 and the CDF across the $5 \times 1000$ models is given in \Cref{fig:lalonde-short-cross-fit-delta}.

\begin{figure}[htb]
    \centering
    \begin{subfigure}{0.4\textwidth}
        \centering 
        \includegraphics[width=\textwidth]{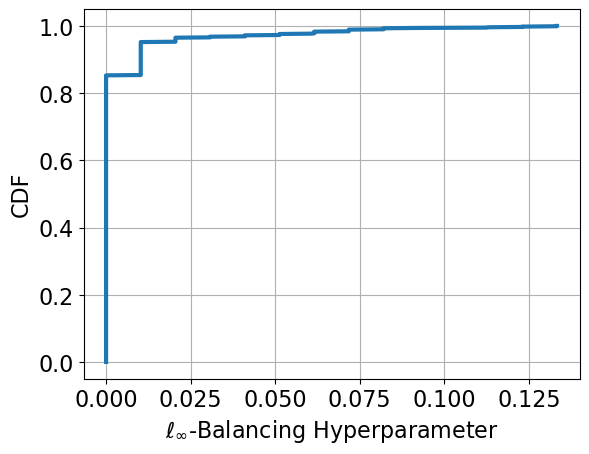}
        \caption{Weighting Hyperparameter CDF}
        \label{fig:lalonde-short-cross-fit-delta}
    \end{subfigure}
    \begin{subfigure}{0.5\textwidth}
        \centering 
        \includegraphics[width=\textwidth]{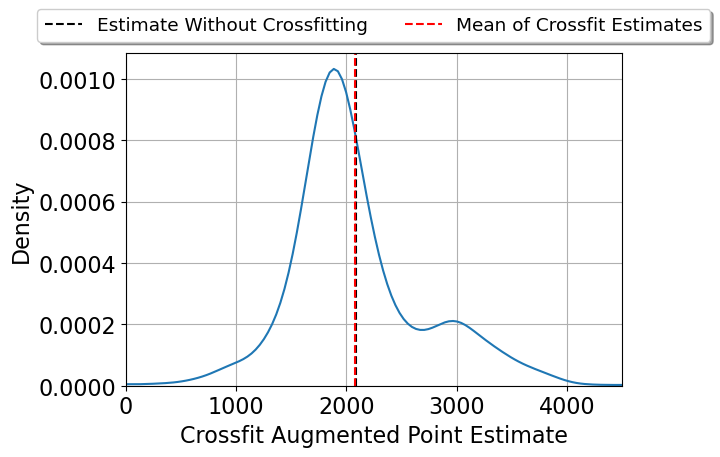}
        \caption{Distribution of cross-fit Estimates}
        \label{fig:lalonde-short-cross-fit-ests}
    \end{subfigure}
    \caption{Results with 5-fold cross-fitting for the \cite{lalonde1986evaluating} dataset with the low-dimensional set of 11 features. The cross-fitting procedure is repeated $1000$ times. Panel (a) plots the empirical cumulative distribution of weighting hyperparameters over the $1000 \times 5$ weighting models; in the main text, the $\ell_\infty$ hyperparameter extends from 0 to 1 (here truncated at 0.12). Panel (b) shows a density plot for the cross-fit augmented point estimates over the $1000$ repeats.   }
    \label{fig:lalonde-cross-fit}
\end{figure}

\Cref{fig:lalonde-short-cross-fit-ests} shows a smoothed density plot for the cross-fit augmented estimates over the 1000 draws. The black dotted line marks the augmented estimate without cross-fitting (recall that this is identical to the plug-in OLS estimate for this dataset). The red dotted line marks the mean of the cross-fit estimates over the 1000 draws.  \emph{On average}, cross-fitting has virtually no effect on the augmented estimate --- at least in this setting where the weighting hyperparameter is usually close to zero. However, there is substantial variation across the 1000 draws. In particular, the density is skewed and the mode of the cross-fit estimates is slightly smaller than the OLS point estimate. There is also a meaningful tail of estimates larger than the OLS point estimate, including a second larger mode.

\subsection{Additional Experiments}\label{apx:additional-experiments}

\subsubsection{Low Dimensional}

We begin by providing corroborating results for the low dimensional setting. We repeat the experiments in \Cref{sec:short-regression} in three additional settings: the LaLonde low-dimensional setting with $\ell_2$ balancing weights in \Cref{fig:lalonde-short-ridge}, the IHDP low-dimensional setting with $\ell_\infty$ balancing weights in \Cref{fig:ihdp-short-lasso}, and the IHDP low-dimensional setting with $\ell_2$ balancing weights in \Cref{fig:ihdp-short-ridge}. In all of these cases, cross-validation for the weighting hyperparameter choose $\delta = 0$, and thus the augmented estimator is equivalent to the OLS plug-in estimate. These plots can be compared to \Cref{fig:lalonde-short-lasso} in the main text. 

\begin{figure}[h]
    \centering
    \begin{subfigure}{0.32\textwidth}
        \centering 
        \includegraphics[width=\textwidth]{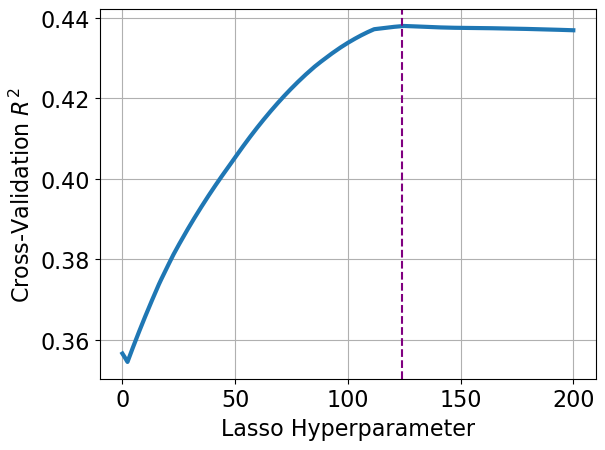}
        \caption{Lasso outcome model}
                \label{fig:lalonde-short-lasso-panel-a}
    \end{subfigure}
    \begin{subfigure}{0.325\textwidth}
        \centering 
        \includegraphics[width=\textwidth]{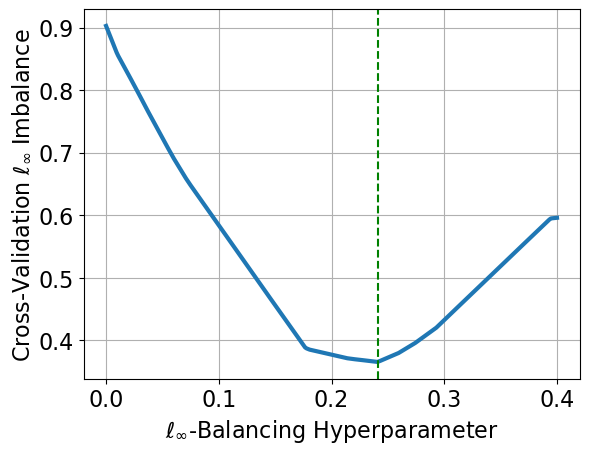}
        \caption{$\ell_\infty$ balancing}
                \label{fig:lalonde-short-lasso-panel-b}
    \end{subfigure}
        \begin{subfigure}{0.31\textwidth}
        \centering 
        \includegraphics[width=\textwidth]{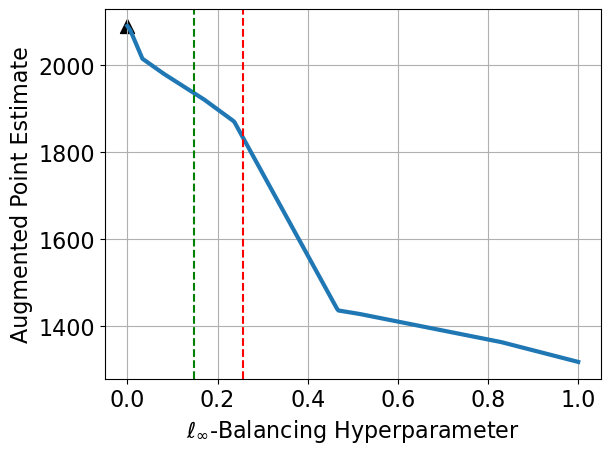}
        \caption{Augmented estimate}
        \label{fig:lalonde-short-lasso-panel-c}
    \end{subfigure}
\caption{Lasso-augmented $\ell_\infty$ balancing weights (``double lasso'') for \citet{lalonde1986evaluating} with the original 11 features. Panel (a) shows the 10-fold cross-validated $R^2$ for the Lasso-penalized regression of $Y_p$ on $\Phi_p$ among control units across the hyperparameter $\lambda$; the purple dotted line shows the CV-optimal value. Panel (b) shows the 10-fold cross-validated imbalance for $\ell_\infty$ balancing weights across the hyperparameter $\delta$; the green dotted line shows the CV-optimal value. Panel (c) shows the point estimate for the augmented estimator across the weighting hyperparameter $\delta$; the black triangle corresponds to the OLS point estimate, the green dotted line corresponds to cross-validating imbalance, the red dotted line corresponds to cross-validating the Riesz loss.}
    \label{fig:lalonde-short-lasso}
\end{figure}

\begin{figure}[h]
    \centering
    \begin{subfigure}{0.32\textwidth}
        \centering 
        \includegraphics[width=\textwidth]{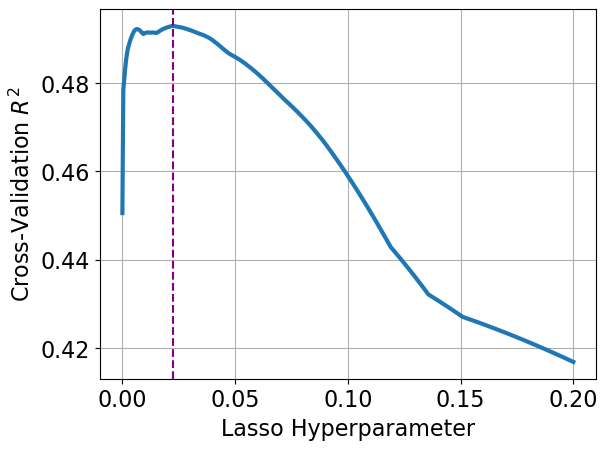}
        \caption{Lasso outcome model}
                \label{fig:ihdp-short-lasso-panel-a}
    \end{subfigure}
    \begin{subfigure}{0.325\textwidth}
        \centering 
        \includegraphics[width=\textwidth]{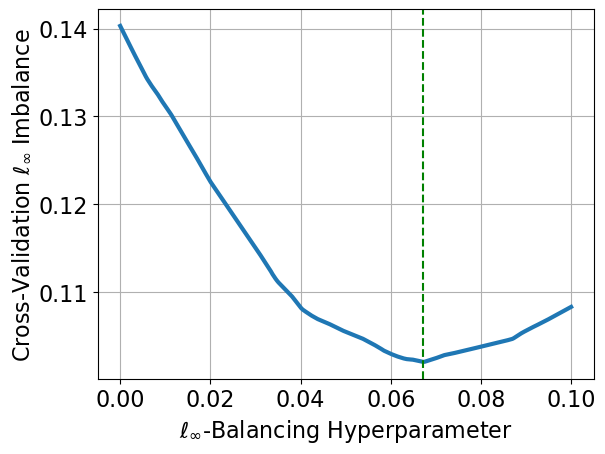}
        \caption{$\ell_\infty$ balancing}
                \label{fig:ihdp-short-lasso-panel-b}
    \end{subfigure}
        \begin{subfigure}{0.31\textwidth}
        \centering 
        \includegraphics[width=\textwidth]{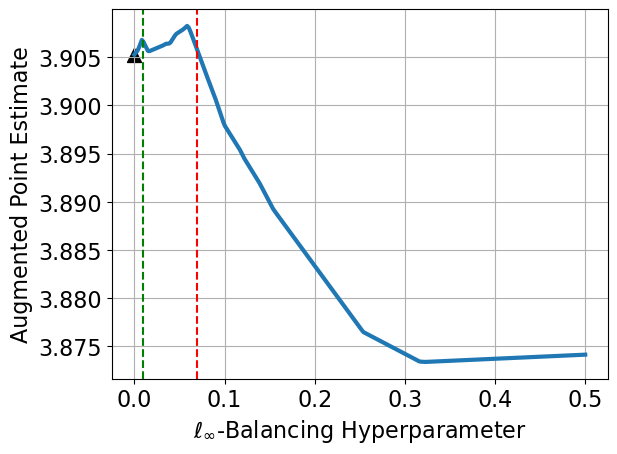}
        \caption{Augmented estimate}
        \label{fig:ihdp-short-lasso-panel-c}
    \end{subfigure}
\caption{Lasso-augmented $\ell_\infty$ balancing weights (``double lasso'') for IHDP with the original 25 covariates. Panel (a) shows the 10-fold cross-validated $R^2$ for the Lasso-penalized regression of $Y_p$ on $\Phi_p$ among control units across the hyperparameter $\lambda$; the purple dotted line shows the CV-optimal value. Panel (b) shows the 10-fold cross-validated imbalance for $\ell_\infty$ balancing weights across the hyperparameter $\delta$; the green dotted line shows the CV-optimal value. Panel (c) shows the point estimate for the augmented estimator across the weighting hyperparameter $\delta$; the black triangle corresponds to the OLS point estimate, the dotted green line corresponds to cross-validated imbalance, and the dotted red line corresponds to cross-validated Riesz loss.}
    \label{fig:ihdp-short-lasso}
\end{figure}

\begin{figure}[h]
    \centering
    \begin{subfigure}{0.32\textwidth}
        \centering 
        \includegraphics[width=\textwidth]{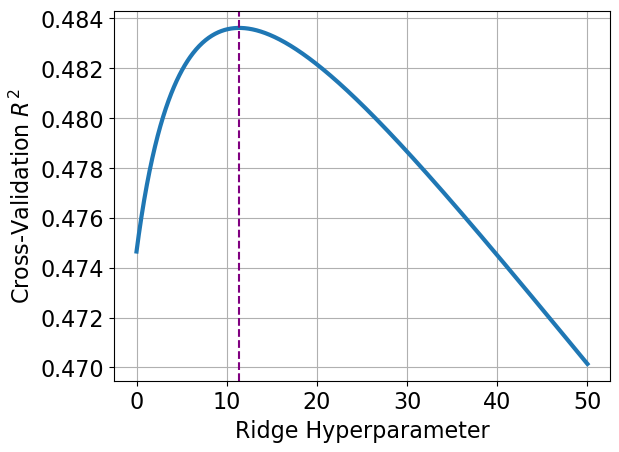}
        \caption{Ridge outcome model}
                \label{fig:ihdp-short-ridge-panel-a}
    \end{subfigure}
    \begin{subfigure}{0.325\textwidth}
        \centering 
        \includegraphics[width=\textwidth]{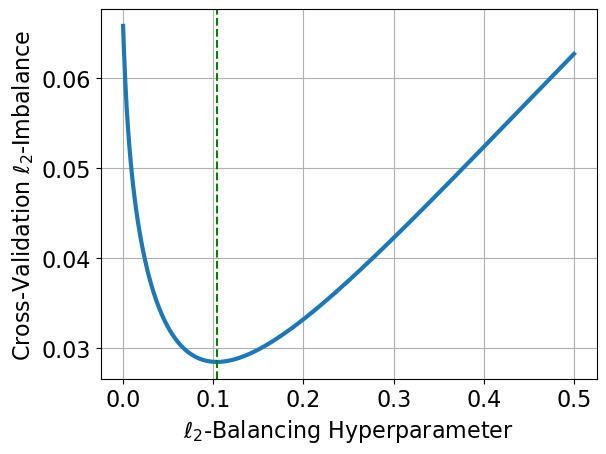}
        \caption{$\ell_2$ balancing}
                \label{fig:ihdp-short-ridge-panel-b}
    \end{subfigure}
        \begin{subfigure}{0.31\textwidth}
        \centering 
        \includegraphics[width=\textwidth]{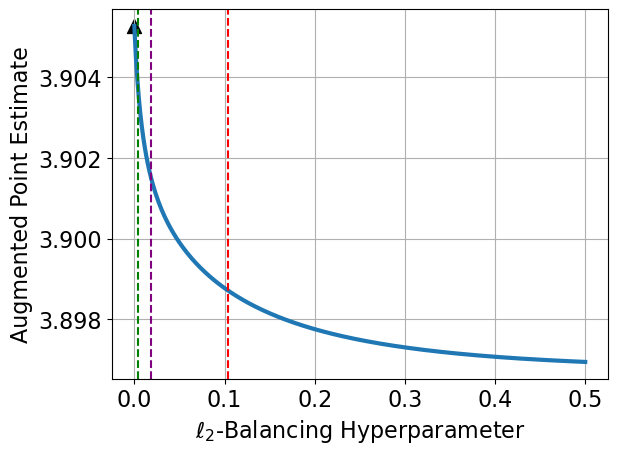}
        \caption{Augmented estimate}
        \label{fig:ihdp-short-ridge-panel-c}
    \end{subfigure}
\caption{Ridge-augmented $\ell_2$ balancing weights (``double ridge'') for IDHP with the original 25 covariates. Panel (a) shows the 10-fold cross-validated $R^2$ for the Ridge-penalized regression of $Y_p$ on $\Phi_p$ among control units across the hyperparameter $\lambda$; the purple dotted line shows the CV-optimal value. Panel (b) shows the 10-fold cross-validated AutoDML loss (\ref{eq:autoform}) for $\ell_2$ balancing weights across the hyperparameter $\delta$; the green dotted line shows the CV-optimal value, which is $\delta = 0$ or exact balance. Panel (c) shows the point estimate for the augmented estimator across the weighting hyperparameter $\delta$; the black triangle corresponds to the OLS point estimate, the green dotted line corresponds to cross-validated balance, the red dotted line corresponds to cross-validated Riesz loss, and the purple dotted line corresponds to the cross-validated ridge outcome hyperparameter.}
    \label{fig:ihdp-short-ridge}
\end{figure}

\subsubsection{High Dimensional}

We also replicate the experiments in \Cref{sec:numerical-high-dim-results} for high dimensional double ridge and double lasso but using the IHDP data with a high dimensional feature set. \Cref{fig:ihdp-long} replicates the reporting of point estimates and cross-validation curves for double lasso and double ridge with IHDP as in \Cref{fig:lalonde-long}; \Cref{fig:ihdp-ridge-undersmooth} replicates undersmoothing in double ridge as in \Cref{fig:lalonde-ridge-undersmooth}; and \Cref{fig:ihdp-aug-norm-fig} replicates undersmoothing in the $\ell_0$ ``norm'' as in \Cref{fig:lalonde-aug-norm-fig}. 

\begin{figure}[h]
    \centering
    \begin{subfigure}{0.32\textwidth}
        \centering 
        \includegraphics[width=\textwidth]{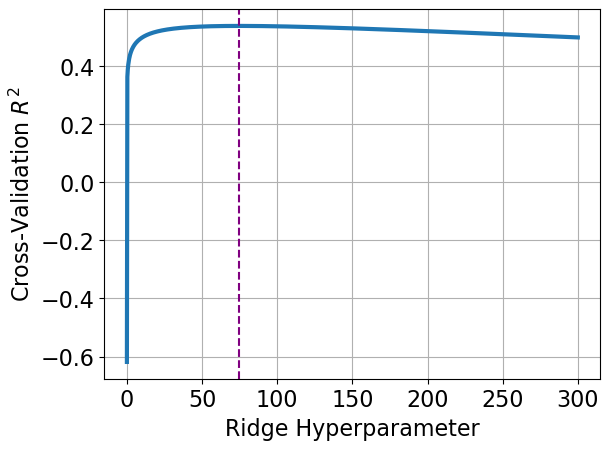}
        \caption{Ridge outcome model}
            \label{fig:ihdp-long-ridge-panel-a}
    \end{subfigure}
    \begin{subfigure}{0.325\textwidth}
        \centering 
        \includegraphics[width=\textwidth]{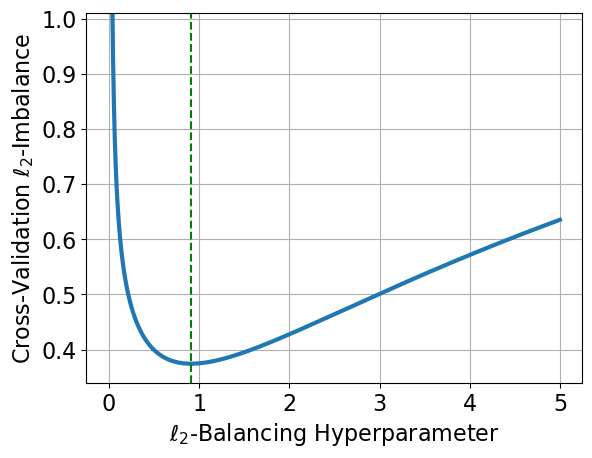}
        \caption{$\ell_2$ balancing}
         \label{fig:idhp-long-ridge-panel-b}
    \end{subfigure}
        \begin{subfigure}{0.31\textwidth}
        \centering 
        \includegraphics[width=\textwidth]{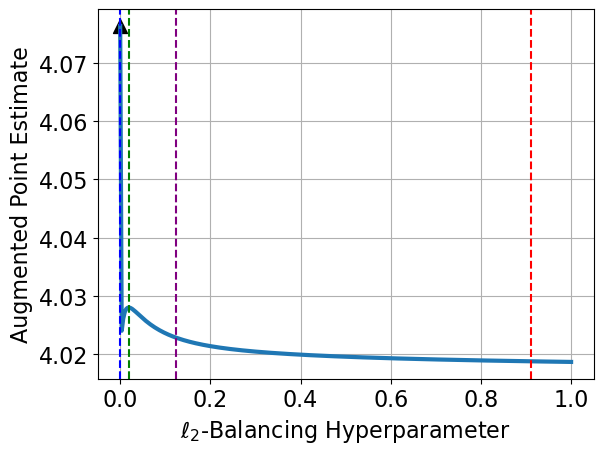}
        \caption{Estimate from ``double ridge''}
        \label{fig:ihdp-long-ridge-panel-c}
    \end{subfigure} \\[1em]
    \begin{subfigure}{0.32\textwidth}
        \centering 
        \includegraphics[width=\textwidth]{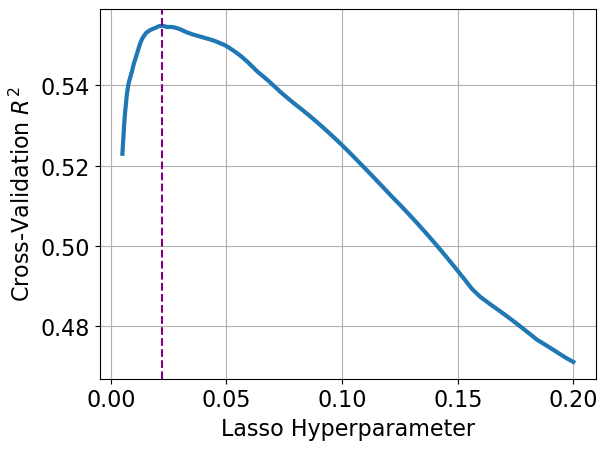}
        \caption{Lasso outcome model}
        \label{fig:ihdp-long-lasso-panel-a}
    \end{subfigure}
    \begin{subfigure}{0.325\textwidth}
        \centering 
        \includegraphics[width=\textwidth]{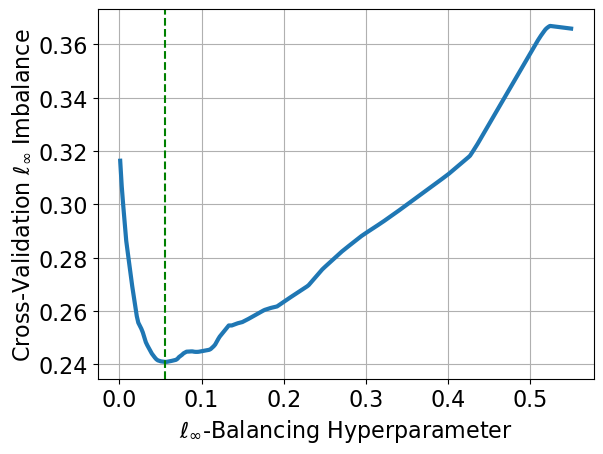}
        \caption{$\ell_\infty$ balancing}
        \label{fig:ihdp-long-lasso-panel-b}
    \end{subfigure}
        \begin{subfigure}{0.31\textwidth}
        \centering 
        \includegraphics[width=\textwidth]{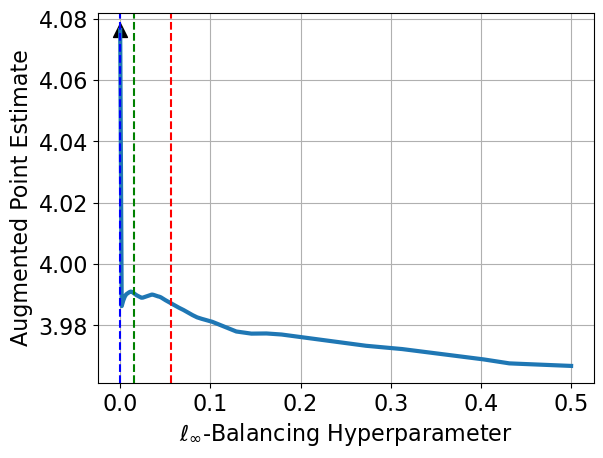}
        \caption{Estimate from ``double lasso''}
        \label{fig:ihdp-long-lasso-panel-c}
    \end{subfigure}    
\caption{Augmented balancing weights estimates for the IHDP data set with the expanded set of 193 features; the top row shows ridge-augmented $\ell_2$ balancing, and the bottom row shows lasso-augmented $\ell_\infty$ balancing.
Panels (a) and (d) show the 3-fold cross-validated $R^2$ for the ridge- and laso-penalized regression of $Y_p$ on $\Phi_p$ among control units across the hyperparameter $\lambda$; the purple dotted lines show the CV-optimal value for each. 
Panel (b) and (e) show the 10-fold cross-validated AutoDML loss (\ref{eq:autoform}) for $\ell_2$ and $\ell_\infty$ balancing weights across the hyperparameter $\delta$; the green dotted lines show the CV-optimal value for each. 
Panels (c) and (f) show the point estimates for the augmented estimators across the weighting hyperparameter $\delta$; the black triangles correspond to the OLS point estimate, and the green dotted lines show the CV-optimal value of $\delta$ for each.}
    \label{fig:ihdp-long}
\end{figure}

\begin{figure}[h]
    \centering
    \begin{subfigure}{0.4\textwidth}
        \centering 
        \includegraphics[width=\textwidth]{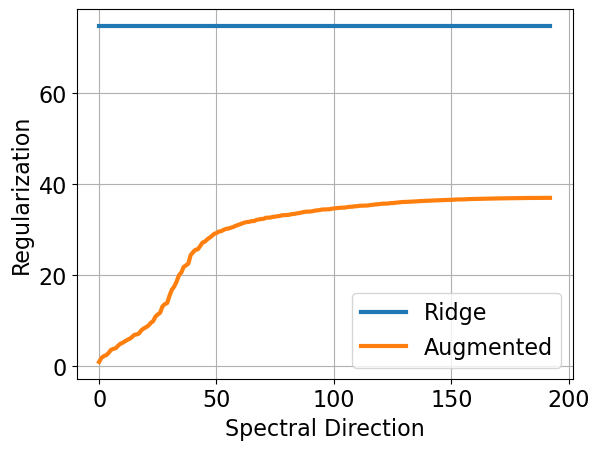}
        \caption{Adaptive regularization}
        \label{fig:ihdp-ridge-undersmooth-panel-a}
    \end{subfigure}
        \begin{subfigure}{0.425\textwidth}
        \centering 
        \includegraphics[width=\textwidth]{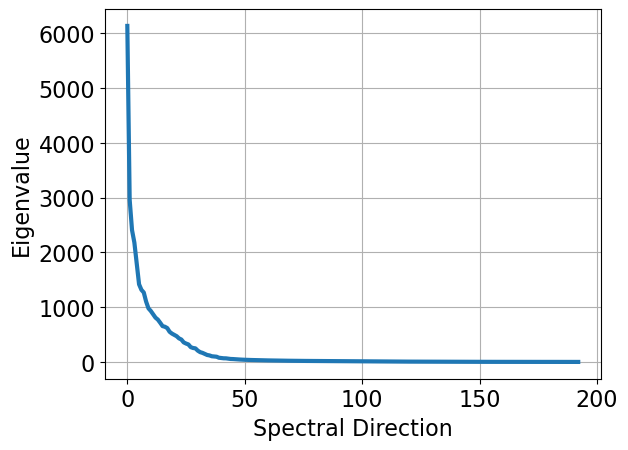}
        \caption{Eigenvalues of $\Phi_p^\top\Phi_p$}
        \label{fig:ihdp-ridge-undersmooth-panel-b}
    \end{subfigure}
    \caption{Single and double ridge regression applied to the ``high-dimensional'' version of the IHDP data set. Panel (a) shows $\lambda_j$, the single ridge hyperparameter, and $\gamma_j$, the implied double ridge hyperparameter, as functions of the spectral direction $j$. Panel (b) shows the corresponding eigenvalues.}
    \label{fig:ihdp-ridge-undersmooth}
\end{figure}

\begin{figure}[h]
    \centering
    \begin{subfigure}{0.4\textwidth}
        \centering 
        \includegraphics[width=\textwidth]{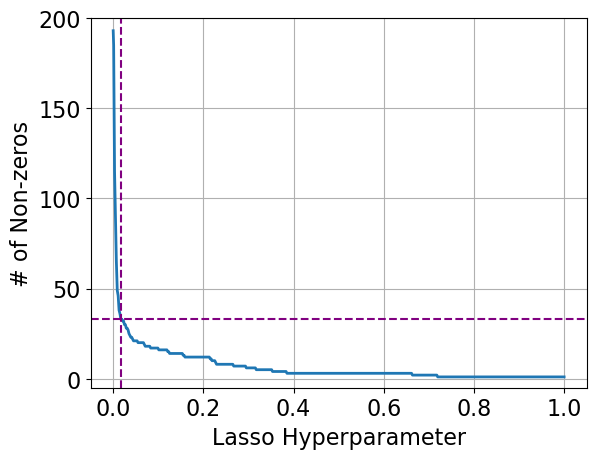}
        \caption{$\ell_0$ ``norm'' for ridge outcome model}
        \label{fig:ihdp-aug-norm-l0-panel-a}
    \end{subfigure}
        \begin{subfigure}{0.4\textwidth}
        \centering 
        \includegraphics[width=\textwidth]{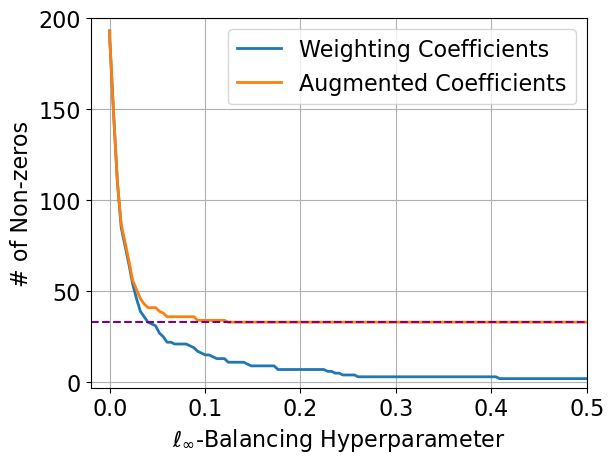}
        \caption{$\ell_0$ ``norm'' for weighting and double lasso}
        \label{fig:ihdp-aug-norm-l2-panel-b}
    \end{subfigure} \\[1em]
    \caption{Undersmoothing in $\ell_0$ ``norm'' for augmented balancing weights applied to the ``high dimesional'' version of the IHDP data set. Panels (a) shows the number of non-zero covariates (the $\ell_0$ ``norm'') for lasso regression of $Y_p$ on $\Phi_p$ among control units. Panels (b) shows the number of non-zero covariates for the weighting model and augmented coefficients. The dotted purple line is the outcome model hyperparameter chosen via cross validation.}
    \label{fig:ihdp-aug-norm-fig}
\end{figure}

\subsection{Undersmoothing in LaLonde}\label{apx:undersmooth-lalonde}

\subsubsection{Undersmoothing for double ridge}

For the special case of double ridge, we now illustrate our result from \Cref{prop:double_ridge} that $\ell_2$ balancing weights produce a new outcome model that is \emph{undersmoothed} relative to the original ridge regression model. Recall that when the base learner is a generalized ridge model with parameters $\lambda_j$, then the augmented estimator is equivalent to plugging in a generalized ridge models with parameters $\gamma_j \leq \lambda_j$. In this example, the $\lambda_j$ are all equal to the same value (chosen via cross-validation), indicated by the purple dotted line in \Cref{fig:lalonde-long-ridge-panel-a}. We compute the corresponding $\gamma_j$ for $\delta$ chosen by cross-validation, indicated by the green dotted line in \Cref{fig:lalonde-long-ridge-panel-b}.

\begin{figure}[htb!]
    \centering
    \begin{subfigure}{0.4\textwidth}
        \centering 
        \includegraphics[width=\textwidth]{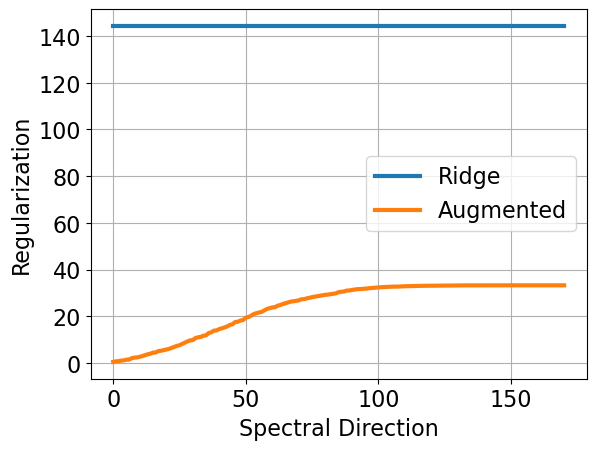}
        \caption{Adaptive regularization}
        \label{fig:lalonde-ridge-undersmooth-panel-a}
    \end{subfigure}
        \begin{subfigure}{0.425\textwidth}
        \centering 
        \includegraphics[width=\textwidth]{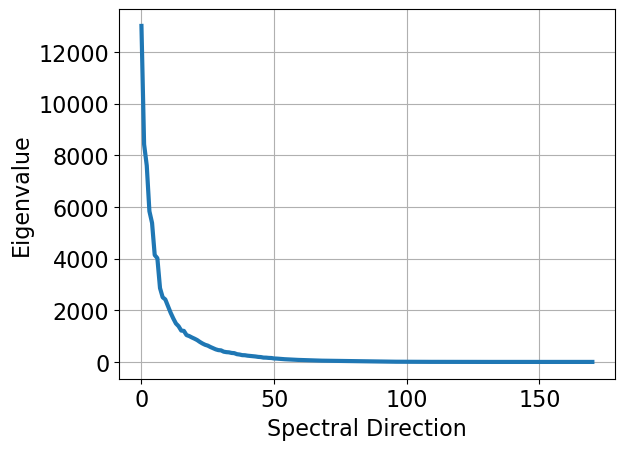}
        \caption{Eigenvalues of $\Phi_p^\top\Phi_p$}
        \label{fig:lalonde-ridge-undersmooth-panel-b}
    \end{subfigure}
    \caption{Single and double ridge regression applied to the ``high-dimensional'' version of the \citet{lalonde1986evaluating} data set. Panel (a) shows $\lambda_j$, the single ridge hyperparameter, and $\gamma_j$, the implied double ridge hyperparameter, as functions of the spectral direction $j$. Panel (b) shows the corresponding eigenvalues.}
    \label{fig:lalonde-ridge-undersmooth}
\end{figure}

\Cref{fig:lalonde-ridge-undersmooth-panel-a} plots the $\lambda_j$ and $\gamma_j$, with $j$ on the $x$-axis sorted in the order of the eigenvalues of $\Phi_p^\top \Phi_p$ from largest to smallest. \Cref{fig:lalonde-ridge-undersmooth-panel-b} plots the corresponding eigenvalues. 
For standard ridge regression, the key idea is to push the eigenvalues of $\Phi_p^\top\Phi_p$ away from zero so that the resulting matrix $\Phi_p^\top \Phi_p + \lambda I$ is invertible and well-conditioned. 
Thus, the original ridge regression applies the same amount of regularization across the spectrum, as shown by the blue line in \Cref{fig:lalonde-ridge-undersmooth-panel-a}.
By contrast, the implied outcome model from the augmented procedure uses far less regularization and is substantially undersmoothed. Importantly, the augmented estimator undersmooths more where the eigenvalues of $\Phi_p^\top\Phi_p$ are large and undersmooths less where the eigenvalues of $\Phi_p^\top\Phi_p$ are close to zero. The augmented estimator therefore avoids bias by only regularizing the spectral directions that are the most significant sources of variation --- and even these to a much smaller degree than is optimal for MSE predictions.

\subsubsection{Undersmoothing in norm}

\begin{figure}[h!]
    \centering
    \begin{subfigure}{0.4\textwidth}
        \centering 
        \includegraphics[width=\textwidth]{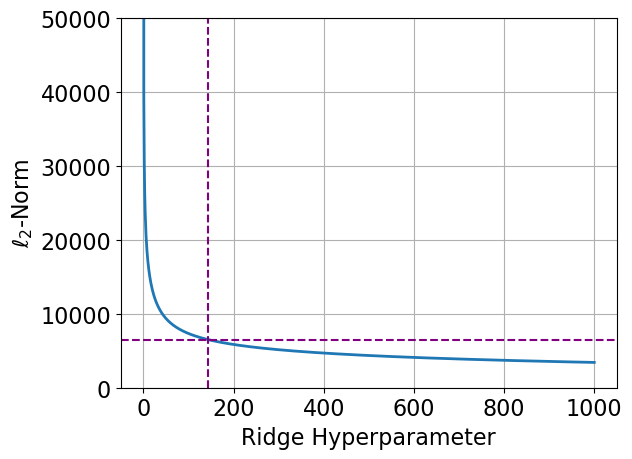}
        \caption{$\ell_2$ norm for ridge outcome model}
        \label{fig:lalonde-aug-norm-l2-panel-a}
    \end{subfigure}
        \begin{subfigure}{0.4\textwidth}
        \centering 
        \includegraphics[width=\textwidth]{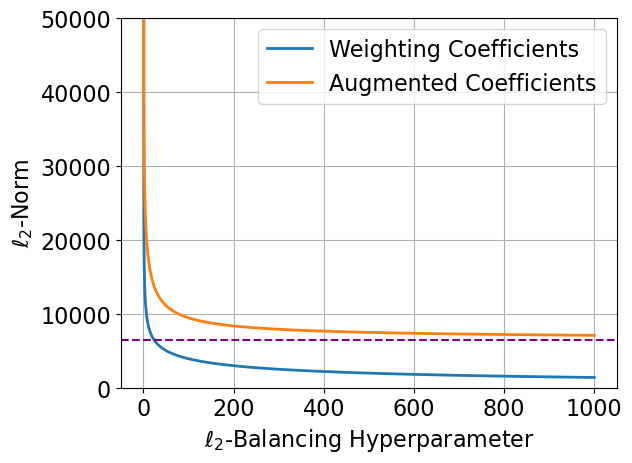}
        \caption{$\ell_2$ norm for weighting and double ridge}
        \label{fig:lalonde-aug-norm-l2-panel-b}
    \end{subfigure} \\[1em]
    \begin{subfigure}{0.4\textwidth}
        \centering 
        \includegraphics[width=\textwidth]{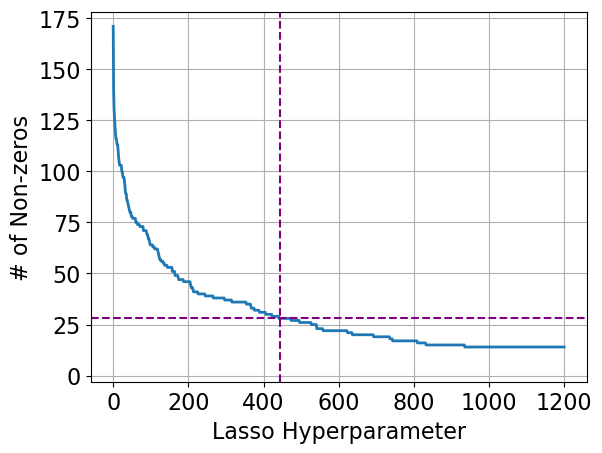}
        \caption{$\ell_0$ ``norm'' for lasso outcome model}
        \label{fig:lalonde-aug-norm-sparsity-panel-a}
    \end{subfigure}
        \begin{subfigure}{0.4\textwidth}
        \centering 
        \includegraphics[width=\textwidth]{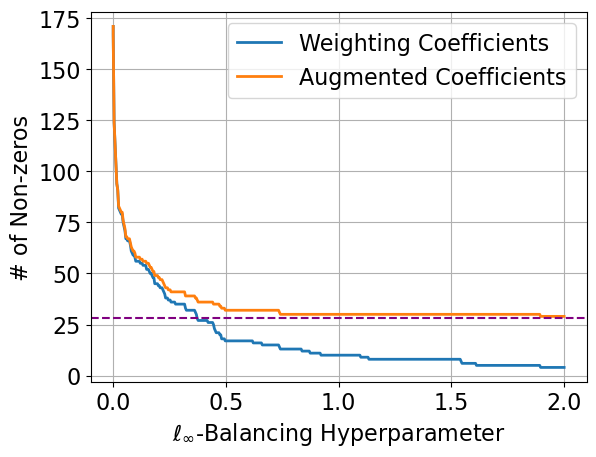}
        \caption{$\ell_0$ ``norm'' for weighting and double lasso}
        \label{fig:lalonde-aug-norm-sparsity-panel-b}
    \end{subfigure}
    \caption{Undersmoothing in norms for augmented balancing weights applied to the ``high dimensional'' version of the \citet{lalonde1986evaluating} data set. Panels (a) and (c) show the $\ell_2$ norm and the number of non-zero covariates (the $\ell_0$ ``norm'') for the ridge and lasso regression, respectively, of $Y_p$ on $\Phi_p$ among control units. Panels (b) and (d) show the $\ell_2$ and number of non-zero covariates for the corresponding unaugmented and augmented balancing weights estimators. The dotted purple line is the outcome model hyperparameter chosen via cross validation. The norms for the augmented estimators are always at least as large as for the outcome model alone.}
    \label{fig:lalonde-aug-norm-fig}
\end{figure}

For the special case with diagonal covariates, both double ridge and double lasso undersmooth in a particular norm of $\| \hat{\beta}_{\text{aug}} \|$, relative to the unaugmented outcome model. We demonstrate this phenomenon in \Cref{fig:lalonde-aug-norm-fig}, where we first use the eigenvectors of $\Phi_p^\top\Phi_p$ to decorrelate the features; as we discuss above, this is without loss of generality for $\ell_2$ balancing but not for $\ell_\infty$ balancing. In particular, \Cref{fig:lalonde-aug-norm-l2-panel-a,fig:lalonde-aug-norm-sparsity-panel-a} show the $\ell_2$ norm and the number of included covariates (the $\ell_0$ ``norm'') for the ridge and lasso outcome regression models, respectively, as a function of the outcome hyperparameter $\lambda$. The purple dotted lines show the values of $\lambda$ chosen via cross validation, and the corresponding values of the $\ell_2$ and $\ell_0$ norms. 
\Cref{fig:lalonde-aug-norm-l2-panel-b,fig:lalonde-aug-norm-sparsity-panel-b} show the corresponding norm for the unaugmented and balancing weights estimators, ``double ridge'' and ``double lasso,'' respectively. For both, the norms for the augmented estimators are always at least as large as the norms for the outcome model alone.
While the patterns are qualitatively the same, the behavior for $\ell_\infty$ balancing does not correspond to traditional undersmoothing, since the union of non-zero coefficients of the outcome and weighting models \emph{cannot} generally be recovered by changing the hyperparameter of the lasso outcome model.


\subsection{Semi-Synthetic Bias vs Variance}\label{apx:semisynth-biasvsvar}

For the Lalonde long data, we use the same semi-synthetic strategy as in \Cref{apx:numerical-experiment-details}, and then plot the analytical design conditional bias and variance in \Cref{fig:bias-var-plot}. Note that the variance dominates, and there's an especially large non-linearity when $\delta$ is close to zero.

\begin{figure}[htb!]
    \centering
        \centering 
        \includegraphics[width=0.5\textwidth]{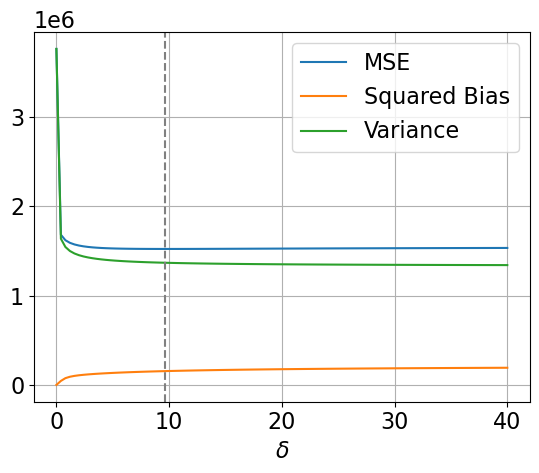}
    \caption{...}
    \label{fig:bias-var-plot}
\end{figure}

\section{Empirical Bayes Interpretation of Augmented $\ell_2$ Balancing Weights}
\label{apx:empirical_bayes}
This section discusses possible Empirical Bayes interpretations of augmented $\ell_2$ balancing weights. We begin with the simple case with fixed outcome model coefficients, before exploring models that incorporate the uncertainty in these coefficients. While plausible, we do not find this interpretation particularly compelling, and instead point interested readers to recent work on weighting and Bayesian double robustness, especially \citet{murray2024bayes_weights}.

\subsection{Fixed outcome coefficients}
We initially consider the straightforward special case where the outcome model coefficients $\hat\beta_{\text{reg}}$ are considered fixed, for example if obtained from an entirely different sample. 

To fix ideas, consider a basic Bayesian linear model for a single data point $i$ with covariates $x \in \mathbb{R}^J$ and outcome $y \in \mathbb{R}$:
$$
y_i = \beta^\top x_i + \varepsilon_i, \quad \varepsilon_i \sim N(0, 1).
$$
\noindent As in the main text, we assume that the covariance matrix is diagonal with entries $\sigma^2_1, \ldots, \sigma^2_J$.

We briefly review the Bayesian interpretation of ridge regression as the posterior mean of $\beta$ corresponding to the following prior:
$$
    \beta_j \mid \tau^2 \sim N(0, \tau^2),
$$
with hyperparameter $\tau^2$. The posterior for $\beta_j$ is conditionally Normal:
\begin{align*}
    \beta_j \mid \tau^2 &\sim N \left( \frac{ \tau^2  }{ \tau^2 + \frac{1}{\sigma^2_j} } \hat{\beta}_{ols,j},
     \frac{ 1 }{ \frac{1}{ \tau^2} + \sigma^2_j } \right) 
\end{align*}
where the posterior mean equals:
$$
\left(\frac{\sigma^2_j}{\sigma^2_j + \frac{1}{\tau^2}}\right) \hat{\beta}_{ols,j} + \left(\frac{\frac{1}{\tau^2}}{\sigma^2_j + \frac{1}{\tau^2}}\right) 0 = \left(\frac{\sigma^2_j}{\sigma^2_j + \frac{1}{\tau^2}}\right) \hat{\beta}_{ols,j}.
$$

To recover the augmented $\ell_2$ balancing weights form, we simply center the prior for $\beta_j$ at the fixed $\hat{\beta}_{\text{reg},j}$:
\begin{align*}
    \beta_j \mid \tau^2 &\sim N(\hat{\beta}_{reg, j}, \tau^2) 
\end{align*}
The posterior mean is then equal to:
$$
\left(\frac{\sigma^2_j}{\sigma^2_j + \frac{1}{\tau^2}}\right) \hat{\beta}_{ols,j} + \left(\frac{\frac{1}{\tau^2}}{\sigma^2_j + \frac{1}{\tau^2}}\right) \hat{\beta}_{reg, j} 
$$
Setting $\tau^2 = \frac{1}{\delta}$ recovers the augmented $\ell_2$ result from the main text.

\subsection{Accounting for uncertainty in outcome model coefficients via global-local shrinkage priors}

We now instead consider an Empirical Bayes interpretation of the ``double ridge'' form:
$$
\left(\frac{ \sigma^2_j}{ \sigma^2_j + \frac{\delta \lambda}{\sigma^2_j + \lambda + \delta}}\right)  \hat{\beta}_{ols, j}.
$$
Note that we cannot simply replace $ \hat{\beta}_{reg, j} $ in the derivation above with $\left(\frac{ \sigma^2_j }{ \sigma^2_j + \lambda}\right) \hat{\beta}_{ols, j}$ since this ignores the uncertainty in estimating these coefficients.

\subsubsection{Scale mixture of Normals}
Again consider a general Bayesian linear model for a single data point $i$:
$$
y_i = \beta^\top x_i + \varepsilon_i, \quad \varepsilon_i \sim N(0, 1)
$$
with the following general class of \emph{global-local shrinkage priors} \citep{polson2010global_local}: 
\begin{align*}
    \beta_j \mid \xi_j, \tau, \sigma_j &\sim N(0, \tau^2 \xi^2_j) \\
    \xi_j &\sim p(\xi) \\
    \tau &\sim p(\tau),
\end{align*}
where $\tau$ is called the \emph{global} parameter, and $\xi_j$ is called the \emph{local} parameter.
To simplify exposition, we again assume that the covariates are uncorrelated with zero mean and variances $\sigma^2_1, \ldots, \sigma^2_J$ (e.g., the principal components).  As above, the posterior for $\beta_j$ is conditionally Normal:
\begin{align*}
    \beta_j \mid \tau^2, \xi_j, \sigma_j &\sim N \left( \frac{ \tau^2 \xi^2_j }{ \tau^2 \xi^2_j + \frac{1}{\sigma^2_j} } \hat{\beta}_j,
     \frac{ 1 }{ \frac{1}{\tau^2 \xi^2_j} + \sigma^2_j } \right),
\end{align*}
where $\hat{\beta}$ is the usual OLS estimate.

Let $\tilde{\beta}_j$ be the posterior mean. We can write this in the following canonical shrinkage form:
$$
    \tilde{\beta}_j = (1 - \kappa_j) \hat{\beta}_j,
$$
where
$$
\kappa_j = \frac{1}{1 + \tau^2 \xi^2_j \sigma^2_j }.
$$
Here $\kappa_j = 0$ implies no shrinkage and $\kappa_j = 1$ implies complete shrinkage. 
\citet{piironen2017sparsity} show that all regression priors that are scale-mixture of Normals can be written in this form.  Some prominent examples include:
\begin{itemize}
    \item \textbf{Ridge.} Standard ridge regression sets $\xi^2_j = 1$.
    \item \textbf{Horseshoe prior.} Here $\xi_j \sim \text{C}^+(0,1)$, where $C^+$ denotes the half-Cauchy density \citep{carvalho2010horseshoe}.
    \item \textbf{Spike-and-slab priors.} When the spike is a point mass at 0, the resulting prior also has this shrinkage form.
    \item \textbf{Zellner's $g$ prior.} While not a proper prior, Zellner's $g$ prior has the same canonical form, with $\xi^2_j = \frac{1}{\sigma^2_j}$ in the case of a diagonal covariance matrix, with corresponding $\kappa_j = \frac{1}{1 + \tau^2}$.
\end{itemize}

\subsubsection{Inverse-Gamma Gamma Prior}
We focus on a flexible hierarchical prior known as the \emph{Inverse Gamma-Gamma} prior \citep{bai2017inverse} or, equivalently, \emph{Beta-prime} prior \citep{armagan2011generalized}: 
\begin{align*}
    \beta_j \mid \xi_j, \tau &\sim N(0, \tau^2 \xi^2_j) \\
    \xi^2_j &\sim \text{Gamma}(a, 1) \\
    \frac{1}{\tau^2} &\sim \text{Gamma}(b, 1)
\end{align*}
\noindent This formulation guarantees that the product $\xi^2 \tau^2 \sim \text{Beta-prime}(a, b)$, also known as the inverse-Beta distribution. For $\sigma^2_j = 1$, we then have $(1-\kappa_j) \equiv \frac{\xi^2 \tau^2}{1 + \xi^2 \tau^2} \sim \text{Beta}(a,b)$, and by symmetry, $\kappa_j \sim \text{Beta}(b,a)$; with the expectation of $\kappa_j = \frac{b}{a+b} = \frac{1}{1 + \frac{a}{b}}$.
The horseshoe is a special case with $a = b = 1/2$.
Incorporating $\sigma^2_j$, the implied prior distribution for the shrinkage factor is $\kappa_j \sim \text{Beta}(b, a \sigma^2_j)$, with expectation $\frac{b}{a\sigma^2_j + b} = \frac{1}{1 + \frac{a\sigma^2_j}{b }}$.

\subsubsection{Application to double ridge}
With a bit of algebra, we can write the shrinkage factor for double ridge as:
$$
\kappa_j^\ast = \frac{1}{1 + \frac{\sigma^2_j(\sigma^2_j + \lambda + \delta)}{\delta \lambda}}
$$
This corresponds to an Inverse Gamma-Gamma prior with $a = \sigma^2_j + \lambda + \delta$ and $b = \delta\lambda$.
\begin{align*}
    \beta_j \mid \xi_j, \tau &\sim N(0, \tau^2 \xi^2_j) \\
    \xi^2_j \mid \delta, \lambda, \sigma^2_j &\sim \text{Gamma}(\sigma^2_j + \lambda + \delta, 1) \\
    \frac{1}{\tau^2} \mid \delta, \lambda &\sim \text{Gamma}(\delta\lambda, 1),
\end{align*}
where the implied prior for $\kappa^\ast_j$ is:
$$
\kappa^\ast_j \mid \delta, \lambda, \sigma^2_j \sim \text{Beta}(\delta \lambda, \sigma^2_j(\sigma^2_j + \delta + \lambda))
$$

As expected, $\delta$ and $\lambda$ have symmetric roles --- the hyperprior distributions only depend on their product and sum. We can assess the behavior of this prior in two extreme cases. First, when 
$\delta \to \infty$ or $\lambda \to \infty$, the prior on $\kappa_j^\ast$ tends to $\text{Beta}(b^\ast,a^\ast)$ with $a^\ast \to \infty, b^\ast \to \infty$. The corresponding expectation of $\kappa_j^\ast$ tends to 1, which implies full shrinkage, $\tilde{\beta}_j = 0$. Second, when $\delta \to 0$ or $\lambda \to 0$, the prior on $\kappa_j^\ast$ tends to $\text{Beta}(b^\ast,a^\ast)$ with $a^\ast \to \sigma^4_j, b^\ast \to 0$. The corresponding expectation of $\kappa^\ast_j$ tends to 0, or no shrinkage, $\tilde{\beta}_j = \hat{\beta}_j$.

As Figure \ref{fig:double_ridge_implied_prior} shows, the shrinkage depends strongly on the eigenvalue $\sigma^2_j$ and the hyperparameters $\delta$ and $\lambda$. Holding $\delta$ and $\lambda$ fixed, $\kappa_j$ is decreasing in $\sigma^2_j$ --- i.e., there is less shrinkage for larger $\sigma^2_j$. Holding $\sigma^2_j$ fixed, $\kappa_j$ is increasing $\delta$ and $\lambda$ --- i.e., as the hyperparameters increase, both the outcome and weighting models move closer to uniform weights and greater shrinkage. 

\begin{figure}[tb]
    \centering
    \begin{subfigure}{0.5\textwidth}
    \includegraphics[width=1\linewidth]{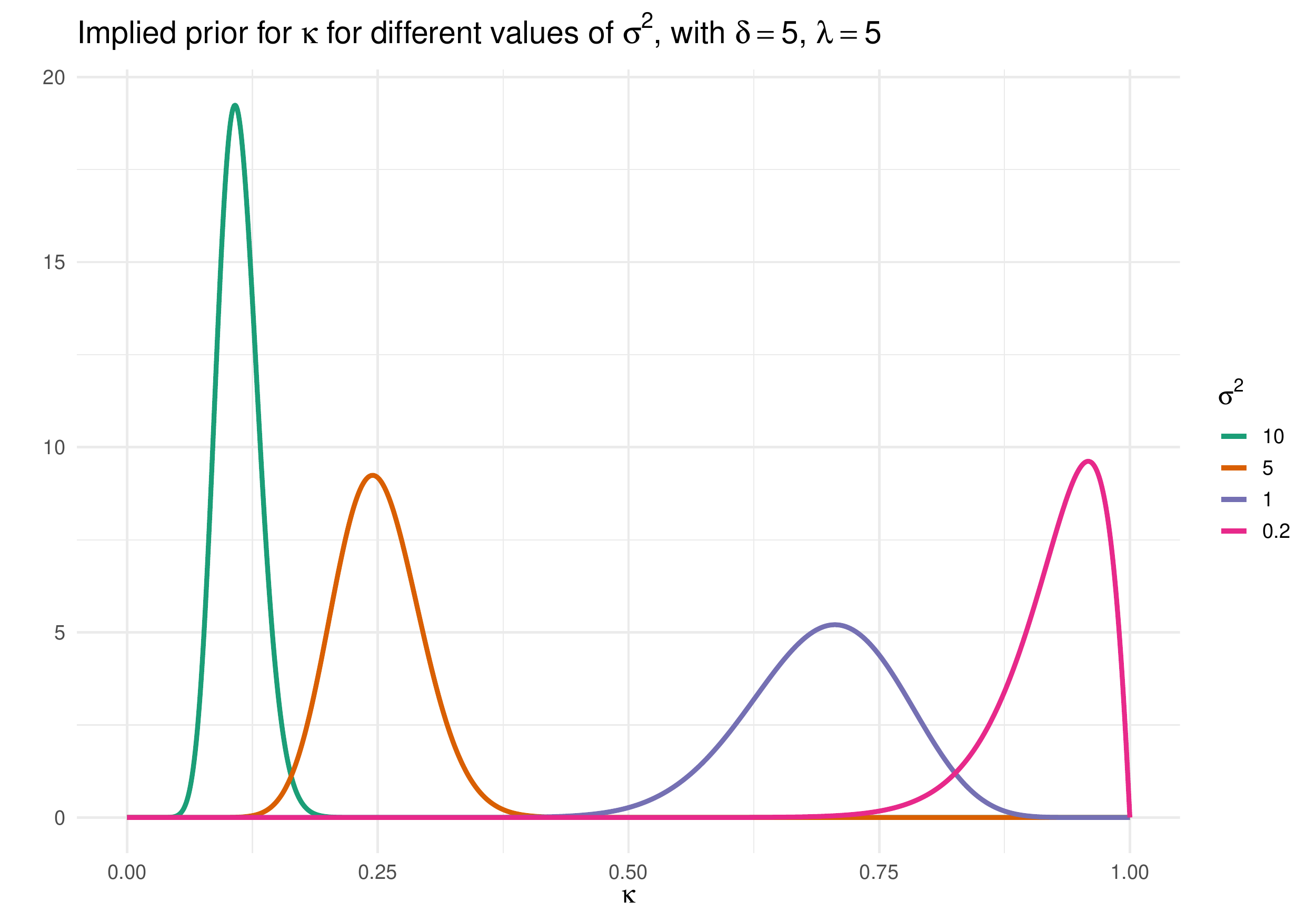}
    \end{subfigure}%
    ~\begin{subfigure}{0.5\textwidth}
    \includegraphics[width=1\linewidth]{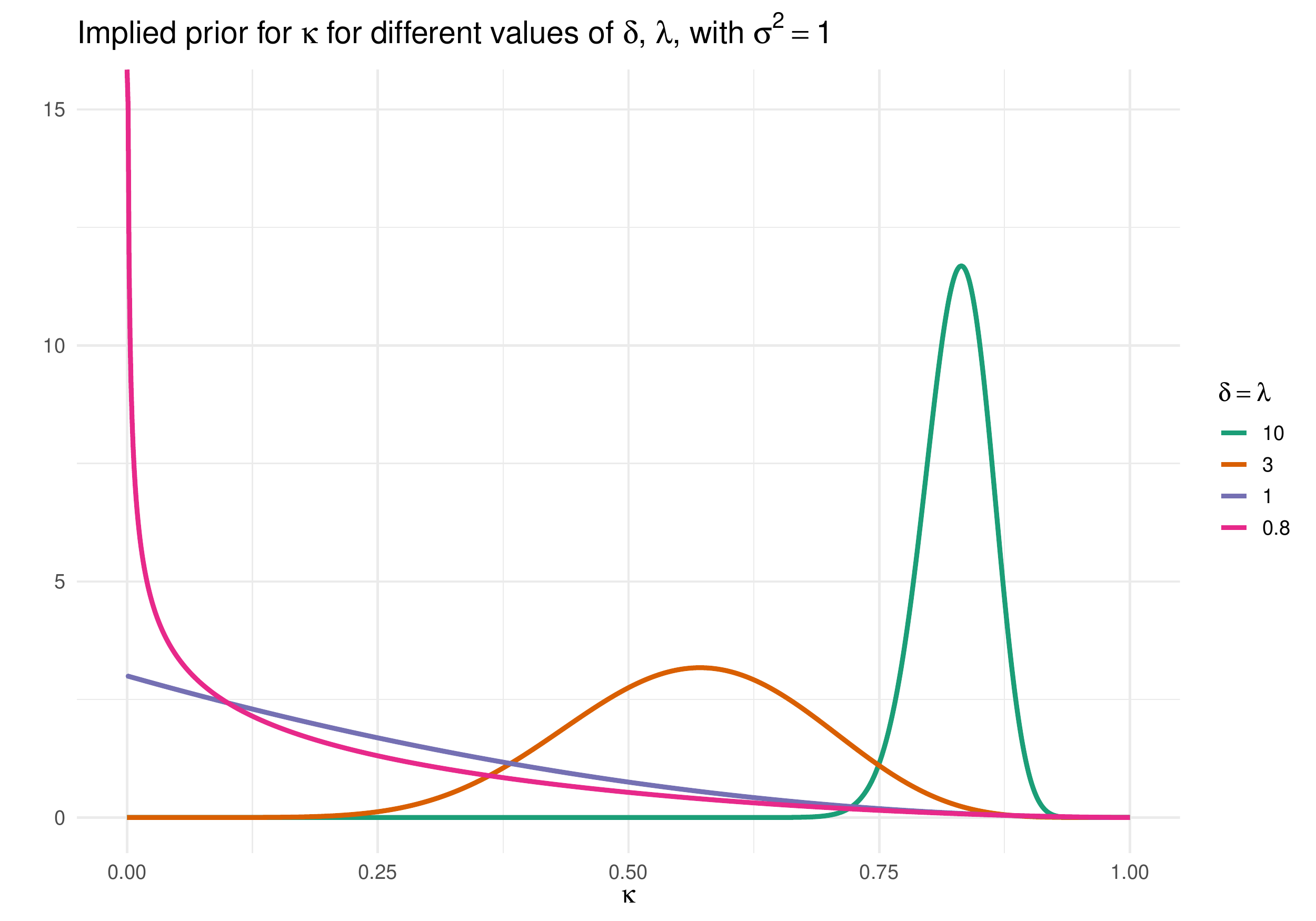}
    \end{subfigure}%
    \caption{Implied prior on $\kappa$ for: (left) varying $\sigma^2_j$, with $\delta = 5$, $\lambda = 5$; (right) varying $\delta = \lambda$, with $\sigma^2_j = 1$.}
    \label{fig:double_ridge_implied_prior}
\end{figure}

\paragraph{Connection to Student $t$ shrinkage.} 
If we drop the hyperprior on $\xi^2$ and set $\xi^2_j = 1$, this reduces to:
\begin{align*}
    \beta_j \mid\tau &\sim N(0, \tau^2 ) \\
    \frac{1}{\tau^2} &\sim \text{Gamma}(\delta\lambda, 1)
\end{align*}
This yields a student $t$ posterior for $\beta_j$.

\paragraph{Connection to Normal-Gamma prior.}
If we instead drop the hyperprior on $\tau^2$ and set $\tau^2 = 1$, this reduces to:
\begin{align*}
    \beta_j \mid \xi_j, \tau &\sim N(0, \xi^2_j) \\
    \xi^2_j &\sim \text{Gamma}(\sigma^2_j + \lambda + \delta, 1)
\end{align*}
which is a Normal-Gamma prior, where the hyperparameter depends on $\sigma^2_j$. Here the prior mean for $\xi^2_j$ is simply $\sigma^2_j + \delta + \lambda$, which is increasing in $\xi^2_j$.

\subsection{Connection to Zellner's $g$ prior}
Finally, we can connect this approach to  generalization of Zellner's $g$ prior known as \emph{adaptive powered correlation prior} \citep{krishna2009bayesian}, which has the form:
\begin{align*}
    \beta_j \mid g &\sim N \left(0, \frac{1}{g} (X'X)^\nu \right)
\end{align*}
with additional hyperparameter $\nu$. 
Zellner's g prior corresponds to $\nu = -1$, and standard ridge regression corresponds to $\nu = 0$. We are interested in the ``inverted $g$ prior'' with $\nu = +1$.
Since we focus on diagonal covariance, we can re-write the prior as:
\begin{align*}
    \beta_j \mid g &\sim N \left(0, \frac{1}{g} \sigma_j^{2\nu} \right)
\end{align*}

Then the posterior mean is:
\begin{align*}
   \tilde{\beta} &= \left(\frac{\frac{1}{g\sigma_j^{2\nu}}}{\frac{1}{g\sigma_j^{2\nu}} + \frac{1}{\sigma^2_j}} \right)\hat{\beta}_j\\
    &= \left(\frac{\sigma^{(2\nu+2)}_j}{\sigma^{(2\nu+2)}_j + g}\right)\hat{\beta}_j
\end{align*}

Special cases include:
\begin{itemize}
    \item \textbf{Zellner's $g$ prior:} ($\nu = -1$)
    $$
    \tilde{\beta} = 
    \left(\frac{1}{{1 + g}}\right)\hat{\beta}_j$$

    \item \textbf{Ridge regression:} ($\nu = 0$)
    $$
    \tilde{\beta} = 
    \left(\frac{\sigma^2_j}{{\sigma^2_j + g}}\right)\hat{\beta}_j$$

    \item \textbf{Inverted $g$ prior:} ($\nu = +1$)
    $$
    \tilde{\beta} = 
    \left(\frac{\sigma^4_j}{{\sigma^4_j + g}}\right)\hat{\beta}_j$$
\end{itemize}

We can view the double ridge as a regularized form of this inverted $g$ prior, with $\nu = +1$ and $g = \delta \lambda$.

\clearpage
\section{Additional Proofs}
\label{sec:proofs}

\begin{proof}[Closed forms for $\ell_2$ and exact balancing weights.]
We derive the closed form for $\ell_2$ balancing weights with parameter $\delta$ (with exact balance following as a special case). The optimization problem:
\begin{align*}
    \min_{w\in\mathbb{R}^n} \Vert w^\top \Phi_p - \bar{\Phi}_q \Vert_2^2 + \delta \Vert w \Vert_2^2 = w^\top \Phi_p \Phi_p^\top w - 2 w^\top \Phi_p \bar{\Phi}_q + \delta w^\top w
\end{align*} 
has first order condition:
\[ 2 (\Phi_p\Phi_p^\top + \delta I_n) w - 2 \Phi_p \bar{\Phi}_q = 0,  \]
which gives the solution:
\begin{align*}
    w^* &= (\Phi_p\Phi_p^\top + \delta I_n)^\dag \Phi_p \bar{\Phi}_q\\
    &= \Phi_p (\Phi_p^\top\Phi_p + \delta I_d)^\dag \bar{\Phi}_q.
\end{align*}
\end{proof}

\begin{proof}[Proof of \Cref{l2augment}]
    We apply \Cref{generalregularizationpath}. We have $a_j^\delta = \hat{\Phi}_{q,j}^\delta / \bar{\Phi}_{q,j}^\delta$. Then:
    \begin{align*}
        \hat{\Phi}_q^\delta &= \hat{w}_{\ell_2}^\delta \Phi_p = \overbar{\Phi}_q(\Phi_p^\top\Phi_p + \delta I)^{-1} \Phi_p^\top\Phi_p.
    \end{align*}
    Since we have assumed that $\Phi_p^\top\Phi_p$ is diagonal, with $j$th diagonal entry, $\sigma_j^2$, we have:
    \[ \hat{\Phi}_{q,j}^\delta = \left(\frac{\sigma_j^2}{\sigma_j^2 + \delta}\right) \bar{\Phi}_{q,j}.  \]
    Plugging this back into $a_j^\delta$ completes the proof. 
\end{proof}
\newpage
\begin{proof}[Proof of \Cref{prop:double_ridge}]
Applying Proposition \Cref{l2augment}:
\begin{align*}
    \hat{\beta}_{\text{$\ell_2$},j} &= \left(\frac{\sigma^2_j}{\sigma^2_j + \delta}\right) \hat{\beta}_\text{ols,j} + \left(\frac{\delta}{\sigma^2_j + \delta}\right) \hat{\beta}_\text{ridge,j}^\lambda\\
    &= \left(\frac{\sigma^2_j}{\sigma^2_j + \delta}\right) \hat{\beta}_\text{ols,j} + \left(\frac{\delta}{\sigma^2_j + \delta}\right) \left(\frac{\sigma^2_j}{\sigma^2_j + \lambda}\right)\hat{\beta}_{\text{ols} , j }\\
    &= \frac{\sigma^2_j (\sigma_j^2 + \lambda + \delta)}{(\sigma^2_j + \delta)(\sigma^2_j + \lambda)}\hat{\beta}_\text{ols,j} \\
\end{align*}
Then taking:
\[\frac{\sigma^2_j (\sigma_j^2 + \lambda + \delta)}{(\sigma^2_j + \delta)(\sigma^2_j + \lambda)} = \frac{\sigma_j^2}{\sigma_j^2 + \gamma_j} \]
and solving for $\gamma_j$ gives:
\[ \gamma_j \coloneqq \frac{\delta \lambda}{\sigma_j^2 + \lambda + \delta}\]
which completes the proof. 
\end{proof}

\begin{proof}[Proof of \Cref{orthonormalellinfty}]
    We begin with the constrained form of the balancing problem:
    \begin{align*}
    \min_{w \in \mathbb{R}^n} & \Vert w \Vert_2^2\\
    \text{such that } & \left\Vert w \Phi_p - \bar{\Phi}_q \right\Vert_\infty \leq \delta.
\end{align*}
Note that we can rewrite the norm constraint as two vector-valued linear constraints:
\begin{align*}
    w \Phi_p &\preceq \bar{\Phi}_q + \delta\\
    - w \Phi_p  &\preceq - \bar{\Phi}_q + \delta,
\end{align*}
which results in the Lagrangian,
\begin{align*}
    \mathcal{L}(w,\mu,\nu) = \Vert w \Vert_2^2 + \mu^\top\left(w \Phi_p - \bar{\Phi}_q - \delta\right) - \nu^\top\left(w \Phi_p -\bar{\Phi}_q + \delta\right).
\end{align*}
The first-order conditions for the optimal $w^*, \mu^*, \nu^*$ are:
\begin{align*}
    w^* = -\frac{1}{2} \left( \Phi_p \mu^* - \Phi_p \nu^* \right)\\
    \mu_j^* \left(w^* \Phi_{p,j} - \bar{\Phi}_{q,j} - \delta\right) = 0, \forall j\\
    \nu_j^* \left(w^* \Phi_{p,j} - \bar{\Phi}_{q,j} + \delta\right) = 0, \forall j\\
    \mu_j^* , \nu_j^* \geq 0, \forall j
\end{align*}
plus the linear constraints on $w^* \Phi_p$. Note that the linear constraints plus the complimentary slackness conditions imply that one of three mutually-exclusive cases holds for each covariate. Case 1: $w^* \Phi_{p,j} = \bar{\Phi}_{q,j} - \delta$, in which case $\mu^*_j = 0$. Case 2: $w^* \Phi_{p,j} = \bar{\Phi}_{q,j} + \delta$, in which case $\nu^*_j = 0$. Or Case 3: $w^* \Phi_{p,j} \in (\bar{\Phi}_{q,j} - \delta, \bar{\Phi}_{q,j} + \delta)$, in which case $\mu^*_j = \nu^*_j = 0$. 

Define:
\begin{align*}
    \theta^*_j \coloneqq \begin{cases}
        0 & \text{if } w^* \Phi_{p,j} \in (\bar{\Phi}_{q,j} - \delta, \bar{\Phi}_{q,j} + \delta)\\
        -\mu_j^*/2 & \text{if } w^* \Phi_{p,j} = \bar{\Phi}_{q,j} + \delta\\
        \nu_j^*/2 & \text{if } w^* \Phi_{p,j} = \bar{\Phi}_{q,j} - \delta.
    \end{cases}
\end{align*}
Then we have $w^* = \Phi_p \theta^*$ from the first-order condition, and thus $w^* \Phi_p = (\Phi_p^\top \Phi_p) \theta^*$. Using the fact that $(\Phi_p^\top \Phi_p)$ is diagonal, we get $w^* \Phi_{p,j} = \sigma_j^2 \theta^*_j$.

Finally, we can plug this into the three cases that define $\theta^*$. First, $\theta_j^* = 0$ when
\begin{align*}
    &\sigma^2_j \theta_j^* \in (\bar{\Phi}_q - \delta, \bar{\Phi}_q + \delta)\\
    \implies& 0 \in (\bar{\Phi}_q - \delta, \bar{\Phi}_q + \delta)\\
    \implies& \bar{\Phi}_q \in (-\delta,\delta).
\end{align*}
Second, $\theta_j^* = -\mu_j^*/2$ when $\sigma_j^2 \theta_j^* = \bar{\Phi}_{q,j} + \delta$, which implies $\mu_j^* = -2(\bar{\Phi}_{q,j} + \delta)/ \sigma_j^2$. We then apply the dual variable constraint:
\begin{align*}
    &\mu_j^* \geq 0 \\
    \implies& -2(\bar{\Phi}_{q,j} + \delta)/ \sigma_j^2 \geq 0 \\
    \implies& \bar{\Phi}_{q,j}  \leq - \delta.
\end{align*}
Third, $\theta_j^* = \nu_j^*/2$ when  $\sigma_j^2 \theta_j^* = \bar{\Phi}_{q,j} - \delta$
, which implies $\nu_j^* = 2(\bar{\Phi}_{q,j} - \delta)/ \sigma_j^2$. We then apply the dual variable constraint:
\begin{align*}
    &\nu_j^* \geq 0 \\
    \implies& 2(\bar{\Phi}_{q,j} - \delta)/ \sigma_j^2 \geq 0 \\
    \implies& \bar{\Phi}_{q,j}  \geq  \delta.
\end{align*}
Putting the cases together we get:
\begin{align*}
    \theta^*_j \coloneqq \begin{cases}
        0 & \text{if }\bar{\Phi}_q \in (-\delta,\delta)\\
        (\bar{\Phi}_{q,j} + \delta)/\sigma_j^2 & \text{if } \bar{\Phi}_{q,j}  \leq - \delta.\\
        (\bar{\Phi}_{q,j} - \delta)/\sigma_j^2 & \text{if } \bar{\Phi}_{q,j}  \geq  \delta.
    \end{cases}
\end{align*}
This is exactly the soft-thresholding operator, which completes the proof. 
\end{proof}

\begin{proof}[Proof of \Cref{linfaugment}]
To obtain this result from the general form of $a^\delta_j = \widehat{\Delta}_j/\Delta_j$ in  \Cref{generalregularizationpath}, notice that the implied feature shift, $\widehat{\Delta}_j = \hat{\Phi}_{q,j}^\delta - \overbar{\Phi}_{p,j} =  \mathcal{T}_\delta(\overbar{\Phi}_{q,j} - \overbar{\Phi}_{p,j})$ is:
   \[   
   \widehat{\Delta}_j  =
\begin{cases}
        0 & \text{if } |\Delta_j| < \delta \\[0.5em]
        \Delta_j - \delta   & \text{if } \Delta_j > \delta \\
       \Delta_j + \delta & \text{if } \Delta_j < -\delta
    \end{cases}.
 \]
 Thus, for instance, $\frac{\widehat{\Delta}_j}{\Delta} =  \frac{\Delta_j - \delta}{\Delta_j} = 1 - \frac{\delta}{\Delta_j}$ when $\Delta_j > \delta$.
 \end{proof}

\begin{proof}[Proof of \Cref{prop:double-selection}]
    The result follows immediately from \Cref{linfaugment}.
\end{proof}

\begin{proof}[Proof of \Cref{prop:asymptotic-sq-rate}]
    Rewriting the definition of $\gamma_n$ with $\lambda_n = \delta_n$, we have
    \begin{align*}
        \gamma_n = \frac{\lambda_n^2}{\sigma^2 + 2\lambda_n^2} = \frac{1}{\sigma^2 \lambda_n^{-2} + 2 \lambda_n^{-1} }.
    \end{align*}

    Because $\sigma^2 x^2 + 2 x = O(x^2)$ as a function of $x$, and because $\lambda_n^{-1}$ is monotonically increasing, $\sigma^2 \lambda_n^{-2} + 2 \lambda_n^{-1} = O(\lambda_n^{-2})$. And $\lambda_n^{-2} = O(\sigma^2 \lambda_n^{-2} + 2 \lambda_n^{-1})$ because $\sigma^2 \geq 0$ and $\lambda_n^{-1} > 0$. Thus $\sigma^2 \lambda_n^{-2} + 2 \lambda_n^{-1} \asymp \lambda_n^{-2}$.
    
    Finally, note that for any two functions of $n$, $f_n$ and $g_n$, 
\[ f_n \asymp g_n \iff f_n^{-1} \asymp g_n^{-1},\] 
    and therefore,
    \[ \gamma_n \asymp \lambda_n^2. \]
\end{proof}

\section{Additional Details for Asymptotic Results}
\label{sec:asymptotics_appendix}

Our setup for the RKHS follows \citet{singh2021debiased}. First, assume that the space $\mathcal{X} \times \mathcal{Z}$ is Polish. Let $\mathcal{H}$ be an RKHS on $\mathcal{X} \times \mathcal{Z}$ with corresponding kernel $k$ satisfying standard regularity conditions \citep[][Assumption 5.2]{singh2021debiased} and let $\eta_j$, $\varphi_j$ denote the eigenvalues and eigenfunctions respectively of its kernel integral operator under $p$. 
Next, assume that the eigenvalues satisfy the decay condition $\eta_j \leq C j^{-b}$ for some $b > 1$ and a constant $C$. The parameter $b$ encodes information on the effective dimension of $\mathcal{H}$.
For a bounded kernel, $b>1$ \citep{fischer2020sobolev}: the case where $b=\infty$ corresponds to a finite-dimensional RKHS; for the case with $1<b<\infty$, the $\eta_j$ must decay at a polynomial rate. 

We then assume that for some $c \in [1,2]$, the outcome function $m(x,z)$ belongs to the set:
\begin{align}
\mathcal{H}^c\coloneqq\left\{f=\sum^\infty_{j=1}a_j\varphi_j:\sum^\infty_{j=1}\frac{a_j^2}{\eta^{c}_j}<\infty\right\}\subset \mathcal{H}, \label{eq:smooth-rkhs}
\end{align}
where $c$ encodes additional \emph{smoothness} of the conditional expectation. If $c=1$, then by the spectral decomposition of the RKHS, \Cref{eq:smooth-rkhs} is equivalent to requiring $m \in \mathcal{H}$; choosing larger values of $c$ corresponds to $m$ being a smoother element of $\mathcal{H}$, with a ``saturation effect'' kicking in for $c>2$ \citep{bauer2007regularization}. Varying $b$ (the effective dimension of the RKHS) and $c$ (the additional smoothness of the outcome function) changes the optimal rates for regression, with larger values of both corresponding to faster rates of convergence. 

Finally, we assume that the Riesz representer, $\alpha(x,z)$, of our linear functional estimand also belongs to $\mathcal{H}^c$. Under these conditions, \cite{singh2021debiased} demonstrates that an augmented estimator combining kernel balancing weights and a kernel ridge regression base learner is root-$n$ consistent and asymptotically normal. 

Following \cite{caponnetto2007optimal}, Theorems 5.1 and 5.2 of \cite{singh2021debiased} use hyperparameter schedules for $\lambda$ and $\delta$, which depend on the effective dimension $b$ and smoothness $c$:
\begin{align*}
\lambda_n = \delta_n  = \begin{cases}
 n^{-1/2} &\textrm{ if } b=\infty\\
n^{-\frac{b}{bc+1}} &\textrm{ if } b\in (1,\infty), \quad c\in(1,2]\\
(n/\log(n))^{-b/(b+1)} &\textrm{ if } b\in (1,\infty),\quad c=1
\end{cases},
\end{align*}
We can compute the implied augmented hyperparameter sequence $\gamma_n$ using the following proposition.
\begin{proposition}\label{prop:asymptotic-sq-rate}
Let $\lambda_n > 0$ be any monotonically decreasing function of $n$ and let $\delta_n = \lambda_n$. Then:
\[ \gamma_n \coloneqq \frac{\lambda_n \delta_n}{\sigma^2 + \lambda_n + \delta_n} \asymp \lambda_n^2. \] 
\end{proposition}
The standard ridge regression case corresponds to the finite-dimensional setting with $b=\infty$. When $c > 1$, the optimal rate for $\lambda_n$ is $n^{-\frac{b}{bc+1}}$;  the implied hyperparameter is then order $n^{-2b/bc+1} \in (n^{-2} , n^{-2/3})$ for $c \in (1,2]$ and $b \in (1, \infty)$. Whether or not this smooths more than $n^{-1}$ therefore depends on the relationship between the effective dimension $b$ and the smoothness $c$. In particular, the implied hyperparameter goes to zero at a slower rate than $n^{-1}$ whenever $c \geq 2 - \frac{1}{b}$. 
It is unclear whether the rates we find here are the only undersmoothed rates that will yield efficiency for fixed $b$ and $c$; we leave a thorough investigation to future work.

\end{document}